\documentclass[
    onecolumn,aps,pra,
    superscriptaddress,
    nofootinbib,notitlepage
]{revtex4-2}

\usepackage{graphicx} % Required for inserting images
\usepackage[]{xcolor}

\usepackage{amssymb}
\usepackage{amsmath}
\usepackage{amsthm}
\usepackage{mathtools}

\usepackage{bm}
\usepackage{bbold} % \mathbb{1}
\usepackage{siunitx}
\usepackage{physics}
\usepackage[version=4]{mhchem}
\usepackage{enumitem}
\usepackage{tabularx}

\usepackage{qcircuit}

\usepackage{algpseudocode}
\usepackage[ruled,vlined]{algorithm2e}
\RestyleAlgo{ruled}

\makeatletter
\renewcommand{\p@subsection}{}
\makeatother

\definecolor{LEI-blue}{cmyk}{1,.75,0,.35} % Leiden blue
\definecolor{LEI-orange}{cmyk}{0,.62,.97,0} % Leiden Science orange
\definecolor{niceblue}{rgb}{.1, .25, .8}

\usepackage[colorlinks=true, linkcolor=niceblue, citecolor=gray]{hyperref}

\newcommand{\topic}[1]{}

\newtheorem{theorem}{Theorem}
\newtheorem{definition}[theorem]{Definition}
\newtheorem{lemma}[theorem]{Lemma}

\newtheorem{corollary}[theorem]{Corollary}

\SetAlgorithmName{Alg.}{Alg.}{List of Algorithms}

\DeclareMathOperator*{\argmax}{arg\,max}
\DeclareMathOperator*{\argmin}{arg\,min}
\DeclareMathOperator{\Var}{Var}

\DeclareMathOperator{\E}{\mathbb{E}}

\newcommand{\NN}{\mathbb{N}}
\newcommand{\RR}{\mathbb{R}}

\newcommand{\google}{%
    \affiliation{Google Quantum AI, Munich, Germany}
}

\newcommand{\aqa}{\affiliation{%
    $\langle a Q a^L \rangle$ \& 
    Instituut-Lorentz,
    Universiteit Leiden, the Netherlands}
}
\newcommand{\qusoft}{\affiliation{%
    QuSoft \& CWI, Amsterdam, the Netherlands}
}
\newcommand{\uva}{\affiliation{%
    QuSoft, HIMS \& IvI, University of Amsterdam, the Netherlands}
}

\begin{document}

\title{Accurate ground state energy estimation with noise and imperfect state preparation 
}
\author{Alicja Dutkiewicz}
\qusoft
\aqa

\author{Thomas E. O'Brien}
\google
\aqa
%\lorentz

\author{Stefano Polla}
\aqa
\uva
%\lorentz

\begin{abstract}
    We introduce a classical estimator for the post-processing of quantum phase estimation (QPE) data when a single target phase is isolated within a known interval, as is typical of ground state energy estimation of gapped systems.
    Our estimator filters the QPE signal within this promise region and recovers the phase through a moment-projection routine, which is robust to both external spurious phases and experimental noise.
    In the noiseless case this achieves an exponential suppression of bias with respect to a naive mean estimator.
    In the presence of global depolarizing noise the bias is exponentially small in the circuit depth $t$, and the variance is $O(t^{-2}F^{-2})$ for circuit fidelity~$F$.
    This improves by a factor of $t^2$ over a naive shifted-and-rescaled-mean approach.
    To mitigate realistic circuit-level noise, we combine our method with the explicit unbiasing scheme described in Ref.~\cite{dutkiewiczError2025}.
    This yields an overhead interpolating between the $F^{-4}$ scaling typical of explicitly unbiased error mitigation and a reduced $F^{-2}$ scaling when the noise samples fall outside the promise interval.
    We validate our estimators on a small-scale simulation of the Ising model, observing better-than-expected performance for a global depolarizing noise approximation.
    This robustness to both multiple eigenvalues and realistic noise makes limited-depth phase estimation practical for early fault tolerant quantum experiments.
\end{abstract}

\maketitle 

% custom code to remove appendix subsections from TOC
\let\oldaddcontentsline\addcontentsline
\newcommand{\stoptocentries}{\renewcommand{\addcontentsline}[3]{}}
\newcommand{\starttocentries}{\let\addcontentsline\oldaddcontentsline}
\tableofcontents

\newpage
\section{Introduction}

\topic{opening}
Quantum hardware is quickly approaching the era of early fault tolerance~\cite{acharyaQuantum2024, acharyaSuppressing2023,bluvsteinLogical2024,paetznick2024demonstration}, where error correction is possible, but not perfect.
Experiments must continue to be designed to tolerate imperfections, and classical post-processing must be exploited to make efficient use of limited gate and qubit counts.
This also requires continued development of reliable error mitigation techniques to overcome residual bias and deliver accurate results.
Initial experimental demonstrations of quantum error correction~\cite{acharyaQuantum2024, acharyaSuppressing2023,paetznick2024demonstration} and logical gates~\cite{bluvsteinLogical2024} have already made progress into early fault-tolerance, alongside early theoretical works outlining what is possible within this era~\cite{campbellEarly2022,katabarwaEarly2023,zhangComputing2022,kshirsagarProving2024,wangState2022,wang2025efficient,nelsonAssessment2024,wanRandomized2022,lin2022heisenberg,dongGround2022,ding2023even,ding2023robust,bultriniBattle2023,liangModeling2024,akahoshiPartially2024,toshioPractical2024,akahoshiCompilation2024}.

\topic{QPE methods} 
A natural target for early-fault-tolerant quantum computing is quantum phase estimation (QPE).
This is a foundational computational task that underpins many applications in quantum simulation~\cite{babbushEncoding2018, reiherElucidating2017, leeEven2021, aspuru2005simulated, goings2022reliably, obrien2019calculating}, and more broadly in quantum information processing~\cite{Harrow09Quantum, Shor95Polynomial}.
QPE targets estimating the eigenvalue $e^{i \phi_0}$ of a unitary $U$\footnote{
    In practice, most applications target a specific eigenvalue $E_0$ of a Hamiltonian $H$ and implement unitaries e.g.~$U=e^{-iHt}$ (Trotter-based) or $U=e^{i\arccos(H/\lambda)}$ (Qubitization-based) with $t^{-1}, \lambda \geq \lVert H \rVert$.
} given access to a circuit implementing $U$ (or a controlled version thereof) and an initial state $\ket{\psi}$ that overlaps with the eigenstate $\ket{\phi_0}$.
To perform phase estimation, one can use the Hadamard test~\cite{lin2022heisenberg, dutkiewicz2022heisenberg, ding2023even, wangQuantum2023} or other methods~\cite{clinton2024quantum,russo2021evaluating} to estimate $\langle U^k\rangle$, and process this classically at multiple points $k$ to infer the spectrum of $U$.
Alternatively, one can coherently accumulate phase on a multi-qubit quantum register, and perform the quantum Fourier transform, which samples from a distribution that is peaked around the eigenphases of~$U$~\cite{aspuru2005simulated,nielsen2001quantum,berry2017improved}.
It is possible to interpolate between these two methods~\cite{dutkiewiczError2025, najafiOptimum2023, rendon2023low}, which becomes relevant in the early-fault-tolerant setting where arbitrarily long circuit depths cannot be afforded.
A third method uses quantum signal processing (QSP) circuits to implement block-encodings of some function of the Hamiltonian $f(H - \mathbb{1}x)$, which allows sampling from a distribution similarly peaked around eigenphases of $U$~\cite{dongGround2022, wang2025efficient, geFaster2019, Lin20Preparation, martyn2021grand}.

\topic{Status of QPE noise processing}
The classical post-processing of any of the above methods is a crucial piece of an early-FT QPE algorithm.
One must compensate here for both the presence of experimental noise and imperfect state preparation ($a_0:=|\braket{\phi_0}{\psi}|^2\ll 1$), which can otherwise bias the estimation of $\phi_0$.
Significant recent work has gone into optimizing Hadamard-test-based QPE in the presence of imperfect state preparation, using matrix pencil~\cite{obrien2019quantum,dutkiewicz2022heisenberg}, cumulative distribution function~\cite{lin2022heisenberg, wanRandomized2022}, and maximum likelihood methods~\cite{ding2024quantum, ding2023even, ding2023simultaneous}.
These methods have been shown in some cases to be robust to small amounts of noise~\cite{kimmel2015robust, guNoiseresilient2022, ding2023robust}, and can be error mitigated using standard techniques~\cite{caiQuantum2023} due to their intermediate estimation of expectation values $\langle\psi|e^{iHt}|\psi\rangle$.
The same is not true for QFT-based or QSP-based QPE algorithms, as these do not work with expectation values.
In previous work~\cite{dutkiewiczError2025}, we demonstrated that QFT-based QPE algorithms could be adapted to handle global depolarizing and circuit-level noise, but under the assumption of access to a perfect eigenstate.
(This used a maximum-likelihood framework that is immediately extensible to QSP-based QPE methods.)
Separately, Ref.~\cite{rendon2023low} constructed a bias-free estimator for QFT-based QPE with imperfect initial states, but in the absence of noise.
This leaves a gap in the literature to combine both sources of imperfection.

\subsection{Summary of key results}

In this work, we address the classical task of learning an eigenphase $\phi_0$ of a unitary from phase estimation data.
We focus on methods that sample from a distribution $p(x)$ sharply peaked around $x=\phi_0$.
In the presence of multiple eigenstates and experimental noise, $p(x)$ becomes distorted and difficult to model; our objective is to robustly infer $\phi_0$ from samples of this distorted distribution.
We focus on the case where experimental noise is large enough that one cannot afford the circuit depths required to estimate phases at the Heisenberg limit, and must instead run shorter QPE circuits and repeatedly sample from this distribution $p(x)$.

The main contribution of this work is to optimize the classical post-processing of samples from $p(x)$ under the assumption that $\phi_0$ is isolated within a region $\mathcal{D}$.
This assumption is necessary to bound the contributions from spurious eigenvalues $\phi_{j\neq 0}$ to the sampled distribution $p(x)$.
\begin{definition}[Def.~\ref{def:classical-subroutine}, informal]
   Fix a unitary $U$, and assume that there exists a known region $\mathcal{D}$ containing a single target eigenphase $\phi_0$. 
   Let $p(x)$ be the output distribution from a quantum computation implementing a phase estimation routine for $U$.
   The classical task of phase estimation is: given $M$ samples $\{x_j\}$ from $p(x)$, construct an estimator $\tilde\phi(=\tilde{\phi}[\{x_j\}])$ of $\phi_0$.
   The performance of this estimator is measured by the bias
    \begin{equation}
        b=\E[\tilde\phi | x \sim p(x)]-\phi_0
    \end{equation}
    and variance
    \begin{equation}
        \epsilon^2 = \E[(\tilde\phi-\phi_0-b)^2 | x \sim p(x)].
    \end{equation}
\end{definition}
We explain how samples from $p(x)$ can be efficiently generated by a quantum computer in Sec.~\ref{sec:qftqpe} and Sec.~\ref{sec:qspqpe}.

The above definition shifts the problem of near-term phase estimation to the slightly more abstract problem of extracting features of data drawn from a complex distribution.
In Sec.~\ref{sec:m_projection} we solve this task by discarding data that lies outside $\mathcal{D}$, to yield a filtered distribution $P(x)$ which we fit to a parametrized model $Q(x|\phi)$ that ignores all eigenphases other than $\phi_0$.
This fitting procedure relies on~\emph{moment projection}~\cite{amariMethods2000,murphyMachine2012,nielsenWHAT2018,tuananhleReverse2017}, so we call the resulting estimator the ``Filtered Moment projection Phase Estimator'' (FMPE).
By assuming $P(x\in\mathcal{D})\approx Q(x|\phi_0)$ we can strictly bound the resulting phase estimation error:

\begin{lemma}[Lemma~\ref{lem:m-projection-expansion}, informal]
Assume $U$ is a unitary with single eigenphase $\phi_0$ in a known interval $\mathcal{D}$. Let $P(x)$ be an output distribution from a phase estimation circuit of $U$ confined to the interval $\mathcal{D}$, and let $Q(x|\phi)$ be a model distribution parametrized by $\phi\in\mathcal{D_\phi}\subset\RR$.
Given $M$ samples $\{x_j\}$ from $P(x)$, define the moment projection estimator as the choice of $\phi$ that maximizes the likelihood $l(\phi|\{x_j\})=\sum_{j\,;\,x_j\in\mathcal{D}}Q(x_j|\phi)$.
To lowest order in the model error $h(x)=P(x)-Q(x|\phi_0)$, in the $M\rightarrow\infty$ limit this has variance $\epsilon^2\leq\mathcal{I}_0^{-1}M^{-1} + O(\|h\|_1)$, and bias $|b|\leq\|h\|_1 \mathcal{I}_0^{-1} S + O(\|h\|_1^2)$, where $S=\max_{x\in\mathcal{D}}[\partial_{\phi}\log Q(x|\phi)]_{\phi=\phi_0}$ is the maximum of the score, and $\mathcal{I}_0$ is the Fisher information of $Q(x|\phi)$ at $\phi=\phi_0$.
\end{lemma}

The moment projection estimator reliably isolates single phases in the presence of spurious eigenvalues outside the promise interval, but only in the setting where the noise can be explicitely included in the model $Q(x|\phi)$.
To address realistic experimental noise, we adapt the unbiasing procedure of Ref.~\cite{dutkiewiczError2025} and introduce the ``Noise-Unbiased" version of FMPE: NU-FMPE.
This approach samples from a quasiprobability distribution --- akin to Probabilistic Error Cancellation (PEC) --- but performs likelihood maximization instead of expectation value estimation.
However, since the quasiprobability sampling distribution generally differs from the model distribution, samples unlikely according to the model can cause large fluctuations in the log-likelihood and a high estimator variance.
Following Ref.~\cite{dutkiewiczError2025}, we solve this by regularising the likelihood function:
\begin{lemma}[Lemma~\ref{lem:nme-expansion}, informal]
Assume $U$ is a unitary with single eigenphase $\phi_0$ in a known interval $\mathcal{D}$, and fix a regularization constant $c\geq 0$.
Let $P(x)$ be the output distribution from a phase estimation circuit of $U$ confined to the interval $\mathcal{D}$, let $Q(x|\phi)$ be a model distribution parametrized by $\phi\in\mathcal{D}$, and let $ Q_c(x_j|\phi)=Q(x|\phi)+c$ be the regularized (non-normalized) model distribution.
Assume the ability to write down a quasiprobability distribution
\begin{equation}
    \sum_{a=0}^{r-1}\alpha_a P_a(x) = P(x),\label{eq:qp_dist}
\end{equation}
and sample from the distributions $P_a(x)$ using a noisy quantum device.
Given $M$ samples $\{x_j\}$ from the $P_a(x)$ distributed with probability $a\sim|\alpha_a|/\|\alpha\|_1$, define the Noise-Unbiased Moment projection Phase Estimator (NME) of~$\phi_0$ as the value that maximises the quasi-likelihood
\begin{equation}
    \ell(\phi|\{x_j,a_j\}) =
\frac{\|\alpha\|_1}{M}\sum_{j=1}^M
\mathrm{sgn}(\alpha_{a_j})\log [ Q_c(x_j|\phi)]+c\int_{\mathcal{D}}\dd  x\,\log[ Q_c(x|\phi)].
\end{equation}
To lowest order in the error $h(x)=P(x)-Q(x|\phi_0)$, the NU-FMPE has bias $\|h\|_1 S_c\,\mathcal{I}_c^{-1}+O(\|h\|_1^2)$ and variance $\|\alpha\|_1^2\epsilon^2=S_c^2\,\mathcal{I}_c^{-2}+O(\|h\|_1)$, where $S_c=\max_{x\in\mathcal{D}}[\partial_{\phi}\log  Q_c(x|\phi)]_{\phi=\phi_0}$ is the maximum of the score, and $\mathcal{I}_c$ is the Fisher information of $ Q_c(x_j|\phi)$ at $\phi=\phi_0$.
\end{lemma}

The results thus far hold for any distribution generated by a QPE circuit.
To obtain analytic resource requirements, we explicitly calculate the above bounds for a model with a Gaussian model distribution $Q(x|\phi)\propto e^{-(x-\phi)^2/2\sigma^2}$, and a linear combination of Gaussians for the true distribution $P(x)$.
We show how this distribution can be generated using phase estimation in Lemma~\ref{lem:gaussian_kernel_synthesis} following Ref.~\cite{rendon2024improved}.
We then make connection to the wider phase estimation literature by replacing our promise interval with an initial guess of $\phi_0$ and a promise of a gap $\Delta$ to other eigenvalues, and obtain the following result for the number of calls to the (controlled) unitary $U$.
\begin{theorem}[Thm.~\ref{thm:fmpe-cost-nonoise}, informal]\label{thm:fmpe-cost-nonoise-informal}
    Let $U$ be a unitary with spectral gap $\Delta$ around a target state $\phi_0$ ($\forall_{j>0} |\phi_j - \phi_0| > \Delta$).
    Assume oracle access to a controlled version of $U$, and an initial state $\ket{\psi}$ such that $|\braket{\phi_0}{\psi}|^2 > \eta$.
    Further assume an initial estimate $\phi_{guess}$ of $\phi_0$ such that $|\phi_0-\phi_{guess}|<\Delta/3$.
    Then, one can produce an estimate $\tilde\phi$ of $\phi_0$ with RMS error $\epsilon$ using $M = O(\eta^{-1}t^{-2}\epsilon^{-2})$ samples of a phase estimation circuit, where each circuit requires $t= {\Omega}(\Delta^{-1}\log^{1/2}(\Delta\epsilon^{-1}\eta^{-1}))$ calls to the unitary $U$, and the total number of calls $T = O(\eta^{-1}t^{-1}\epsilon^{-2})$.
\end{theorem}

In the absence of any noise, Theorem~\ref{thm:fmpe-cost-nonoise-informal} recovers the Heisenberg limit when $t\sim\epsilon^{-1}$.
To relax the noiseless assumption, in Sec.~\ref{sec:gaussian_with_gdn} we extend moment projection with a Gaussian model distribution to circuits affected by global depolarizing noise.
This corresponds to a uniform probability distribution, and the optimal mitigation strategy simply incorporates this into the model function $Q(x|\phi)$.
This yields a similar theorem to the above, but with an additional noise cost.
\begin{theorem}[Theorem~\ref{thm:fmpe-cost-gdn}, informal]
    Let $U$ be a unitary with spectral gap $\Delta$ around a target state $\phi_0$ ($\forall_{j>0} |\phi_j - \phi_0| > \Delta$).
    Assume oracle access to a controlled version of $U$ with global depolarizing noise $e^{-\gamma}$ per call, and an initial state $\ket{\psi}$ such that $|\braket{\phi_0}{\psi}|^2 = a_0$.
    Further assume an initial estimate $\phi_{guess}$ of $\phi_0$ such that $|\phi_0-\phi_{guess}|<\Delta/3$.
    Then, one can produce an estimate $\tilde\phi$ of $\phi_0$ with RMS error $\epsilon$ using $M = O(e^{2\gamma t}a_0^{-2}t^{-2}\epsilon^{-2})$ samples of a noisy phase estimation circuit, where each circuit requires $t = \Omega(\Delta^{-1}(\gamma\Delta^{-1} + \log^{1/2}(a_0^{-2} \epsilon^{-1})))$ calls to the unitary $U$.
    The total number of calls to $U$ to execute the algorithm is $T = O(\epsilon^{-2}t^{-1}e^{2\gamma t}a_0^{-2})$.
    Minimizing $T$ as a function of $t$ at fixed $\gamma$ yields a cost $T=\Theta(\gamma\epsilon^{-2}a_0^{-2})$.
\end{theorem}

In Sec.~\ref{sec:gaussian_unbiasing}, we extend the above results to general noise models using NU-FMPE.
As the NU-FMPE is constructed using a quasiprobability distribution in a similar manner to PEC, one would expect that the resulting estimator recovers the Heisenberg limit with a similar factor $F^{-4}$ overhead to standard PEC.
However, the restriction of our data to the interval $\mathcal{D}$ allows us to post-select away noisy data that falls outside.
We quantify the remaining error mitigation overhead using a parameter $\kappa$, which measures the fraction of the mitigated noise that remains inside $\mathcal{D}$.
The resulting complexity interpolates between $F^{-4}$ when $\kappa=\Theta(1)$ and $F^{-2}$ when the noise lies entirely outside of the promise interval $\mathcal{D}$.
We summarize this in the following theorem:
\begin{theorem}[Theorem~\ref{thm:fnmpe-cost-pauli}, informal]\label{thm:fnmpe-cost-pauli-informal}
    Let $U$ be a unitary with spectral gap $\Delta$ around a target state $\phi_0$ ($\forall_{j>0} |\phi_j - \phi_0| > \Delta$).
    Assume access to a unitary circuit that implements controlled-$U$ under a known local Pauli noise model with circuit fidelity $F$, the ability to add gates to the unitary circuit to generate a quasiprobability distribution (Eq.~\eqref{eq:qp_dist}), and access to a preparation of an initial state $\ket{\psi}$ such that $|\braket{\phi_0}{\psi}|^2 \geq \eta$.
    Assume an initial estimate $\phi_{guess}$ of $\phi_0$ such that $|\phi_0-\phi_{guess}|<\Delta/3$.
    Then, one can construct an estimate $\tilde{\phi}$ of $\phi_0$ using $M=\widetilde{O}\big((F^{-2}+\kappa F^{-4})\eta^{-1}t^{-2}\epsilon^{-2}\big)$ samples of a phase estimation circuit, where each circuit has fidelity $F$ and uses $t=\Omega\big(\Delta^{-1}\log^{1/2}(\Delta\epsilon^{-1}\eta^{-1})\big)$ calls to the unitary~$U$,
    where $\kappa < 1$ is the excess sampling probability within the promise interval.
\end{theorem}
\noindent
In practice, we expect the circuit depth to be chosen such that $F^{-2}\approx e$ \cite{dutkiewicz2022heisenberg}, and for approximately uniform noise distributions $\kappa \approx |\mathcal{D}|/2\pi \ll 1$.
Therefore, the $\kappa F^{-4}$ contribution is expected to be subdominant to $F^{-2}$.

In Sec.~\ref{sec:numerics} we test our estimators numerically in the presence of local depolarizing noise, on a toy phase estimation problem of a 4-qubit Ising model with up to 10 ancilla qubits.
Beyond confirming our analytic results, this provides a practical implementation guide for those desiring to use these estimators.
We observe that the moment projection estimator assuming global depolarizing noise performs surprisingly well, often outperforming the noise-unbiased estimator, especially at low sample counts $M$.
Because local and global depolarizing noise produce broadly similar distributions, fitting the functional form of the noise proves preferable to cancelling it via a quasiprobability distribution.
This suggests future improvements to the NU-FMPE via more accurate modelling of the noisy distribution may be achievable.

\section{Definitions}
\label{sec:definitions}

In this section we define the phase estimation problem that we will focus on solving in this work (Def.~\ref{def:classical-subroutine}).
This splits phase estimation as a whole into quantum and classical subroutines~\cite{dutkiewicz2022heisenberg}, which in our case interface via the distribution from which the quantum computer provides samples from (Def.~\ref{def:kernel_function}).
(We defer the discussion of how these samples are obtained to Section~\ref{sec:background}.)
We modify this distribution in Def.~\ref{def:noisy-distribution} by adding noise, and in Def.~\ref{def:filtered_distribution} by filtering (via rejection) to an interval, and in Def.~\ref{def:gaussian_kernel_function} give a specific Gaussian example (which we will use throughout this work).

Quantum phase estimation takes as input a unitary $U$ and initial state $|\psi\rangle$. 
The output of QPE depends on the decomposition of this state $|\psi\rangle$ in the eigenbasis of $U$.
\begin{definition}[Spectral distribution]\label{def:spectral_distribution}
    Let $U$ be a unitary operator with eigenbasis $U|\phi_j\rangle=e^{i\phi_j}|\phi_j\rangle,\, \phi_j\in[0,2\pi)$.
    The spectral distribution of a state $|\psi\rangle$ in the eigenbasis of $U$ is the function
    \begin{equation}\label{eq:spectral_distribution}
        a(x) \Big(= a_{U,|\psi\rangle}(x)\Big) := \sum_j a_j\delta(x-\phi_j),\hspace{1cm} a_j := |\langle \phi_j|\psi\rangle|^2.
    \end{equation}
\end{definition}
The normalization of the state $|\psi\rangle$ ensures that $a(x)$ is a normalized probability distribution:
\begin{equation}
    \int_0^{2\pi}a(x)\dd x = \sum_j |\langle \phi_j|\psi\rangle|^2 = 1.
\end{equation}
In this work, we consider variants of QPE which use a quantum computer to provide samples from a distribution approximating $a(x)$.
Perfectly sampling from $a(x)$ is in general not possible; instead, one typically approximates the delta functions $\delta(x-\phi_j)$ in Eq.~\eqref{eq:spectral_distribution} by convolving with a so-called kernel function $f(x)\approx \delta(x-\phi_j)$
\begin{definition}[Kernel function]\label{def:kernel_function}
    A kernel function is a non-negative normalized function
    $f:[-\pi,\pi)\rightarrow [0,\infty)$; $\int_{-\pi}^{\pi}f(x)\dd x=1$.
    Given such a function, the smoothed spectral distribution of a state $|\psi\rangle$ in the eigenbasis of $U$ is
    \begin{equation}
        [f * a](x) := \sum_ja_jf(x-\phi_j),
    \end{equation}
    where $a(x)$ is the spectral distribution (Def.~\ref{def:spectral_distribution}).
\end{definition}
Various kernel functions have been explored in the quantum phase estimation literature, in particular the Fejer kernel \cite{nielsen2001quantum}, sine window \cite{luisOptimum1996,leeEven2021, dutkiewiczError2025}, cosine taper \cite{rendon2022effects}, DPSS taper \cite{patel2024optimal}, the Kaiser window \cite{berry2024analyzing}, and the Gaussian kernel \cite{wang2025efficient, rendon2023low, rendon2024improved}.
In this work we focus on Gaussian kernels due to their ease of manipulation, but our techniques can be readily adapted to any kernel function with exponentially-decaying tails and controllable width.
\begin{definition}[Gaussian kernel function]\label{def:gaussian_kernel_function}
    A Gaussian kernel function of width $\sigma > 0$ and precision $\epsilon_{\text{synth}} \geq 0$ is a kernel function $f_\sigma(x)$ which approximates a Gaussian in the interval $(-\pi, \pi)$ such that
    \begin{equation}
        \sup_{x\in (-\pi, \pi)} \left| f_\sigma(x) - \frac{e^{- x^2 / 2\sigma^2}}{\int_{-\pi}^{\pi} e^{- x^2 / 2\sigma^2} \dd x}  \right| < \epsilon_{\text{synth}},
    \end{equation}
\end{definition}
The Gaussian kernel function $f_\sigma(x)$ can be constructed by polynomial approximation, with only a logarithmic overhead in the approximation precision~ \cite{wang2025efficient}.
We will ignore the details of this approximation in this work, and assume $f_\sigma \propto  e^{- x^2 / 2\sigma^2}$.

\begin{figure}
    \centering
    \includegraphics[width=\linewidth]{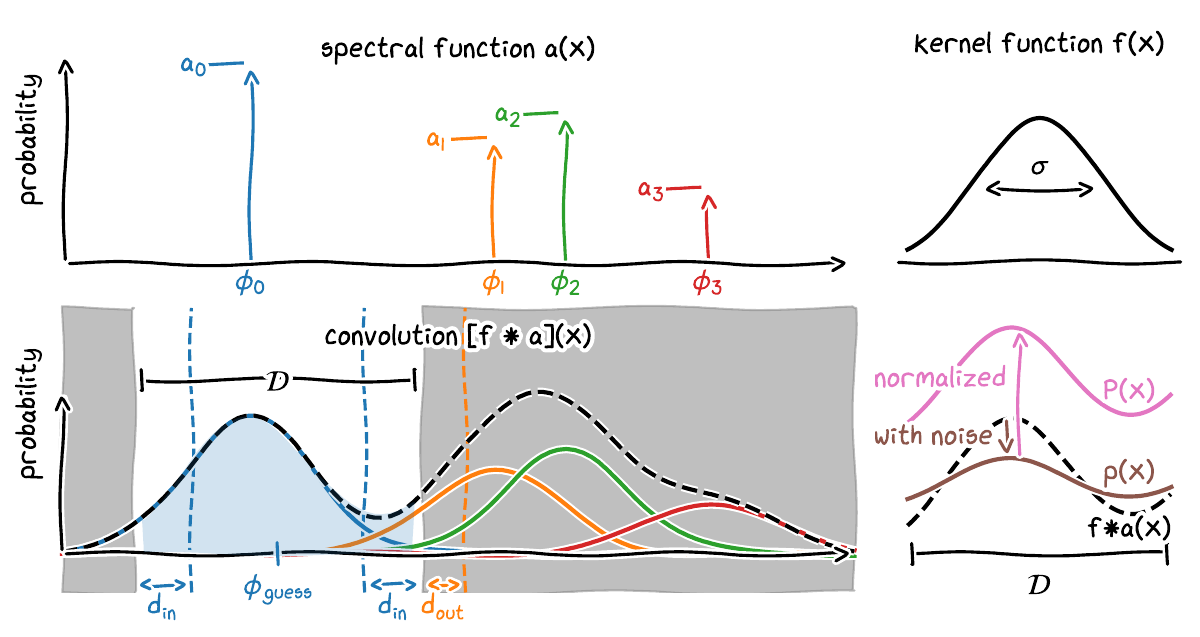}
    \caption{{ 
        Schematic plot illustrating the construction of the signal distribution of a filtered QPE experiment.
        \textbf{(Top left)} A unitary operator $U$ and initial state $|\psi\rangle$ define the spectral distribution $a(x)$ in Def.~\ref{def:spectral_distribution} -- a sum of Dirac deltas centered at the eigenphases $\phi_j$ of $U$ with amplitudes $a_j$.
        \textbf{(Top right)} We define a kernel function $f(x)$ (Def.~\ref{def:kernel_function}), with a shape and width $\sigma$ depending on the details of the circuit.
        \textbf{(Bottom left)}
        The QPE circuits can sample from a distribution obtaining by convolution of $a$ and $f$ -- the distribution in Def.~\ref{def:kernel_function}, shown by the dashed black line.
        (The colored lines indicate the contributions of each eigenphase $\phi_j$ to the total distribution.)
        The filtering procedure discards all samples outside of a filtering region $\mathcal{D}$, which is chosen assuming that the phase of interest $\phi_0$ is within the interval (farther than an \emph{inner buffer} distance $d_{\text{in}}$ from the interval edges) and all other phases are outside the interval (farther than an \emph{outer buffer} distance $d_{\text{out}}$ from the interval edges) -- see Def.~\ref{def:promise_interval}.
        The unnormalized distribution of filtered outcomes in the absence of noise is highlighted in blue.
        \textbf{(Bottom right)}
        Adding noise to the distribution $f*a$ yields the distribution $p(x)$ (Def.~\ref{def:noisy-distribution}). 
        Normalizing this distribution within the filtering interval $\mathcal{D}$ we obtain the filtered noisy distribution $P(x)$ (Def.~\ref{def:filtered_distribution}).
    }}
    \label{fig:scheme}
\end{figure}

In the absence of noise, the quantum computer targets sampling from a smoothed spectral distribution $[f*a_{U,|\psi\rangle}]$.
The circuits required for this distribution typically have a depth proportional to the inverse of width $\sigma$ of the kernel $f$ (which is a consequence of the no-fast-forward theorem~\cite{Berry07Efficient}).
This becomes more complicated in an early fault-tolerant or NISQ setting.
Under a stochastic noise model (where noise is treated as a series of discrete events that either occur or do not), the noisy probability distribution can be rewritten as a convex combination of the noiseless distribution $a(x)$ and a distribution $u(x)$ of all cases where a noise event happened, weighted by the fidelity $F\in(0,1)$:

\begin{definition}[Noisy distribution] \label{def:noisy-distribution}
Let $U$ be a unitary, $|\psi\rangle$ be a state, and $f$ be a kernel function (Def.~\ref{def:kernel_function}).
In the presence of stochastic noise, let $F$ be the circuit fidelity (the probability of no noise event occuring), and let $u(x):[0, 2\pi)\rightarrow [0,\infty)$ be the distribution sampled from the quantum computer in the event that at least one noise event occurs.
The noisy distribution is then the function
\begin{equation}
    p(x) = F \cdot [f * a](x) + (1-F) \cdot u(x),
\end{equation}
where $[f * a](x)$ is the distribution of $|\psi\rangle$ in the eigenbasis of $U$ with kernel function $f$ in the absence of noise (Def.~\ref{def:kernel_function}).
\end{definition}
Typically the circuit fidelity $F$ is exponentially small in the circuit depth $T$.
As we expect $T\propto\sigma^{-1}$, one can assume a form $F=e^{-\gamma/\sigma}$ for some decoherence rate $\gamma$.
Optimizing QPE in the presence of noise trades thus requires trading between deep circuits with small $\sigma$ (making estimation easier), and shallow circuits with large $F$~\cite{dutkiewiczError2025}.

The classical subroutine of a QPE algorithm takes the samples output from the quantum subroutine, and processes them to recover information about the phases $\{\phi_j\}$.
In this work we consider the estimation of a specific single phase, $\phi_0$, which is identified by the promise of a region $\mathcal{D}$ in which it alone exists.
This is in contrast to methods which attempt to estimate all phases in the problem simultaneously~\cite{obrien2019quantum, dutkiewicz2022heisenberg, ding2023simultaneous,ding2024quantum}, or to estimate a discretized form of the spectral distribution $a(x)$ itself~\cite{somma2019quantum}, or to prepare the ground state itself~\cite{geFaster2019,Lin20Preparation}.
In practice, this promise is a reasonable assumption for e.g.~the ground state energy of a gapped Hamiltonian, where $\mathcal{D}$ could be estimated via classical means or a lower-cost QPE method.
\begin{definition}[Promise interval]\label{def:promise_interval}
    Let $U$ be a unitary with eigendecomposition $U|\phi_j\rangle =e^{i\phi_j}|\phi_j\rangle$, $\phi_j\in[0,2\pi)$.
    A promise interval $\mathcal{D}$ for an eigenphase $\phi_0$ with inner buffer $d_{\text{in}}$ and outer buffer $d_{\text{out}}$ is a connected~\footnote{To avoid unnecessary complications, we do not consider promise intervals that wrap around the circle. In the case where this would occur (i.e. if $\phi_0$ is near $0$ or $2\pi$), one can trivially shift all phases by a constant to yield a connected promise interval $\mathcal{D}$.} subset of $[d,2\pi-d)$ centred around $\phi_{\mathrm{guess}}$ that satisfies the following two properties:
    \begin{enumerate}
        \item $\phi_0\in[\phi_{\mathrm{guess}}-\frac{|\mathcal{D}|}{2} + c, \phi_{\mathrm{guess}}+\frac{|\mathcal{D}|}{2} - c]$
        \item $\forall j\neq 0 ,\quad \phi_{j}\notin [\phi_{\mathrm{guess}}-\frac{|\mathcal{D}|}{2} - d ,\phi_{\mathrm{guess}}+\frac{|\mathcal{D}|}{2} + d]$
    \end{enumerate}
\end{definition}
In words, we require a buffer zone of width $d_{\text{in}} + d_{\text{out}}$ around the edges of $\mathcal{D}$, such that no phases lie within this buffer.
The existence of a promise interval $\mathcal{D}$ with inner buffer $d_{\text{in}}$ and outer buffer $d_{\text{out}}$ implies a gap $\Delta\geq d_{\text{in}}+d_{\text{out}}$ between $\phi_0$ and any other phase, and a promise interval can always be constructed given an $\epsilon$-accurate estimate of $\phi_0$ (i.e.~$|\phi - \phi_\text{guess}| < \epsilon$) and the promise of a gap $\Delta>2\epsilon + d_{\text{in}} + d_{\text{out}}$ (choosing a promise interval of size $|\mathcal{D}| = 2\epsilon + 2d_{\text{in}}$).
For simplicity, in this work we will fix the inner buffer size $d_{\text{in}} = \frac{|\mathcal{D}|}{6}$ as a fraction of the interval width $|\mathcal{D}|$.

To use the promise interval $\mathcal{D}$ to optimize our estimation of $\phi_0$, we will filter the noisy distribution $p(x)$ to lie within $\mathcal{D}$ only.
This yields a new distribution, that can be sampled from by sampling from $p(x)$ and rejecting samples from outside $\mathcal{D}$:
\begin{definition}[Filtered noisy distribution]\label{def:filtered_distribution}
    Let $p(x)$ be a noisy distribution (Def.~\ref{def:noisy-distribution}) for unitary $U$, state $|\psi\rangle$, kernel function $f$ and fidelity $F$, and let $\mathcal{D} \subset [0, 2\pi)$.
    The filtered noisy distribution $P(x)$ is the normalized distribution with support on $\mathcal{D}$ defined as
    \begin{equation}
        P(x) = \begin{cases}\frac{p(x)}{\int_{\mathcal{D}} p(x) \dd x} & x\in\mathcal{D}\\
     0 & x\notin\mathcal{D}
    \end{cases}.
    \end{equation}
\end{definition}

Given oracular access to some filtered noisy distribution $P(x)$, we measure the performance of our classical estimation of $\phi_0$ by bounding the bias and standard deviation of the constructed estimator $\tilde\phi_0$.
This is a common metric used in the quantum metrology community~\cite{luisOptimum1996,belliardoAchieving2020}.
It differs slightly from the confidence interval formalism commonly used in computer science~\cite{nielsen2001quantum}, however the two can be related to each other with at most a logarithmic overhead in the error probability through Chebyshev's inequality.
We are now ready to state the phase estimation problem considered in this work.
\begin{definition}[Classical and quantum subroutines of QPE]\label{def:classical-subroutine}
    Let $P(x)$ be a filtered noisy distribution (Def.~\ref{def:filtered_distribution}) for given $U$, $|\psi\rangle$, $f$, $F$, $\mathcal{D}$; and assume $\mathcal{D}$ is a promise interval (Def.~\ref{def:promise_interval}) for $\phi_0$, with buffers $d_{\text{in}}$ and $d_{\text{out}}$.
    The classical subroutine of the quantum phase estimation algorithm, given $M$ samples $\{x_j\}$ from $P(x)$, constructs an estimator $\tilde{\phi}(=\tilde{\phi}[\{x_j\}])$ for $\phi_0$, with bias
    \begin{equation}
        b=\E[\tilde\phi | x \sim P(x)]-\phi_0
    \end{equation}
    and variance
    \begin{equation}
        \epsilon^2 = \Var[\tilde\phi|x \sim P(x)] = \E[(\tilde\phi-\phi_0-b)^2 | x \sim P(x)]
    \end{equation}
    The quantum subroutine of a QPE algorithm is to generate samples from the distribution $P(x)$ given $\mathcal{D}$, $f$, $F$, a~circuit implementation of $U$, and copies of $\ket{\psi}$.
\end{definition}

In this work, for the sake of recovering simpler constant factors and making the proofs clearer, we restrict to a version of the above problem, but with the distribution $p(x)$ fixed to be a Gaussian:
\begin{definition}[Classical QPE subroutine with Gaussian kernels]
    Def.~\ref{def:classical-subroutine} with $f=f_\sigma$ Gaussian kernel function with variance $\sigma^2$ and $\epsilon_{\text{synth}}=0$, and inner buffer $d_{\text{in}} = |\mathcal{D}|/6$.
\end{definition}
The techniques used in our proofs easily extend to smaller inner buffers and any kernel function that vanishes exponentially in $x/\sigma$ [i.e.~$f_\sigma(x) \sim o\left(\frac{1}{\exp(x/\sigma)}\right)$], but the resulting expressions for bias and variance will have different constant factors.

\section{Background}
\label{sec:background}

Due to its BQP-completeness~\cite{Wocjan06Several} and use as a subroutine in various quantum algorithms~\cite{Shor95Polynomial, Harrow09Quantum, aspuru2005simulated}, much prior work has focused on optimizing phase estimation in various settings.
In this section, we describe the various methods for constructing the quantum subroutine for quantum phase estimation, followed by a review of the state of the art in constructing the classical subroutine (following our division of QPE into two subroutines in Def.~\ref{def:classical-subroutine}).

A large body of work exists on quantum phase estimation that has a different quantum-classical interface to the one considered in this work, namely single-control QPE and related methods~\cite{kimmel2015robust,obrien2019calculating,lin2022heisenberg,dutkiewicz2022heisenberg}.
These methods still have a well-defined split into quantum and classical subroutines, but here the quantum computer provides estimates of expectation values $\langle\psi|e^{iHt}|\psi\rangle$ instead of samples from a smoothed spectral distribution $f*a(x)$ (Def.~\ref{def:kernel_function}).
This allows these methods to access standard error mitigation techniques for expectation values~\cite{caiQuantum2023}.
However, it was shown in Ref.~\cite{najafiOptimum2023} that single control methods converge slower in estimation in the absence of noise.
Furthermore, the ability to filter noisy data (as studied in this work) allows phase estimation to tolerate higher levels of noise, analogous to the difference in fidelity cost between postselection and rescaling~\cite{dutkiewiczError2025}.
Thus, extending error mitigation techniques from single-control to QFT-based and QSP-based methods is clearly of relevance for early-fault-tolerant phase estimation.

\subsection{QFT-based phase estimation}\label{sec:qftqpe}

Quantum phase estimation was first studied as a subroutine in Shor's factoring algorithm~\cite{Shor95Polynomial}.,
Here, the quantum algorithm uses application of $U^k$ controlled on the $k$th basis state of a control register prepared in some initial state $\sum_{k=0}^{K-1}b_k|k\rangle$, to generate
\begin{equation}
    (c_K-U)\sum_{k=0}^{K-1}b_k|k\rangle|\psi\rangle = \sum_{k=0}^{K-1}\sum_{j}b_ka_je^{ik\phi_j}|k\rangle|\phi_j\rangle.
\end{equation}
Here, $(c_K-U)=\oplus_{k=0}^{K-1}(|k\rangle\langle k|\otimes U^k)$ is the unitary $U$ controlled by the entire quantum register.
The algorithm proceeds by performing the quantum Fourier transform on the control register, and reading out the result.
The cost of executing a single shot of $c_K-U$ is proportional to the maximum number of calls $K$ to the unitary; for Shor's algorithm this is logarithmic, however in quantum simulation this is bounded below by the no-fast-forward theorem~\cite{Berry07Efficient} to be worst-case linear in $K$.

A large body of work in phase estimation has focused on the optimization of the control register state.
Originally, the $b_k$ values were chosen to be a uniform superposition across $K=2^{n_a}$ qubits ($b_k=2^{-n_a}$) \cite{cleve1998quantum}, which was popularized as the `textbook phase estimation' due to its appearance in Ref.~\cite{nielsen2001quantum}.
Following the quantum Fourier transform, measurement in the computational basis sample bitstrings $\tilde{x}$ from a convolution of the spectral function with a Fejer kernel as kernel function
\begin{equation}
    p^{(QPEA)}(\tilde x) = \sum_j a_j \frac{2\pi}{K}f^{\text{(Fejer)}}_K\left(\frac{2\pi}{K}\tilde{x} - \phi_j\right), 
    \quad \frac{2\pi}{K}f^{\text{(Fejer)}}_K(x) = \frac{1}{K^2}\frac{1-\cos(K x)}{1-\cos(x)},
    \quad \tilde{x}\in\{0, 1, ..., K-1\}.
\end{equation}
Textbook phase estimation has the advantage of having simple state preparation, and yielding exact eigenvalues given the promise that $2^{n_a}\phi_j/(2\pi)\in\mathbb{N}$.
However, the Fejer kernel has a suboptimal width $\sigma$ (as a function of $K$); $\sigma\sim K^{-1/2}$, as $\epsilon^{-1}=O(\sigma^{-1})$, this implies that a classical estimator constructed from this data cannot achieve the Heisenberg limit (variance $\epsilon^{2}\sim K^{-2}$).
Refs.~\cite{higgins2009demonstrating} improved on this by careful choice of the control state amplitudes $c_k$, such that the resulting kernel functions achieved tighter widths $\sigma$.
Namely, Ref.~\cite{luisOptimum1996, babbushEncoding2018} uses a sine kernel to achieve an optimal standard deviation, while Ref.~\cite{gorecki2020pi, berry2024analyzing} uses a Kaiser window to achieve optimal confidence-probability bounds.
Gaussian kernels are also considered in Refs. \cite{rendon2024improved, rendon2023low}, as they allow for easy analysis when multiple samples are involved.
Circuit constructions for sine states and Kaiser window states are known~\cite{babbushEncoding2018, berry2024analyzing}.
More generally, as the cost of phase estimation for non-fast-forwardable unitaries grows linearly in the control register Hilbert space size, constructing even arbitrary control initial states should not be a significant factor in the overall cost of phase estimation.

\topic{Random phase technique}
In their standard definition, QFT-based QPE algorithms sample discrete variables.
However these can be easily be adapted to the continuous description of the distributions we gave in Defs.~\ref{def:kernel_function}, \ref{def:noisy-distribution} and \ref{def:filtered_distribution} through the random-phase technique \cite{cornelissen2023sublinear, vanapeldoorn2023quantum}.
This technique consists in classically sampling a phase $\phi_\text{ref}$ uniformly at random in the interval $[0, 2\pi)$ before every circuit run, and implementing the QPE circuit on the modified unitary $e^{i \phi_\text{ref}}U$ (this implies a very small additive overhead, logarithmic in the precision of the classical variable).
The reference phase is then added the output of the quantum circuit, yielding a random variable $x = \frac{2\pi}{K}\tilde{x} - \phi_\text{ref}$ with continuous support in $[0, 2\pi)$.
For instance, the resulting distribution for textbook QPE becomes
\begin{equation} \label{eq:continuous-fejer-kernel}
    p^{(QPEA)}(x) = \sum_j a_j f^{\text{(Fejer)}}_K\left(x - \phi_j\right), 
    \qquad f^{\text{(Fejer)}}_K(x) = \frac{1}{2\pi\,K}\frac{1-\cos(K x)}{1-\cos(x)},
    \qquad x\in[0, 2\pi).
\end{equation}
More generally, if the circuit samples bitstrings $\tilde{x}\in {0,1,\dots, K-1}$ with probability $p(\tilde x) = \sum_j a_j \frac{2\pi}{K}f(\frac{2\pi}{K}\tilde x-\phi_j)$ for some kernel function $f$, then the random phase technique will modify this to $p(\tilde x | \phi_\text{ref}) = \sum_j a_j \frac{2\pi}{K}f(\frac{2\pi}{K}\tilde x-(\phi_j+\phi_\text{ref}))$ and yield samples $x \in [0, 2\pi)$ distributed as
\begin{align}
    p(x) &= \sum_{\tilde x = 0}^{K-1}\int_{0}^{2\pi} \frac{\dd \phi_{\text{ref}}}{2\pi}  p(\tilde x | \phi_\text{ref}) \ \kappa\left(x - \left(\frac{2\pi}{K}\tilde{x} - \phi_\text{ref}\right)\right)
    \notag\\
    &= \sum_{\tilde x = 0}^{K-1}\int_{0}^{2\pi} \frac{\dd \phi_{\text{ref}}}{2\pi}\frac{2\pi}{K}\sum_j a_j f\left(\frac{2\pi}{K}\tilde{x} - (\phi_j+ \phi_\text{ref})\right)\ \kappa\left(x - \left(\frac{2\pi}{K}\tilde{x} - \phi_\text{ref}\right)\right)
    \notag\\
    &= \frac{1}{K}\sum_{\tilde x = 0}^{K-1}\sum_j a_j f\left(x-\phi_j\right) = \sum_j a_j f\left(x-\phi_j\right) = (a*f)(x).
    \label{eq:random_phase_noiseless}
\end{align}

To sample from $p(x)$ with a Gaussian kernel function of width $\sigma$ (Def.~\ref{def:gaussian_kernel_function}), we must prepare the QPE control register in a quantum state $\sum_{k=0}^{K-1}b_k|k\rangle$ whose computational basis amplitudes $b_k$ approximate a Gaussian distribution. 
The random phase technique (Eq.~\eqref{eq:random_phase_noiseless}) then smoothens out the discrete measurement grid, allowing us to sample from the (continuous) convolved distribution $f*a(x)$ (Def.~\ref{def:kernel_function}). 
We can bound the Gaussian tails to precision $\epsilon_{\text{synth}}$ by exploiting the Fourier duality of the discrete sampling errors analysed in \cite{rendon2024improved}.
This requires a register dimension (and thus maximum evolution time) scaling as $K = O(\sigma^{-1} \sqrt{\log(\epsilon_{\text{synth}}^{-1})})$.
Neglecting the one-off cost of preparing the state on the $n = \log_2(K)$ control qubits (which is subdominant to the cost of applying the $K$ controlled unitaries), we obtain the following result as a direct consequence of Ref.~\cite{rendon2024improved}, Theorem16:
\begin{lemma}[Gaussian kernel synthesis]
\label{lem:gaussian_kernel_synthesis}
One can prepare the window state and sample from the continuous phase distribution $p(x)$ of a Gaussian kernel function $f_\sigma$ (Def.~\ref{def:gaussian_kernel_function}) to precision $\epsilon_{\text{synth}} > 0$ using a preparation circuit of depth $O\left(\sigma^{-1} \sqrt{\log(\epsilon_{\text{synth}}^{-1})}\right)$.
\end{lemma}

In the presence of noise this distribution will change as per Def.~\ref{def:noisy-distribution}. 
In order to get to the filtered distribution of Def.~\ref{def:filtered_distribution}, we neglect the samples that lie outside of the given promise interval $\mathcal{D}$.
To obtain $M$ samples from the filtered distribution $P(x)$, we need to run the quantum subroutine $M'>M$ times, yielding an average sample overhead {{$\E[M']/M = 1 / \int_\mathcal{D} p(x)$}}.
For kernel functions with fast decaying tails, such as the Gaussian kernel $f_\sigma$ (Def.~\ref{def:gaussian_kernel_function}), this overhead is approximately $a_0^{-1}$.
We will discuss this more in detail in section~\ref{sec:m-projection-for-qpe}.

\topic{Classical subroutine}
In the absence of noise, and given an initial eigenstate $|\psi\rangle=|\phi_0\rangle$, the optimal strategy for phase estimation involves a single-shot readout of an estimate of $\phi_0$ from the control register~\cite{nielsen2001quantum}.
This renders complicated classical post-processing unnecessary, as one cannot optimize further over a single estimate.
As a mixed state can be purified to the ground state using a circuit of depth $a_0^{-1/2}\Delta^{-1}$~\cite{geFaster2019,Lin20Preparation}(with $\Delta$ the gap to the first excited state), less focus has been traditionally given to the classical QPE subroutine.
However, in the presence of noise, one cannot afford the depth of such circuits,.
In Ref.~\cite{rendon2023low}, QFT-based phase estimation was studied in the absence of noise, using a mean estimator on a subset of lowest-energy outcomes.
This yielded a bound on the cost of estimation of circuit depth $T = O(1/\Delta)$ and number of repetitions $M = O(1/\epsilon^2)$.
In Ref.~\cite{dutkiewiczError2025}, we studied QFT-based phase estimation of an eigenstate in the presence of general noise and global depolarizing noise, finding that optimal phase estimation occurs at circuit fidelities $\sim 1/e \sim 30\%$.
However, no works have yet studied the realistic phase estimation context, with non-eigenstate starting states and noise.

\subsection{QSP-based phase estimation}\label{sec:qspqpe}

Recently a new class of algorithms to estimate eigenvalues emerged, which use a completely different quantum subroutine from the QFT-based or Hadamard-test-based QPE algorithms \cite{martyn2021grand, dongGround2022, wang2024faster}.
Given a target unitary $U$, these algorithms rely on variants of quantum signal processing  \cite{low2017optimal, lowHamiltonian2019, Gilyen19QSVT} to construct block-encodings of polynomial functions~$h(U)$:
\begin{equation}
    W_{h, U} = \begin{bmatrix}h(U) & \cdot \\ \cdot & \cdot\end{bmatrix}
    \, , \quad
    \bra{0} W_{h, U} \ket{0} = h(U),
\end{equation}
with $|h(e^{i\phi})| < 1\, \forall\phi\in[0, 2\pi)$.
Applying $W_{h, U}$ on $\ket{0}\ket{\psi}$ and measuring the control qubit will yield $0$ with probability
\begin{equation} \label{eq:binary-test-probability}
    p^{0}_{h, U, \psi} = \bra{\psi}\bra{0} W^\dag_{h, U} \ketbra{0} W_{h, U}\ket{0}\ket{\psi} = \bra{\psi} h^\dag(U) h(U) \ket{\psi},
\end{equation}
and $1$ otherwise.
Sampling from these binary-test circuits with an appropriately-chosen set of functions $\{h\}$ allows to extract information about the phases of $U$.

Following the approach of \cite{wang2025efficient}, we aim to reconstruct the smoothed spectral distribution $f*a(x)$ of Def.~\ref{def:kernel_function} with kernel $f$ by choosing a set of $h_x$ such that $h_x(e^{i \phi}) = \sqrt{f(x - \phi)}/\max_{\phi}\sqrt{f(\phi)}$, with $x$ taking values in the interval $[0, 2\pi)$.
The sample probability $p^{0}_{f(x-\phi), U, \psi}$ will be proportional to $f*a(x)$:
\begin{equation}
\label{eq:prob-accept}
    p^{0}_{h_x, U, \psi} %= \bra{\psi} h_x^\dag(U) h_x(U) \ket{\psi} 
    = \frac{1}{\max_\phi f(\phi)} \sum_j a_j f(x-\phi_j).
\end{equation}
Given access to these binary samples with probabilities, we can obtain $M$ samples from the smoothed spectral distribution $[f*a](x)$ through rejection sampling, with the following steps: (1) sample $x$ at random in $[0, 2\pi)$, (2) run the binary-test circuit with $h=h_x$; if the outcome is $0$ accept the sample $x$ (with probability $p^{0}_{h_x, U, \psi}$)
and (3) repeat from 1 until $M$ samples are accepted.
We call $M'$ the total number of repetitions, i.e.~the total number of quantum circuits ran in order to obtain $M$ accepted samples.
The sampling overhead is equal to the inverse of the expected acceptance probability
\begin{equation}
    \frac{\E[M']}{M} = \left[\frac{1}{2\pi}\int_0^{2\pi} p^0_{h_x, U, \psi}  \dd x \right]^{-1} = 2\pi \max_\phi f(\phi).
\end{equation}
In this setting, we can naturally implement filtering (Def.~\ref{def:filtered_distribution}) by changing step (1) of rejection sampling, choosing $x$ uniformly at random in $\mathcal{D}$ rather than in $[0, 2\pi)$.
The sampling overhead is then reduced to
\begin{equation}
    \frac{\E[M']}{M} = \left[\frac{1}{|\mathcal{D}|}\int_{\mathcal{D}} p^0_{h_x, U, \psi} \dd x\right]^{-1} =  \frac{|\mathcal{D}| \, \max_\phi f(\phi)}{\int_{\mathcal{D}} [f * a](x) \dd x}.
\end{equation}
In the case of a Gaussian kernel function $f_\sigma$ (Def.~\ref{def:gaussian_kernel_function}), the maximal value $\max_\phi f(\phi) = (\int_{-\pi}^{\pi} e^{- x^2 / 2\sigma^2} \dd x)^{-1}$ is proportional to $\sigma^{-1}$,
while $\int_{\mathcal{D}} [f * a](x) \dd x \xrightarrow{\sigma \to 0} a_0$ because only $\phi_0$ is in the promise interval.
The sampling overhead  is $\propto a_0^{-1}\sigma^{-1}|\mathcal{D}|$,
with an additional factor of $\sigma^{-1}|\mathcal{D}|$ comparing to the QFT-based method.

\topic{classical subroutine of \cite{wang2025efficient}}
In a pre-print version of \cite{wang2025efficient}, the authors proposed an algorithm (Algorithm 2 in \cite{wang2024faster}) to estimate $\phi_0$ using samples from a (filtered) Gaussian distribution (Def.~\ref{def:gaussian_kernel_function}). 
The authors generated data $\{x_j\}$ for this using the QSP-based circuits described in this section, and constructed a classical estimator by taking an average of the accepted samples; $\tilde\phi=\langle x_j\rangle_{x_j\in\mathcal{D}}$.
This algorithm further uses an adaptive choice of the interval $\mathcal{D}$ and Gaussian width $\sigma$ to achieve arbitrarily low bias.
Throughout this work we will use the mean estimator suggested here without these adaptive updates as an estimator to compare our results to.

Relative to the QFT-based implementation, QSP-based phase estimation techniques carry an additional sampling overhead of $\sigma^{-1}|\mathcal{D}|$.
In principle this could be reduced by an adaptive choice of $\mathcal{D}$, and QSP-based circuits have lower requirements for ancilla qubits ($1$ as opposed to $O(\sigma^{-1})$).
For simplicity we do not consider the overhead from the QSP scheme (nor do we consider adaptive updates of $\mathcal{D}$) further in this work.
However, the estimators designed in this work can be applied immediately to samples generated by the QSP-based rejection-sampling scheme above.

\section{Results}
\label{sec:results}

In this work, we construct a classical estimator for QPE (the ``filtered moment projection phase estimator'') in two steps.
The key idea here is that a sample from the smoothed spectral distribution that falls within the filtering interval is, with high likelihood, caused by the eigenvalue of interest $\phi_0$. 
Thus, we can fit the samples from the filtered distribution $P(x)$ [Def.~\ref{def:filtered_distribution}] with a simple model $Q(x|\phi)$ that consider a single eigenvalue $\phi$.
The small amount of samples due to eigenvalues other than $\phi_0$ (i.e.~spurious phases), will however produce a small amount of bias in the resulting estimator.
We bound both the bias and the variance of the resulting estimator.
We calculate these bounds in Sec.~\ref{sec:m_projection} for a generic moment projection estimator, under the assumption that $P(x)$ is close to $Q(x|\phi_0)$, but without assuming any specific form for $P(x)$ and $Q(x|\phi)$.

In practice, the effect of noise on the outcome distribution is far more complex than global depolarizing noise, and can't be modelled with an explicit functional form.
Instead, we can mitigate the noise using explicit unbiasing, a method we developed in previous work \cite{dutkiewiczError2025}.
In section \ref{sec:explicit-unbiasing} we show that the moment projection estimator also works with explicit unbiasing, which we refer to as \emph{noise-unbiasing} in this work.

To calculate the scaling of these bounds with quantities like the cost of execution on a quantum device, we must first fix a family of parametrized distributions.
We achieve this by ignoring all phases except for the target phase $\phi_0$, after which a distribution naturally occurs from the chosen kernel function (Def.~\ref{def:kernel_function}) that we sample data from.
In Sec.~\ref{sec:gaussian_no_noise} and Sec.~\ref{sec:gaussian_with_gdn} we focus on a Gaussian kernel (Def.~\ref{def:gaussian_kernel_function}) in the case of no noise and global depolarizing noise respectively.
In Lemma~\ref{lem:gaussian-no-noise} and Lemma~\ref{lem:gaussian-gdn}, we estimate the first-order constant factor terms for both cases, and in Theorem~\ref{thm:fmpe-cost-nonoise} and Theorem~\ref{thm:fmpe-cost-gdn}), we propagate this to costs in a standard phase estimation model.

\subsection{Moment projection estimator}\label{sec:m_projection}

\topic{motivation: max-likelihood}
When a stochastic phenomenon producing samples $x\sim P(x)$ can be modelled exactly with a parametrized distribution -- i.e.~there exists $Q(x|\phi)$ which matches the true distribution $Q(x|\phi_0) = P(x)$ for some \emph{true value} of the parameter $\phi=\phi_0$ -- the asymptotically optimal estimator for $\phi_0$ is obtained by likelihood maximisation.
In quantum phase estimation, this is the case if we are promised the initial state is the eigenstate $\ket{\psi}=\ket{\phi_0}$, thus the spectral distribution $a(x) = \delta(x-\phi_0)$.
Under the knowledge of the noise distribution $u(x)$ we can then fully model the distribution $p(x)$ (Def.~\ref{def:noisy-distribution}) by
\begin{equation}
    q(x|\phi) = F f(x-\phi) + (1-F) u(x)
\end{equation}
We explored this setting in a previous work \cite{dutkiewiczError2025}, further relaxing the assumption that $u(x)$ is known and defining a noise-unbiased maximum-likelihood estimator based on probabilistic error cancellation circuits.

To extend these techniques to the case of a more complicated spectral distribution, we propose to give up exactly modelling $P(x)$.
Instead, we fit the same single-phase model $Q(x|\phi)$ to $P(x)$ only within a promise interval $\mathcal{D}$ (Def.~\ref{def:promise_interval}), where we know that $f(x-\phi_0)$ is the main contributor to $P(x)$.
The estimator we choose is the maximiser of the likelihood of the model $Q(x|\phi)$. However, as we do not expect $Q(x|\phi_0) = P(x)$ exactly this is not a canonical \emph{maximum-likelihood estimator}.
Instead, this estimator is known in information geometry \cite{amariMethods2000} and machine learning \cite{murphyMachine2012} as \emph{moment projection}, \emph{M-projection} or \emph{reverse-KL minimization} \cite{nielsenWHAT2018,tuananhleReverse2017}.
These names derive from the observation that $Q(x|\phi_0)$ is an orthogonal projection $P(x)$ onto the manifold defined by $Q(x|\phi)$, in a geometry defined by the reverse Kullback-Leibler (KL) divergence $D_\text{KL}(P(x) \lVert Q(x|\phi)) = \int P(x) \log\frac{P(x)}{Q(x|\phi)} \dd x$.
We use the term ``moment projection'' going forward

\begin{definition}[Moment projection estimator]
\label{def:m-projection-estimator}
    Let $Q(x|\phi)$ be a model distribution parametrized by $\phi\in\mathcal{D}_{\phi}$, and let $\{x_j\}_{j=1,...,M}$ be $M$ independent samples distributed according to $P(x)$. 
    The moment projection estimator is defined as
    \begin{align}
        \tilde\phi &= \argmax_{\phi\in\mathcal{D}_{\phi}} \ell(\phi | \{x_j\})\\
        \ell(\phi | \{x_j\}) &=\frac{1}{M}\sum_j \log Q(x_j|\phi)
    \end{align}
\end{definition}

We want to apply this estimator to the case where $P = Q(x|\phi_0) + h(x)$, with $|h(x)|$ sufficiently small.
The discrepancy between $Q(x|\phi_0)$ and $P(x)$ will result in a bias in the estimator $\tilde{\phi}$ of $\phi_0$, i.e.~$\phi_* := \lim_{M\to\infty}\E[\tilde{\phi}]\neq \phi_0$.
We want to study this in the $M\rightarrow\infty$ limit; let us first obtain functional forms for $\phi_*$ and the variance $\epsilon^2/M$ of the $M$-projection estimator for an arbitrary family of distributions $Q(x|\phi)$.

\begin{lemma}[Asymptotic distribution of the moment projection estimator]
\label{lem:m-projection-convergence}
Let $Q(x|\phi)$ be a model distribution on $\mathcal{D}$ parametrized by 
$\phi\in\mathcal{D}_{\phi}\subset\mathbb{R}$.
Let $\tilde\phi$ be the moment projection estimator defined in
Def.~\ref{def:m-projection-estimator}, using $M$ independent samples 
$x\in\mathcal{D}$ drawn according to $P(x)$.
Assume that $\log Q(x|\phi)$ is twice continuously differentiable in $\phi$, 
that the expectations of its first and second derivatives exist under $P(x)$, 
and that the minimizer of $D_{\mathrm{KL}}(P(x)||Q(x|\phi))$ is unique and lies in the interior of $\mathcal{D}_\phi$.
Then, in the limit $M\rightarrow\infty$, the estimator is asymptotically Gaussian:
the random variable $\sqrt{M}(\tilde\phi-\phi_*)$ converges in distribution to 
$\mathcal{N}(0,\epsilon_*^2)$, where
\begin{align}
\phi_*  &= \lim_{M\to\infty} \E[\tilde{\phi}] = \argmin_\phi D_{\mathrm{KL}}(P(x)||Q(x|\phi))
      = \argmax_\phi \int_{\mathcal{D}}\dd  x\, P(x)[\log Q(x|\phi)],\\
\epsilon_*^2  &= \lim_{M\to\infty} M \Var[\tilde\phi] =
\frac{
\int_{\mathcal{D}}\dd  x\, P(x)[(\partial_\phi \log Q(x|\phi_*))^2]
}{
\left(
\int_{\mathcal{D}}\dd  x\, P(x)[\partial_\phi^2 \log Q(x|\phi_*)]
\right)^2
}.
\end{align}
\end{lemma}
This result is a relatively standard application of the central limit theorem; we prove this for completeness in App.~\ref{app:m-projection-proof}.

The main result of this section is a bound on the closeness of the mean $\phi_*$ and variance $\sigma$ of the moment projection estimator in the case where $P(x) = Q(x|\phi_0) + h(x)$ for some ``target'' $\phi_0$, with $\|h\|_1 = \int_\mathcal{D} |h(x)| \dd x$ sufficiently small.
In the phase estimation case, $\phi_0$ is the underlying phase we are trying to estimate, and the bias $\phi_* - \phi_0 \neq 0$ emerges from our incomplete modelling of the target distribution $P(x) \neq Q(x|\phi_0)$ [Def.~\ref{def:filtered_distribution}]. %, which may be partially optimized by biasing $\phi_*$.
The resulting bias can be bounded proportionally to the norm of the model error $\lVert h \rVert$,
whilst the variance can be linked back to the Fisher information of $Q(x|\phi_0)$.

\begin{lemma}
[Moment projection estimator for distributions close to the model]
\label{lem:m-projection-expansion}

Under the same assumptions as in Lemma~\ref{lem:m-projection-convergence},
let
$\phi_0$ be a hidden target parameter such that
\begin{equation}
P(x)=Q(x|\phi_0)+h(x),
\end{equation}
where the deviation $h$ is small in the $1$-norm
\begin{equation}
\|h\|_1=\int_{\mathcal D}|h(x)|\dd x.
\end{equation}
Assume furthermore that
$\partial_\phi^2D_{\mathrm{KL}}(P(x)||Q(x|\phi))>0$
for $\phi\in[\phi_*,\phi_0]$.
Then for $M\to\infty$, the asymptotic bias $|\phi_0 - \phi_*|$, and the variance $\epsilon_*^2/M$ in Lemma~\ref{lem:m-projection-convergence} satisfy
\begin{align}
%|\phi_*-\phi_0|
|\phi_0 - \phi_*|&=\left|\frac{
\int_{\mathcal{D}}\dd  x\, h(x) [\partial_\phi\log Q(x|\phi)]_{\phi=\phi_0}
}{
\mathcal I_0
}\right|
+O(\|h\|_1^2)\label{eq:m-proj-1st-order-bias}\\
&\le
\|h\|_1
\frac{
\max_{x\in\mathcal{D}} |[\partial_\phi\log Q(x|\phi)]_{\phi=\phi_0}|
}{
\mathcal I_0
}
+O(\|h\|_1^2),\label{eq:mproj-bias}
\\
\epsilon_*^2
&=
\frac{1}{\mathcal I_0}+O(\|h\|_1),\label{eq:mproj-var}
\end{align}
where
\begin{equation}
\label{eq:fisher-info}
\mathcal I_0
=
\int_{\mathcal D}\dd x \,
Q(x|\phi_0)
\left[\partial_\phi\log Q(x|\phi)\right]^2_{\phi=\phi_0}
\end{equation}
is the Fisher information of $Q(x|\phi_0)$.
\end{lemma}

We prove this lemma in appendix~\ref{app:m-projection-proof}.

\noindent In practice, the assumptions of this theorem are satisfied if $Q(x|\phi)$ is a reasonable model for $P(x)$.
This in turn requires that:
\begin{enumerate}
\item $Q(x|\phi)$ is close to $P(x)$ in a single region around $\phi_0$, and for values of $\phi$ far from $\phi_0$ the distributions are very different [this ensures a well-defined global minimum of $D_{\mathrm{KL}}(P(x)||Q(x|\phi))$]
\item $Q(x|\phi)$ is smooth [this ensures that $D_{\mathrm{KL}}(P(x)||Q(x|\phi))$ is convex around in a finite region around the optimum $\phi_*$]
\item $\lVert h \rVert$ is small enough [this ensures $\phi_0$ and $\phi_*$ are close enough, and both contained in the convex region $\mathcal{D}_\phi$].
\end{enumerate}

\subsection{Noise unbiasing for moment projection}
\label{sec:explicit-unbiasing}

The moment projection estimator requires a model $Q(x|\phi)$ of the output distribution; in the case of general noise, this model is not efficiently computable.
One possibility is to approximate the effect of general noise with a simplified model, e.g.~one that assumes global depolarizing noise.
This will generally result in an estimation error (bias) due to the incorrect modelling of the noise.
Though this error may not be terribly large, we desire a method that can provably remove the bias from noise in the asymptotic resource limit.
For expectation value estimation, this is achievable via probabilistic error cancellation (PEC)~\cite{temmeError2017, endoPractical2018}, which expands the target expectation value as a linear combination of expectation values that can be estimated on a noisy device.
In Ref.~\cite{dutkiewiczError2025}, we extended the PEC approach to QFT-based phase estimation, by constructing a noise-unbiased maximum likelihood estimator (EUMLE).
The EUMLE writes the output distribution of the inaccesible (noiseless) QPE circuit as a quasi-probabilistic sum $P(x)=\sum_{a=0}^{r-1} \alpha_a P_a(x)$, where the $P_a$ are the output distributions of accessible (noisy) circuits.
From this, one can derive a likelihood function to optimize over, giving an estimate of the phase that is asymptotically bias-free whenever the noisy decomposition is correct.
In this section, we extend this result to a result that uses the full moment projection estimator.

Care needs to be taken when implementing the noise-unbiasing procedure, as sampling from the distributions $P_a$ can yield data that has near zero probability to be sampled from $Q(x|\phi_0)$.
This causes any estimation of $\phi_0$ to be dominated by the cost of obtaining sufficient samples to cancel this effect out; the variance of the resulting estimator becomes decoupled from the width $\sigma$ of the kernel function.
This was circumvented in Ref.~\cite{dutkiewiczError2025} by regularization; one adds a small spurious constant to $Q(x|\phi_0)\rightarrow Q_c(x|\phi_0)= Q(x|\phi_0)+c$.
This works when our range of estimation covers the full circle, however here the smaller interval $\mathcal{D}$ reintroduces a bias.
To solve this problem, we could consider adding a small amount of uniform noise to the data itself $P\rightarrow P _c(x)= P+c$, so that this regularization term properly models the data~\footnote{$ Q_c(x_j|\phi)$ and $P_c(x)$ are not normalized; we add the effect of the regularization directly to the likelihood rather than sampling from either distribution, so normalizing these functions is unnecessary.}.
However, this noisy distribution comes with a large variance term.
To circumvent this, instead of randomly sampling the noise, we can add the expected contribution to the likelihood itself.

\begin{definition}[Noise-unbiased M-Projection Estimator (NME)]\label{def:NME}
Let $\alpha \in \RR^r$ be a vector of real coefficients and let 
$\{P_a(x)\}_{a=1}^r$ be probability distributions on $\mathcal{D} \subset \RR$.
Let $Q(x|\phi)$ be a model distribution parametrized by $\phi\in\mathcal{D}_{\phi}$.
Fix a regularization constant $c\geq 0$, and define $ Q_c(x_j|\phi)=Q(x|\phi)+c$
Let $\{(x_j,a_j)\}_{j=1,...,M}$ be $M$ independent samples generated by first sampling
$a_j\in\{1,...,r\}$ with probability $P(a_j = a)=|\alpha_{a}|/\|\alpha\|_1$, and then sampling
$x_j$ from $P_{a_j}(x)$. The noise-unbiased moment projection estimator estimates
\begin{align}
\tilde\phi &= \argmax_{\phi\in\mathcal{D}_{\phi}} \ell(\phi|\{x_j,a_j\}),\\
\ell(\phi|\{x_j,a_j\}) &=
\frac{\|\alpha\|_1}{M}\sum_{j=1}^M
\mathrm{sgn}(\alpha_{a_j})\log [ Q_c(x_j|\phi)]+c\int_{\mathcal{D}}\dd  x\log[ Q_c(x|\phi)].
\end{align}
\end{definition}
As one might expect, the NME is not significantly different from the EUMLE of Ref.~\cite{dutkiewiczError2025}, and one can derive similar results to that work.
\begin{lemma}
[Asymptotic distribution of NME]
\label{lem:nme-convergence}
Let $Q(x|\phi)$ be a model distribution on $\mathcal{D}$ parametrized by 
$\phi\in\mathcal{D}_{\phi}\subset\mathbb{R}$  and $\{P_{a}(x)\}_{a=1}^r$ be sampleable distributions on $\mathcal{D}$. Fix a regularization constant $c > 0$, and define the regularized distributions $ Q_c(x|\phi)=Q(x|\phi)+c$, $P_c(x)=\sum_{a=0}^{r-1} \alpha_a P_a(x)+c$.
Let $\tilde\phi$ be the estimator defined in 
Def.~\ref{def:NME}, using $M$ independent samples.
Assume that $\log Q_c(x|\phi)$ is twice continuously differentiable in $\phi$, that the expectations of its derivatives exist under each distribution $P_a$, and that the minimizer of $D_{\mathrm{KL}}(P_c(x)|| Q_c(x|\phi))$ is unique and lies in the interior of $\mathcal{D}_\phi$.
Then, in the limit $M\rightarrow\infty$, the estimator is asymptotically Gaussian: $\sqrt{M}(\tilde\phi-\phi_*)$ converges in distribution to $\mathcal{N}(0,\epsilon_*^2)$, where
\begin{align}
\phi_* &= \argmax_\phi \int_{\mathcal{D}}\dd  x\, P_c(x)[\log  Q_c(x|\phi)],\\
\epsilon_*^2 &\leq
\lVert \alpha \rVert_1^2 \frac{
\int_{\mathcal{D}}\dd  x\, R(x)
[\partial_\phi \log  Q_c(x|\phi)]_{\phi = \phi_*}^2
}{
\left(
\int_{\mathcal{D}}\dd  x\, P_c(x)[\partial_\phi^2 \log  Q_c(x|\phi)]_{\phi = \phi_*}
\right)^2
}. \label{eq:eume-variance}
\end{align}
Here, $R(x)$ is the marginal distribution $R(x) = \frac{1}{\lVert \alpha \rVert_1} \sum_{a = 1}^r|\alpha_a|P_a(x)$.
\end{lemma}

\begin{lemma}
[NME for distributions close to the model]
\label{lem:nme-expansion}
Under the same assumptions as in Lemma~\ref{lem:nme-convergence}, let $\phi_0$ be a hidden target parameter such that
\begin{equation}
\label{eq:h-nme}
\sum_{a=0}^{r-1} \alpha_a P_a(x)=Q(x|\phi_0)+h(x),
\end{equation}
where the deviation $h$ is small in the $1$-norm
\begin{equation}
\|h\|_1=\int_{\mathcal D}|h(x)|\dd x.
\end{equation}
Assume further that $\partial_\phi^2D_{\mathrm{KL}}(P(x)||Q(x|\phi))>0$ for $\phi\in[\phi_*,\phi_0]$.
Then the asymptotic bias $|\phi_*-\phi_0|$ and variance $\epsilon_*^2/M$ in Lemma~\ref{lem:nme-convergence} satisfy
\begin{align}\label{eq:mle-bias}
|\phi_*-\phi_0|
&\le
\|h\|_1
\frac{
\max_{x\in\mathcal{D}} [\partial_\phi\log  Q_c(x|\phi)]_{\phi=\phi_0}
}{
 \mathcal{I}_c
}
+O(\|h\|_1^2),
\\
\epsilon_*^2
&\le
\|\alpha\|_1^2
\left(\frac{
\max_{x\in\mathcal D}
[\partial_\phi\log  Q_c(x|\phi)]_{\phi=\phi_0}
}{
\mathcal{I}_c
}\right)^2
+O(\|h\|_1),
\end{align}
where
\begin{equation}
\label{eq:fisher-info-nme}
\mathcal{I}_c
=
\int_{\mathcal D} \dd x \,
Q_c(x|\phi_0)
\left[\partial_\phi\log  Q_c(x|\phi)\right]^2_{\phi=\phi_0},
\end{equation}
is the Fisher information of $Q_c(x|\phi_0)$ [Eq.~\eqref{eq:fisher-info}].
\end{lemma}
We prove Lemma \ref{lem:nme-convergence} and Lemma~\ref{lem:nme-expansion} in Appendix~\ref{app:m-projection-proof}. 

\noindent While the EUMLE in Ref.~\cite{dutkiewiczError2025} was compatible with PEC alone, NME is compatible with a larger class of error mitigation methods.
This is because EUMLE required the quasiprobability to exactly match the model, i.e. $\sum_{a=0}^{r-1} \alpha_a P_a(x) = Q(x|\phi_0)$, and therefore the quasiprobability distribution needed to reconstruct the noiseless probability $P_0(x)$ exactly.
Instead, since NME allows for a small mismatch between the quasiprobability $\sum_{a=0}^{r-1} \alpha_a P_a(x)$ and the model $Q(x|\phi_0)$, it suffices to have an approximate linear decomposition $\sum_{a=0}^{r-1} \alpha_a P_a(x) \approx P_0(x)$. This can be achieved by any linear quantum error mitigation technique, including PEC but also methods such as zero-noise extrapolation or symmetry verification \cite{cai2021practical}, which reconstruct mitigated quantities as linear combinations of measurements from noisy circuits.
This can be treated as an additional contribution to $h$.
The asymptotic bias in Eq.~\eqref{eq:mle-bias} will then contain contributions from both the modelling error (the mismatch between the noiseless distribution $P_0$ and $Q(x|\phi_0)$) and the imperfect $\sum_{a=0}^{r-1} \alpha_a P_a$ reconstruction of $P_0$ produced by the QEM procedure.
We can separate these contributions to the error via the triangle inequality,
\begin{equation}
    \|h(x)\|_1 = \|\sum_{a=0}^{r-1} \alpha_a P_a(x) -Q(x|\phi_0)\| \leq \|\sum_{a=0}^{r-1} \alpha_a P_a(x) -P_0(x)\| + \| P_0(x) -Q(x|\phi_0)\|.
\end{equation}

\subsection{Moment projection for phase estimation}
\label{sec:m-projection-for-qpe}

In the previous section we focused on how the error in fitting a distribution with an imperfect parametric model propagates to a bias in the estimate of the parameter.
In this section, we apply these tools to QPE.
First, in Lemma \ref{lem:generating-filtered-samples} we prove that the underlying sampling input (Def.~\ref{def:classical-subroutine}) can be generated by QFT-based QPE techniques, and that we can construct an appropriate promise interval $\mathcal{D}$ under standard assumptions.
Then, in Lemma \ref{lem:generating-filtered-samples-pec}, we extend this procedure to show we can generate samples for the noise-unbiased estimator (Def.~\ref{def:NME}) under the same assumptions, given the ability to sample from the PEC-decomposition circuits for the same QFT-based QPE routine.
Finally, we define estimators for noiseless and noisy phase estimation, to which  we apply the results of the previous section.

\begin{lemma}[Generating filtered QPE samples]
\label{lem:generating-filtered-samples}
Let $U$ be a unitary with target phase $\phi_0$.
Assume an initial guess $\phi_{guess}$ of $\phi_0$ accurate up to $\Delta/3$, for $\Delta\geq \min_{j\neq 0}|\phi_j-\phi_0|$.
Then, given the ability to sample from $p(x)$ (Def.~\ref{def:noisy-distribution}) for $U$, a state $\ket{\psi}$ with ground state overlap $a_0 > 0$, and a kernel function $f$ using $t$ calls to a circuit implementation of $U$ per sample, one can construct a promise interval $\mathcal{D}$ and generate $M$ samples from $P(x)$ that satisfies the conditions of Def.~\ref{def:classical-subroutine} (i.e. one can execute the quantum subroutine), using on average $Mt/P_{A}$ calls to a circuit implementation of $U$, where $P_A$ is a lower bound on the probability of accepting a sample $P_{accept}$ given by
\begin{equation}
\label{eq:p_accept}
    P_A = Fa_0\int_{\phi_{guess}-\Delta/2}^{\phi_{guess}+\Delta/2}f(x-\phi_0)\dd x  +(1-F) \frac{|\mathcal{D}|}{2\pi}
    \geq
    Fa_0\int_{-\Delta/6}^{\Delta/6}f(x)\dd x  +(1-F) \frac{|\mathcal{D}|}{2\pi} 
\end{equation}
\end{lemma}
\begin{proof}
    First, we will construct a promise interval that satisfies the properties in Def.~\ref{def:promise_interval}.
    Let $\mathcal{D}$ be an interval of size $|\mathcal{D}| = \Delta$ centered around $\phi_{guess}$, i.e. $\mathcal{D} = [\phi_{guess}-\Delta/2, \phi_{guess}+\Delta/2]$.
    By assumption $|\phi_0-\phi_{guess}|<\Delta/3$, and the first condition in Def.~\ref{def:promise_interval} is satisfied for an inner buffer $d_{\text{in}} = \Delta/6$.
    Using triangle inequality, and the assumption that $\forall_{j>0} |\phi_j - \phi_0| > \Delta$,
    we have
    \begin{equation}
        |\phi_j - \phi_{guess}| \geq |\phi_j-\phi_0| - |\phi_0 - \phi_{guess}| \geq \frac{2}{3}\Delta
    \end{equation}
    Therefore the second condition in Def.~\ref{def:promise_interval} is satisfied with an outer buffer $d_{\text{out}} =\min_{j\neq 0}|\phi_j-\phi_{guess}|-|\mathcal{D}|/2 \geq \Delta/6$.

    To obtain $M$ samples within the promise interval $\mathcal{D}$, we need to get $M' \geq M$ samples from $p(x)$.
    We can write probability of accepting each sample as
    \begin{equation}
      P_{accept} = \int_\mathcal{D} \dd x \,\ \left[F (a*f)(x) +(1-F) \frac{1}{2\pi} \right]
      = F \left[ \sum_j a_j \int_\mathcal{D} f(x-\phi_j) \dd x \right] +(1-F) \frac{|\mathcal{D}|}{2\pi}.
\end{equation}
Since $f$ is positive, we can bound it from below by neglecting the contributions of $\phi_{j\neq0}$ as
\begin{equation}
    P_{accept} \geq Fa_0 \int_\mathcal{D} f(x-\phi_0) \dd x +(1-F) \frac{|\mathcal{D}|}{2\pi}.
\end{equation}
Again using positivity of $f$, we can bound the integral in the expression above by an integral over a subset $[\phi_0-\Delta/6, \phi_0+\Delta/6] \subset \mathcal{D}$, yielding the desired bound.
\end{proof}

While Lemma~\ref{lem:generating-filtered-samples} establishes the sampling procedure required by the M-projection estimator, the noise-unbiased estimator (Definition~\ref{def:NME}) requires samples from a quasiprobability decomposition of the filtered distribution. The following lemma shows that such samples can be generated directly from a quasiprobability decomposition of the noisy QPE distribution by the same rejection-sampling procedure.

\begin{lemma}[Generating filtered QPE samples with PEC]
\label{lem:generating-filtered-samples-pec}
Let $U$ be a unitary with target phase $\phi_0$.
Assume an initial guess $\phi_{guess}$ of $\phi_0$ accurate up to $\Delta/3$, for $\Delta\geq \min_{j\neq 0}|\phi_j-\phi_0|$.
 Assume
a quasiprobability decomposition $p(x) = \sum_{a}\beta_a p_a(x)$ of the distribution $p(x)$ (Def.~\ref{def:noisy-distribution}), for $U$, a state $\ket{\psi}$ with ground state overlap $a_0 > 0$, and a kernel function $f$.
Then, given the ability to sample from the quasiprobability distibution, i.e. sample $a\sim\frac{|\beta_a|}{\|\beta\|_1}$ and sample from each $p_a(x)$ using $t$ calls to a circuit implementation of $U$ per sample, one can construct a promise interval $\mathcal{D}$ and generate $M$ samples from a quasiprobability decomposition $P(x) = \sum_{a}\alpha_a P_a(x)$
of $P(x)$ that satisfies the conditions of Def.~\ref{def:classical-subroutine} (i.e. one can execute the quantum subroutine), using on average $Mt/P'_{A}$ calls to a circuit implementation of $U$, where $P'_A$ is a lower bound on the probability of accepting a sample $P_{accept}$ given by
\begin{equation}
\label{eq:p_accept_pec}
    P'_A = \frac{\|\alpha\|_1}{\|\beta\|_1}P_A \geq \frac{1}{\|\beta\|_1}P_A.
\end{equation}
The weights $\alpha$ in the quasiprobability decomposition are given by
\begin{align}
    \label{eq:filtered-weights}
    \alpha_a &= \beta_a \frac{\int_\mathcal{D} p_a(x) \dd x }{\int_\mathcal{D} p(x) \dd x }
\end{align}
and their norm satisfies
\begin{align}\label{eq:filtered-weights-norm}
    \|\alpha\|_1 &=\sum_{a=0}^{r-1} |\beta_a| \frac{\int_\mathcal{D} p_a(x) \dd x }{\int_\mathcal{D} p(x) \dd x } \leq 1+\kappa\frac{\|\beta\|_1}{P_A},
\end{align}
where $P_A$ is the lower bound on the probability of acceptance for the target probability $p(x)$ given in Eq.~\eqref{eq:p_accept}, and $\kappa = \|\beta\|_1^{-1}\int_{\mathcal{D}}(\sum_{a=0}^{r-1}|\beta_a|p_a(x) - p(x))\dd x \in [0,1)$ is the excess sampling probability inside the promise interval.
\end{lemma}
\begin{proof}
    We can construct the same promise interval $\mathcal{D}$ as in Lemma~\ref{lem:generating-filtered-samples}.
    
    First, we prove that the filtered target probability $P(x)$ admits the quasiprobability decomposition $P(x) = \sum_{a=0}^{r-1} \alpha_a P_a(x)$ with the weights defined in Eq.~\eqref{eq:filtered-weights} and that we can sample from it.
    Let $P_a$ be the distributions of  filtered samples from $p_a$, i.e.
    \begin{equation}
        P_a(x) = \begin{cases}
            \frac{p_a(x)}{\int_\mathcal{D}p_a(x)}, & x\in\mathcal{D}\\
            0, & x \notin \mathcal{D}.
        \end{cases}
    \end{equation}
    Then, for any $x\in \mathcal{D}$ we have
    \begin{equation}
        \sum_{a=0}^{r-1} \alpha_a P_a(x)=\frac{\sum_{a=0}^{r-1} \beta_a  p_a(x)}{\int_\mathcal{D}
        p(x) \dd x}=\frac{p(x)}{\int_\mathcal{D}p(x) \dd x}=P(x)
    \end{equation}
    which is the desired quasiprobability decomposition.
    We generate samples from this decomposition by first sampling $a\sim\frac{|\beta_a|}{\|\beta\|_1}$, then sampling $x\sim p_a(x)$, and accepting if $x\in \mathcal D$.
    The overall acceptance probability is
\begin{align}
        P_{accept} =  \frac{1}{\|\beta\|_1}\sum_{a=0}^{r-1} |\beta_a|\int_\mathcal{D} p_a(x) \dd x = \frac{1}{\|\beta\|_1}\sum_{a=0}^{r-1} |\alpha_a|\int_\mathcal{D} p(x) \dd x = \frac{\|\alpha\|_1}{\|\beta\|_1} \int_\mathcal{D} p(x) \dd x.
\end{align}
Conditioned on acceptance, the joint distribution of $(a,x)$ is, for $x\in\mathcal{D}$,
\begin{align}
    \Pr(a,x|\mathrm{accepted})
    &=
    \frac{1}{P_{accept}}\frac{|\beta_a|}{\|\beta\|_1} p_a(x)
    =
    \frac{1}{\frac{\|\alpha\|_1}{\|\beta\|_1} \int_\mathcal{D} p(x) \dd x}\frac{|\beta_a|\int_\mathcal{D} p_a(x) \dd x}{\|\beta\|_1} \frac{p_a(x)}{\int_\mathcal{D} p_a(x) \dd x}
    =
    \frac{|\alpha_a|}
         {\|\alpha\|_1}
    P_a(x).
\end{align}
This enables the desired sampling procedure for the quasiprobability decomposition of $P(x)$, since the accepted samples satisfy $\Pr(a) = \frac{|\alpha_a|}
         {\|\alpha\|_1}$ and $\Pr(x|a) = P_a(x)$.

Finally, we prove the claimed inequalities.
Since the target probability $p(x)$ and the promise interval $\mathcal{D}$ are the same as in Lemma~\ref{lem:generating-filtered-samples}, we have $\int_\mathcal{D} p(x) \dd x \geq P_A$.
Substituting this bound into the expression for the acceptance probability gives
\begin{equation}
    P_{accept} \geq \frac{\|\alpha\|_1}{\|\beta\|_1}P_A = P'_A.
\end{equation}
    The lower bound $\|\alpha\|_1 \geq 1$ used in Eq.~\eqref{eq:p_accept_pec} is a straightforward consequence of the normalisation of distributions $p(x)$, $p_a(x)$.
    To prove the desired upper bound on the norm $\|\alpha\|_1$, we substitute in the definition of $\kappa$,
    \begin{align}
        \|\alpha\|_1 = \frac{1}{\int_{\mathcal{D}}p(x)\dd x}\sum_{a=0}^{r-1}|\beta_a| \int_{\mathcal{D}}p_a(x) \dd x
        = 1 + \frac{\kappa \|\beta\|_1}{\int_{\mathcal{D}}p(x) \dd x}
    \end{align}
    and use $\int_{\mathcal{D}}p(x) \geq P_A$.
\end{proof}

The quantity $\kappa$ introduced above measures how much of the quasiprobability overhead remains after filtering, and therefore determines the error mitigation overhead in Theorem~\ref{thm:fnmpe-cost-pauli-informal}.
It depends both on the quasiprobability distribution and the location of the promise interval $\mathcal{D}$.
Since
\begin{align}
    \kappa =
    \frac{1}{\|\beta\|_1}\int_{\mathcal{D}}\left(\sum_{a=0}^{r-1}|\beta_a|p_a(x)-p(x)\right)\dd x
    =
    \frac{1}{\|\beta\|_1}\sum_{a=0}^{r-1}(|\beta_a|-\beta_a)\int_{\mathcal{D}}p_a(x)\dd x,
\end{align}
we can equivalently interpret $\kappa$ as the average probability mass inside $\mathcal{D}$ contributed by the quasiprobability terms with negative coefficients.
To gain intuition for the meaning of $\kappa$, we analyze two limiting cases and the example of global depolarising noise discussed earlier.
%(1)
First, suppose all the negative contributions are entirely supported outside of $\mathcal{D}$.
Then $\kappa = 0$ and $\|\alpha\|_1 = 1$, so filtering completely removes the quasiprobability overhead, at the expense of the lowest acceptance probability.
%(2)
Conversely, if all distributions $p_a(x)$ are fully supported on $\mathcal{D}$, then every sample is accepted, but $\kappa = 1-\|\beta\|_1^{-1}$ and $\|\alpha\|_1 = \|\beta\|_1$.
In this case filtering introduces no rejection overhead, but also provides no reduction of the quasiprobability overhead.
%(3)
Finally, for global depolarising noise, the negative component is uniform %($p_1(x) = \frac{1}{2\pi}$, $p_0(x) = F p(x) + (1-F)p_1(x)$).
and $\|\alpha\|_1 = 1 + O(\|\beta\|_1 \frac{|\mathcal{D}|}{2\pi})$,
so the quasiprobability overhead is substantially reduced when $|\mathcal{D}| \ll 2\pi$.

% \begin{align}
%     p_1(x) = \frac{1}{2\pi}\\
%     p_0(x) = F p(x) + (1-F)p_1(x)\\
%     \beta = [1/F, -(1/F-1)]\\
%     \|\beta\| = 2/F -1\\
%     \kappa = 2|\beta_1|/\|\beta\| |D|/2\pi = \frac{1-F}{1-F/2} |D|/2\pi\\
%     \alpha = [1 +(1/F-1)|D|/2\pi/\int p, -(1/F -1)|D|/2\pi/\int p]\\
%     \|\alpha\|=1+2(1/F)|D|/2\pi/\int p
% \end{align}

With Lemmas~\ref{lem:generating-filtered-samples} and \ref{lem:generating-filtered-samples-pec} in hand, we have established that the required sampling subroutine can be implemented in the ideal setting, in the presence of global depolarising noise, and under arbitrary noise using PEC.
We now construct estimators for phase estimation which satisfy the assumptions of Lemmas~\ref{lem:m-projection-expansion} and \ref{lem:nme-expansion}.
As described in Sec.~\ref{sec:background}, these depend on a choice of kernel function $f$ (Def.~\ref{def:kernel_function}); we leave this free for now, but will consider Gaussian kernels in the rest of this section (Sections~\ref{sec:gaussian_no_noise}, \ref{sec:gaussian_with_gdn}, \ref{sec:gaussian_unbiasing}).
In the noiseless case, the probability distribution $P(x)$ is a combination of contributions from different phases, which we approximate by a model $Q(x|\phi)$ which assumes a single phase.
The mismatch $P(x)-Q(x|\phi) = h(x)$ then comes from the signal due to residual phases; within the promise interval this contribution is small, which we investigate in the next sections.
Here, we introduce the models $Q(x|\phi)$ and related moment projection phase estimators, with different assumptions about noise. Starting from the noiseless case:

\begin{definition}
\label{def:m-projection-phase-estimator}
    [Filtered Moment projection Phase Estimator (FMPE), noiseless case] 
    Let $f$ be a kernel function (Def.~\ref{def:kernel_function}),
    $\mathcal{D}$ be a promise interval (Def.~\ref{def:promise_interval}),
    $p(x)$ be the distribution in Def.~\ref{def:noisy-distribution} with kernel function $f$, fidelity $F =1$ and ground state overlap $a_0$,
    and $P(x)$ be the corresponding filtered distribution on $\mathcal{D}$ (Def.~\ref{def:filtered_distribution}).
    In the noiseless case, the FMPE is moment projection estimator (Def.~\ref{def:m-projection-estimator}) with samples $x\sim P(x)$ and a model distribution
    \begin{align} \label{eq:noiseless-model}
        Q(x|\phi) = \frac{f(x-\phi)}{\int_{\mathcal{D}}f(x-\phi)\dd x \,}.
    \end{align}
\end{definition}

If instead the circuit is affected by global depolarising noise, we can model it exactly by adding a constant noise level of $(1-F)\frac{1}{2\pi}$ to the model for $p(x)$ and normalizing properly after filtering.
This yields a new estimator, where again the mismatch $h$ only comes from the spurious phases:
\begin{definition}
\label{def:m-projection-phase-estimator-with-gdn}
    [FMPE, assuming global depolarising noise] 
    Let $f$ be a kernel function (Def.~\ref{def:kernel_function}),
    $\mathcal{D}$ be a promise interval (Def.~\ref{def:promise_interval}),
    $p(x)$ be a noisy distribution in Def.~\ref{def:noisy-distribution} with kernel function $f$, fidelity $F$ and ground state overlap $a_0$,
    and $P(x)$ be the corresponding filtered distribution on $\mathcal{D}$ (Def.~\ref{def:filtered_distribution}).
    Assuming global depolarising noise, FMPE is moment projection estimator (Def.~\ref{def:m-projection-estimator}) with samples $x\sim P(x)$ and a model distribution
    \begin{align} \label{eq:qphi_with_noise}
        Q(x|\phi) = \frac{Fa_0f(x-\phi) +(1-F)\frac{1}{2\pi}}{F a_0\int_{\mathcal{D}}f(x-\phi)\dd x + (1-F)\frac{|\mathcal{D}|}{2\pi}}.
    \end{align}
\end{definition}

In Figure~\ref{fig:m-projection} we give a schematic representation of this estimator.

\noindent The case of general noise cannot be simply modelled by adding a term in $Q(x|\phi)$, as modelling general noise would amount to simulating the full quantum circuit. 
Instead, we use the noise-unbiased moment projection estimator introduced above:

\begin{definition}
\label{def:m-projection-phase-estimator-unbiased}
    [Noise-Unbiased Filtered Moment projection Phase Estimator (NU-FMPE)] 
    Let $f$ be a kernel function (Def.~\ref{def:kernel_function}),
    $\mathcal{D}$ be a promise interval (Def.~\ref{def:promise_interval}).
    Let the coefficients $\alpha \in \RR^r$ and probability distributions on $\mathcal{D}$ $\{P_a(x)\}_{a=1}^r$ be such that $P(x) = \sum_{a=0}^{r-1} \alpha_a P_a(x)$ is the distribution in Def.~\ref{def:filtered_distribution} with kernel function $f$, fidelity $F=1$ and ground state overlap $a_0$.
    The NU-FMPE is NME (Def.~\ref{def:NME}) with
    samples $(a,x)\sim \frac{|\alpha_a|}{\|\alpha\|_1}P_a(x)$ and
    the noiseless model distribution from Eq.~\eqref{eq:noiseless-model}.
    % \begin{align}
    %     Q(x|\phi) = \frac{f(x-\phi)}{\int_{\mathcal{D}}f(x-\phi)\dd x}.
    % \end{align}
\end{definition}

\begin{figure}
    \centering
    \includegraphics[width=0.55\linewidth]{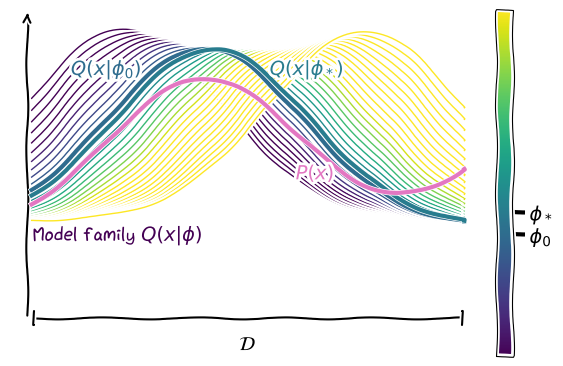}
    \caption{Schematic representation of the FMPE (Def.~\ref{def:m-projection-phase-estimator-with-gdn}): among a family of models $Q(x|\phi)$ [Eq.~\eqref{eq:qphi_with_noise}, colored as per the colorbar based on the value of the parameter $\phi$] supported on $\mathcal{D}$, the one that minimizes the inverse KL divergence with $P(x)$ (Def.~\ref{def:filtered_distribution}) is $Q(x|\phi_*)$.
    The target distribution $P(x)$ is the final result of the filtered-QPE scheme represented in Fig.~\ref{fig:scheme}.
    We also highlight the model distribution for the ideal value of $\phi=\phi_0$: $Q(x|\phi_0)$ -- the difference between $\phi_0$ and the optimal parameter $\phi_*$ is the bias of the FMPE.
    }
    \label{fig:m-projection}
\end{figure}

\subsection{Gaussian kernel and no noise}
\label{sec:gaussian_no_noise}

In order to apply Lemma~\ref{lem:m-projection-expansion} to a real phase estimation problem, we need to define a noise model and a kernel function.
We first consider the noiseless setting, and choose a Gaussian kernel.
This allows us to directly compare our moment projection estimator to the sample mean estimator of Ref.~\cite{wang2024faster}.
For simplicity of notation, we use the following shorthand for the Gaussian and its integral on the filtering interval $\mathcal{D}$
\begin{align}
g_\sigma(x) & = \frac{1}{\sqrt{2\pi}\sigma} e^{-\frac{x^2}{2\sigma^2}},\\
G_\sigma(\phi) & = \int_\mathcal{D} g_\sigma(x-\phi)\dd x.
\end{align}

\topic{explicit difference between estimators}
Let us first give some intuition into how the M-projector differs from a simple mean estimator $\bar{x} = M^{-1}\sum_j x_j$.
In the case where the domain of the filtering is $[-\infty, \infty]$, the sample mean estimator matches the moment projection estimator with $Q(x|\phi) = g_\sigma(x-\phi)$.
Taking a finite filtering interval $\mathcal{D}$,
\begin{equation}
\label{eq:gaussian-prop-dist-no-noise}
    P(x) = \frac{\sum_j a_j g_\sigma(x-\phi_j)}{\sum_j a_j G_\sigma(x-\phi_j)},
\end{equation}
modifies the model for the moment projection estimator to
\begin{equation}
\label{eq:gaussian-model-no-noise}
    Q(x|\phi) = \frac{g_\sigma(x-\phi)}{G_\sigma(\phi)} \text{ for } x\in \mathcal{D}.
\end{equation}
The objective function we have to maximize (log-likelihood $\ell$ of the model) is then
\begin{equation}
    \ell(\phi |\{x_j\})
     = -\frac{1}{M}\sum_j\frac{(x_j-\phi)^2}{2\sigma^2} 
     - \log G_\sigma(\phi)
\end{equation}
Maximising this leads to
\begin{equation}
\label{eq:estimator-vs-mean-gaussian-no-noise}
    \tilde\phi = 
    \bar{x} + \sigma^2\frac{g_\sigma(\tilde\phi-\phi_{guess}-\frac{|\mathcal{D}|}{2})-g_\sigma(\tilde\phi-\phi_{guess}+\frac{|\mathcal{D}|}{2})}{G_\sigma(\tilde\phi)}.
    %\sigma^2 \frac{\partial_\phi G_\sigma(\tilde\phi)}{G_\sigma(\tilde\phi)}
\end{equation} 
The second term is an additional correction due to the normalization term, which depends on $\phi$ but not on the samples.
This clarifies explicitly the difference between the mean estimator and the moment projection estimator for a filtered Gaussian model.

If the probability distribution of samples matches the model $P(x) = Q(x|\phi_0)$ (i.e., for the case of QPE, if we do not have spurious phases), the moment projection estimator matches the maximum-likelihood estimator, which we know to be unbiased.
This implies the bias of the mean estimator in this case is precisely the second term of Eq.~\eqref{eq:estimator-vs-mean-gaussian-no-noise}.
This is exponentially small in $\sigma^{-1}$, as long as $\phi_0$ is contained in the promise interval with inner buffer (Def.~\ref{def:promise_interval}).

In the presence of spurious phases, the moment projection estimator also picks up a bias. 
We characterise the bias and variance of this estimator in the following lemma.
\begin{lemma} \label{lem:gaussian-no-noise}
    Consider the FMPE (Def.~\ref{def:m-projection-phase-estimator}) with a Gaussian kernel function $f_\sigma(x) \propto g_\sigma(x)$ (Def.~\ref{def:gaussian_kernel_function} with $\epsilon_{\text{synth}} = 0$) and ground state overlap $a_0$.
    Assume that the $\mathcal{D}$ is a promise interval (Def.~\ref{def:promise_interval}) with inner buffer $d_{\text{in}} = |\mathcal{D}|/6$ and outer buffer $d_{\text{out}} = d$, i.e. target phase $\phi_0$ falls within the filtering interval $\mathcal{D} = [\phi_{guess}-|\mathcal{D}|/2, \phi_{guess}+|\mathcal{D}|/2]$,
    \begin{equation}
    \label{eq:phi0_buffer-fixed}
        |\phi_0 - \phi_{guess}| \leq |\mathcal{D}|/3,
    \end{equation}
    and further that all $\phi_j$ for $j\neq0$ are sufficiently far from the filtering region
    \begin{equation}
    \label{eq:phi1_buffer-fixed}
        \min_{j\neq0} \max_{x\in \mathcal{D}}|\phi_j -x|\geq d.
    \end{equation}
    Then, for any $\sigma \leq |\mathcal{D}|/6$,
    the asymptotic bias and variance of the estimator in Lemma~\ref{lem:m-projection-convergence} satisfy
    \begin{align}
    \label{eq:nonoise-bias-bound}
    % |\phi-\phi_*| &\leq g_\sigma(d) \times \frac{1-a_0}{a_0} |\mathcal{D}|^2 \times C_b
    % % F_1\left(\frac{|\mathcal{D}|}{\sigma}\right)
    % + O(|\mathcal{D}|^2 g_\sigma(d)^2)\\
    %|\phi_0-\phi_*|
    \lim_{M\to\infty} |b|
    &\leq 4\sigma^2 g_\sigma(d) \times \frac{1-a_0}{a_0} \frac{|\mathcal{D}|}{d} + O\big( g^2_\sigma(d)\big),\\
    \lim_{M\to\infty} M \Var[\tilde\phi] &\leq 2\sigma^2 
    %\times F_2\left(\frac{|\mathcal{D}|}{\sigma}\right)
    + O\big(g_\sigma(d)\big).
    \label{eq:nonoise-var-bound}
    \end{align}
\end{lemma}
The proof of this theorem is given in Appendix~\ref{app:proof-of-gaussian-cors}. 
Our choice of the buffer between the ground state eigenphase and interval edges [Eq.~\eqref{eq:phi0_buffer-fixed}] is artificial; in principle this can be removed entirely without affecting the asymptotic scaling of our estimator with $\sigma$ and $a_0$.
The non-zero buffer between the spurious phases $\phi_j$ and the filtering interval $\mathcal{D}$ [Eq.~\eqref{eq:phi1_buffer-fixed}] is necessary however to accommodate the fact that our model distribution $Q(x|\phi)$ explicitly does not consider any additional phases; when $d=0$, the $M$-projection estimator remains biased for arbitrarily small $\sigma$.

\topic{comparison of biases}
In order to characterize the performance of the moment projection estimator beyond upper bounds and compare it to the mean estimator, we integrate numerically the first-order bias in Eq.~\eqref{eq:m-proj-1st-order-bias}, ignoring the $O(|\mathcal{D}|^2 g_\sigma(d)^2)$ correction.
We observe in Fig.~\ref{fig:no-noise} that the bias of the moment projection estimation is largely independent on the value of $\phi_0$; the bias is only due to the contribution of the spurious phases and decreases as $\phi_1$ is further from the interval $\mathcal{D}$, as predicted by the bound in Eq.~\eqref{eq:nonoise-bias-bound}.
In contrast, the mean estimator $\tilde\phi=\langle x_j\rangle_{x_j\in\mathcal{D}}$ picks up an additional bias that depends on the value of $\phi_0$ [Eq.~\eqref{eq:estimator-vs-mean-gaussian-no-noise}]. 
When this latter bias dominates, the moment projection estimator achieves a total bias that is exponentially smaller in $\sigma^{-1}$ than the bias of the mean estimator.
The improvement is especially evident for $\phi_1$ farther from the filtering interval.
The narrow region where the mean estimator has vanishing bias is due to a fortuitous cancellation of the positive bias coming from the spurious phases and the negative bias coming from the filtered-out samples from the ground phase.

At this point we have constructed an estimator that takes samples from a distribution $P(x)$ and fits a Gaussian model $Q(x|\phi)$ with a fixed distribution $\sigma$.
This considers only the classical subroutine of phase estimation (Def.~\ref{def:classical-subroutine}).
To connect this to the standard phase estimation literature, we extend this to a quantum algorithm
\begin{theorem}[cost of noiseless FMPE]
\label{thm:fmpe-cost-nonoise}
    Let $U$ be a unitary with spectral gap $\Delta$ around a target state $\phi_0$ ($\forall_{j>0} |\phi_j - \phi_0| > \Delta$).
    Assume oracle access to a controlled version of $U$, and an initial state $\ket{\psi}$ such that $|\braket{\phi_0}{\psi}|^2 > \eta$.
    Further assume an initial estimate $\phi_{guess}$ of $\phi_0$ such that $|\phi_0-\phi_{guess}|<\Delta/3$.
    Then, the FMPE (Def.~\ref{def:m-projection-phase-estimator}) using $t$ calls to $U$ per circuit produces an estimate $\tilde\phi$ with RMS error up to $\epsilon$ using $M' = O(\eta^{-1}t^{-2}\epsilon^{-2})$ samples and $T = O(\eta^{-1}t^{-1}\epsilon^{-2})$ total calls to $U$, as long as $t= {\Omega}(\Delta^{-1}\log^{1/2}(\Delta\epsilon^{-1}\eta^{-1}))$ and $M'=\Omega(\eta^{-1})$.
\end{theorem}
\begin{proof}
By Lemma~\ref{lem:gaussian_kernel_synthesis}, we can generate samples from $p(x)$ with a Gaussian kernel function $f_\sigma$ using $t = \Theta(\sigma^{-1})$ calls to $U$ [we ignore the overhead of $\log(\epsilon_{\text{synth}})$].
By Lemma~\ref{lem:generating-filtered-samples}, we can then generate samples from $P(x)$ with $\mathcal{D}$ that satisfies the assumptions of Lemma~\ref{lem:gaussian-no-noise} with $|\mathcal{D}|, d = \Theta(\Delta)$.
By Lemma~\ref{lem:gaussian-no-noise}, $t = {\Omega}(\Delta^{-1}\log^{1/2}(\Delta\epsilon^{-1}\eta^{-1}))$ is enough to ensure that the bias is $O(\epsilon)$.
To ensure that the variance is $O(\epsilon^2)$, we need $M = O(t^{-2}\epsilon^{-2})$ filtered samples, and we need $M=\omega(1)$ to ensure asymptotic normality.
By Lemma~\ref{lem:generating-filtered-samples}, to generate $M$ filtered samples, we need an expected $M' = \Theta(M \eta^{-1})$ total samples and thus $T = O(\eta^{-1}Mt) = O(\eta^{-1}t^{-1}\epsilon^{-2})$ total calls to $U$.
\end{proof}

It remains to fix the number of calls to the unitary $t$ in each circuit, however, in doing so we must maintain the large-$M$ limit in which our results were obtained.
In the early-FT setting, this is achieved as we fix $t$ by the maximum depth allowable for a circuit (see Sec.~\ref{sec:gaussian_with_gdn} for more details), and achieve arbitrary precision by increasing the number of samples $M$.
In the absence of experimental error, one can fix $t=c_{t}\epsilon^{-1}$ and recover the Heisenberg limit at sufficiently large $M$.
However, one cannot for instance set $t \propto a_0^{-1}$ and take the limit $a_0\rightarrow 0$, as this would lead to arbitrarily low $M$.

 \begin{figure}
     \centering
     \includegraphics[width=\linewidth]{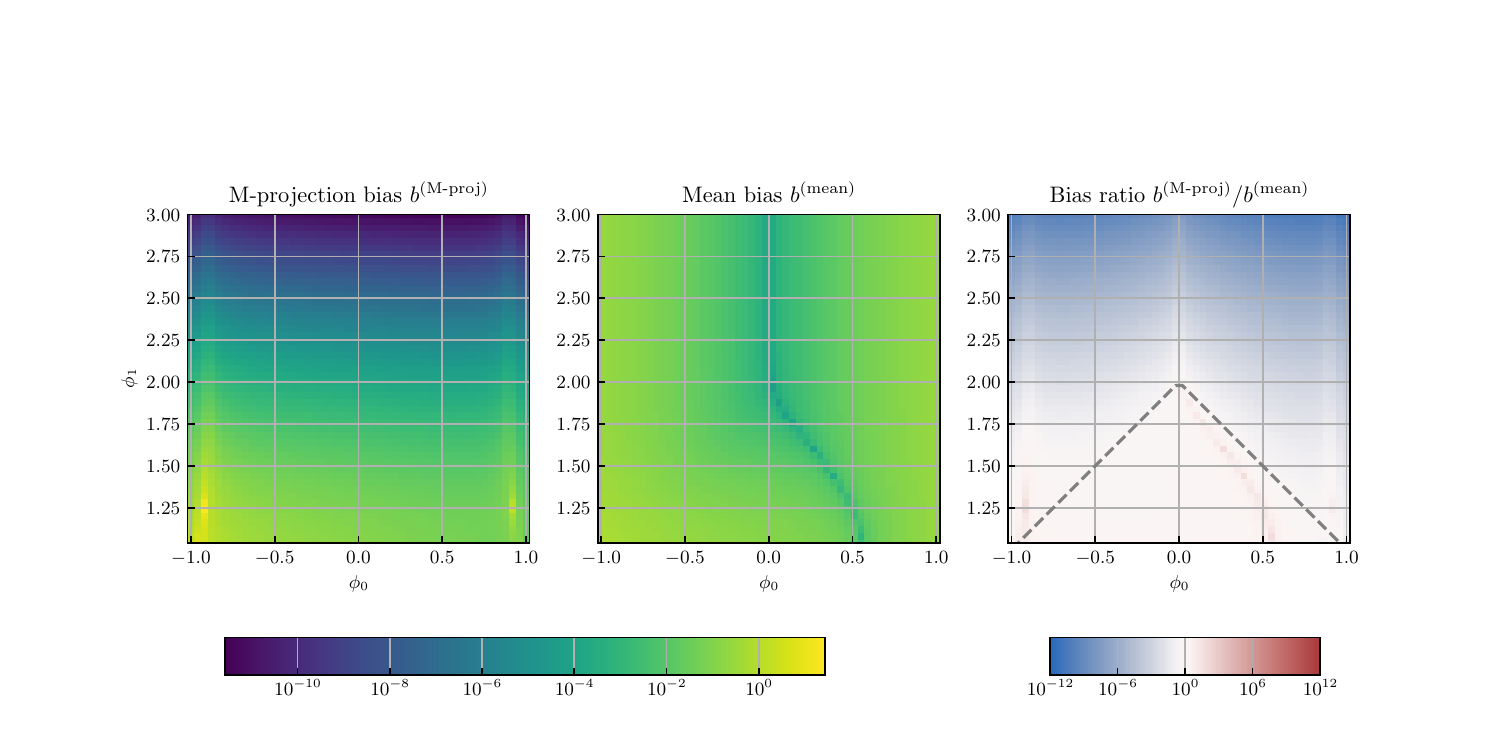}
     \caption{
     \label{fig:no-noise}
     Comparison of the bias of the moment projection and  mean estimators for the distribution $P(x)$ in Eq.~\eqref{eq:gaussian-prop-dist-no-noise}, with phases $\phi_0$, $\phi_1$, amplitudes $a_0 = 1-a_1 = 0.7$, filtering region $\mathcal{D} = [-1,1]$, and kernel width $\sigma = 0.3$. 
     (Left) The first order bias $b^{(\text{M-proj})}$ of the moment projection estimator [first term in Eq.~\eqref{eq:m-proj-1st-order-bias}].
     (Center) The bias of the mean estimator $b^{(\text{mean})} = \left| \int_{\mathcal{D}}P(x)x\dd x - \phi_0 \right|$.
     (Right) The ratio of the two biases, $b^{(\text{M-proj})}/b^{(\text{mean})}$.
     The dashed lines represent $d = |\phi_0 - \phi_{guess}|$ and  mark the regime in which the dominant source of the mean estimator's bias is the distance of $\phi_0$ from the center of the interval [Eq.~\eqref{eq:estimator-vs-mean-gaussian-no-noise}].
     }
 \end{figure}

\subsection{Gaussian kernel and global depolarising noise}\label{sec:gaussian_with_gdn}

We now consider the performance of the moment projection estimator again using a Gaussian kernel, but this time under the presence of global depolarizing noise (GDN). Global depolarizing noise assumes that each error event is maximally scrambling; this can be modelled by (with probability $p(x)$) replacing the quantum state with the maximally mixed state on $N$ qubits; i.e. $\rho\rightarrow (1-p)\rho + p I$.

In the presence of GDN, the probability distribution $p(x)$ becomes (as per Def.~\ref{def:noisy-distribution})
\begin{equation}
    p(x) = F a*f_\sigma(x) + (1-F)\frac{1}{2\pi}
\end{equation}
where $F$ is the circuit fidelity.
To see that this is correct, one can repeat the calculation in Eq.~\eqref{eq:random_phase_noiseless} modifying the probability of each circuit as $p(\tilde x | \phi_\text{ref}) \to Fp(\tilde x | \phi_\text{ref}) + (1-F)\frac{1}{K}$, as in the presence of GDN each bitstring is equally probable.
The probability of accepting a sample is
\begin{equation}
    \int_\mathcal{D} \dd x \, p(x) = F\sum_j a_j \int_\mathcal{D}\dd x f_\sigma(x-\phi_j) + (1-F)\frac{|\mathcal{D}|}{2\pi}
\end{equation}
and the filtered distribution becomes
\begin{equation}
\label{eq:gaussian-prob-dist-with-noise}
    P(x) =\frac{F \sum_j a_j g_\sigma(x-\phi_j) + (1-F) \frac{1}{2\pi}\mathcal{M}_\sigma}{F \sum_j a_j G_\sigma(\phi_j) + (1-F)\frac{|\mathcal{D}|}{2\pi}\mathcal{M}_\sigma}.
\end{equation}
where $\mathcal{M}_\sigma = \int_{-\pi}^{\pi}g_\sigma(x)\dd x = \erf\left(\frac{\pi}{\sqrt2}\sigma^{-1}\right)$ is the normalisation of the kernel function $f_\sigma$ (Def.~\ref{def:gaussian_kernel_function}).

In the presence of global depolarizing noise, mean sampling incurs a large sampling overhead compared to the moment projection estimator.
Ignoring terms exponentially small in $\sigma$ for simplicity, the expected value of $P(x)$ is
\begin{equation}
    E[x] \approx (1-w)\phi_0 + w\phi_{\mathrm{guess}},
\end{equation}
where the noise weight $w$ is given by
\begin{equation}
w = \frac{(1-F)\frac{1}{2\pi}\mathcal{M}_\sigma}{Fa_0 +(1-F)\frac{|\mathcal{D}|}{2\pi}\mathcal{M}_\sigma}.
\end{equation}
To get an estimator with an exponentially small bias using the sample average $\bar{x}$, we can use a shifted and rescaled mean estimator:
\begin{equation}
    \tilde\phi = \Big(1+\frac{1-w}{w}\Big)\bar{x}- \phi_{guess}\frac{1-w}{w}.
\end{equation}
However, as $\sigma\to0$, the variance of $P(x)$ remains constant, which means in turn that the variance of the mean estimator is constant.

In contrast, to extend our moment projection estimator to the setting with GDN, we can include the additional noise term in our model probability, and neglect the contribution of the other phases as before:
$Q(x|\phi) \propto F\,a_0\,f_\sigma(x-\phi_j) + (1-F)\, \frac{1}{2\pi}$.
We assume that $F$ and $a_0$ are known; determining these from the data itself would be an interesting target for future work.
As in the noiseless case, the bias decreases exponentially with decreasing $\sigma$, and the variance decreases as $\sigma^2$.

\begin{lemma} \label{lem:gaussian-gdn}
    Consider the moment projection phase estimator (Def.~\ref{def:m-projection-phase-estimator-with-gdn}) with  Gaussian kernel function $f_\sigma$ (Def.~\ref{def:gaussian_kernel_function} with $\epsilon_{\text{synth}} = 0$), fidelity $F$ and ground state overlap $a_0$.
    Assume that the $\mathcal{D}$ is a promise interval (Def.~\ref{def:promise_interval}) with inner buffer $d_{\text{in}} = |\mathcal{D}|/6$ and outer buffer $d_{\text{out}} = d$.
    Then, for
    \begin{equation}
        \sigma \leq \frac{|\mathcal{D}|}{6}\min\left(1,   \left[1.2+\log\left(1 + 0.7\frac{|\mathcal{D}|}{2\pi}\frac{1-F}{F a_0}\right)\right]^{-1/2}\right)
    \end{equation}
    the asymptotic bias and variance of the estimator in Lemma~\ref{lem:m-projection-convergence} satisfy
    \begin{align}
    |\phi_0-\phi_*| &\leq 22\ \sigma^2g_\sigma(d)\frac{1-a_0}{a_0} \left( 1+ \frac{1-F}{Fa_0}\right) \frac{|\mathcal{D}|}{d}   + O(g_\sigma(d))\\
    \lim_{M\to\infty} M \Var[\tilde\phi] &\leq 11\ \sigma^2 \left( 1 + \frac{1-F}{Fa_0}\right)\left( 1 + \frac{1}{2\pi}\frac{1-F}{Fa_0}\right) + O( g_\sigma(d)^2).
    \end{align}
\end{lemma}

We provide the proof of Lemma~\ref{lem:gaussian-gdn} in Appendix~\ref{app:proof-of-gaussian-cors}.
We can now convert the above result into the physical cost to execute phase estimation, assuming a fixed circuit fidelity $F$.

\begin{theorem} [Cost of noisy FMPE with fixed fidelity]
\label{thm:fmpe-cost-fixed-fidelity}
    Let $U$ be a unitary with spectral gap $\Delta$ around a target state $\phi_0$ ($\forall_{j>0} |\phi_j - \phi_0| > \Delta$).
    Assume oracle access to a controlled version of $U$, and an initial state $\ket{\psi}$ such that $|\braket{\phi_0}{\psi}|^2 = a_0$, and global depolarising noise with fixed circuit fidelity $F$.
    Further assume an initial estimate $\phi_{guess}$ of $\phi_0$ such that $|\phi_0-\phi_{guess}|<\Delta/3$.
    Then, the FMPE (Def.~\ref{def:m-projection-phase-estimator}) using $t$ calls to $U$ per circuit produces an estimate $\tilde\phi$ with RMS error $\epsilon$ using $M' = O(F^{-2}a_0^{-2}t^{-2}\epsilon^{-2})$ samples and $T = O(\epsilon^{-2}t^{-1}F^{-2}a_0^{-2})$ total calls to $U$, as long as $t= {\Omega}(\Delta^{-1}\log^{1/2}(\Delta\epsilon^{-1}a_0^{-2}F^{-1}))$ and $M'=\Omega(F^{-2}a_0^{-2})$.
\end{theorem}
\begin{proof}
As in the proof of Theorem~\ref{thm:fmpe-cost-nonoise}, we can then generate samples from $P(x)$ with $\mathcal{D}$ that satisfies the assumptions of Lemma~\ref{lem:gaussian-gdn} with $|\mathcal{D}|, d = \Theta(\Delta)$.
By Lemma~\ref{lem:gaussian-gdn}, $t= {\Omega}(\Delta^{-1}\log^{1/2}(\Delta\epsilon^{-1}a_0^{-2}F^{-1}))$ is enough to ensure that the bias is $O(\epsilon)$.
By Lemma~\ref{lem:generating-filtered-samples} and Eq.~\eqref{eq:fisher-info-before-filtering} to ensure that the variance is $O(\epsilon^2)$, we need $M' = \Theta(\sigma^{2}\epsilon^{-2}F^{-2}a_0^{-2}) = \Theta(t^{-2}\epsilon^{-2}F^{-2}a_0^{-2}) $ shots, so $T = M't = \Theta(\epsilon^{-2}t^{-1}F^{-2}a_0^{-2}) $ total calls to $U$.
\end{proof}

Note that while in the noiseless case it was enough to assume a lower bound $\eta \leq a_0$, here we instead assume that both $a_0$ and $F$ are known exactly.
The difference is that in the noiseless setting these parameters do not enter the model $Q(x|\phi)$, and a lower bound on $a_0$ is only needed to choose the total number of samples $M$.
In the present setting, by contrast, both $a_0$ and $F$ appear directly in the model, since they determine the relative size of the signal term $F a_0 f(x-\phi_0)$ and the noise term $(1-F)/(2\pi)$.
In principle these parameters could be estimated from the data itself; optimizing this would be an interesting task for future work.

So far in this work, we have treated the circuit fidelity and the width of the Gaussian $g_{\sigma}$ as independent variables.
However, to generate a distribution of width $\sigma$ using either the methods described in Sec.~\ref{sec:qspqpe} or Sec.~\ref{sec:qftqpe} requires a circuit depth $t\sim \sigma^{-1}$.
Under a typical noise model, the circuit fidelity decreases exponentially in the circuit depth~$t$: $F=e^{-\gamma t}$ for some decay rate $\gamma$.
This creates a trade-off in choosing the optimal circuit depth between limiting the onset of noise and achieving Heisenberg rather than sampling noise scaling, which we studied in detail in Ref.~\cite{dutkiewiczError2025}.
We now re-apply these methods to develop a QPE estimator in the presence of global depolarizing noise.

\begin{theorem} [Cost of noisy FMPE with fixed noise rate]
\label{thm:fmpe-cost-gdn}
    Let $U$ be a unitary with spectral gap $\Delta$ around a target state $\phi_0$ ($\forall_{j>0} |\phi_j - \phi_0| > \Delta$).
    Assume oracle access to a controlled version of $U$, and an initial state $\ket{\psi}$ such that $|\braket{\phi_0}{\psi}|^2 = a_0$, and global depolarising noise with circuit fidelity $F = e^{-\gamma t}$ for $t$ uses of $U$.
    Further assume an initial estimate $\phi_{guess}$ of $\phi_0$ such that $|\phi_0-\phi_{guess}|<\Delta/3$.
    Then, the FMPE (Def.~\ref{def:m-projection-phase-estimator}) using $t$ calls to $U$ per circuit produces an estimate $\tilde\phi$ with RMS error $\epsilon$ using $M' = O(e^{2\gamma t}a_0^{-2}t^{-2}\epsilon^{-2})$ samples and $T = O(\epsilon^{-2}t^{-1}e^{2\gamma t}a_0^{-2})$ total calls to $U$, as long as $t = \Omega(\Delta^{-1}(\gamma\Delta^{-1} + \log^{1/2}(a_0^{-2} \epsilon^{-1})))$  and $M'=\Omega(e^{2\gamma t}a_0^{-2})$.
\end{theorem}
\begin{proof}
As in the proof of Theorem~\ref{thm:fmpe-cost-nonoise}, we can then generate samples from $P(x)$ with $\mathcal{D}$ that satisfies the assumptions of Lemma~\ref{lem:gaussian-no-noise} with $|\mathcal{D}|, d = \Theta(\Delta)$.
By Lemma~\ref{lem:gaussian-gdn}, to ensure that the bias is $O(\epsilon)$, it is enough to take $t = \Omega(\Delta^{-1}(\gamma\Delta^{-1} + \log^{1/2}(a_0^{-2} \epsilon^{-1})))$
By Lemma~\ref{lem:generating-filtered-samples} and Eq.~\eqref{eq:fisher-info-before-filtering} to ensure that the variance is $O(\epsilon^2)$, we need $M' = \Theta(\sigma^{2}\epsilon^{-2}e^{2\gamma t}a_0^{-2}) = \Theta(t^{-2}\epsilon^{-2}e^{2\gamma t}a_0^{-2}) $ shots, so $T = M't = \Theta(\epsilon^{-2}t^{-1}e^{2\gamma t}a_0^{-2}) $ total calls to $U$.
\end{proof}

Unlike in Theorems~\ref{thm:fmpe-cost-nonoise} and~\ref{thm:fmpe-cost-fixed-fidelity}, where one can take $t = \Theta(\epsilon^{-1})$ and obtain Heisenberg-limited scaling $T = \Theta(\epsilon^{-1})$, such a choice is no longer optimal here, as the cost grows exponentially with $t$.
Instead, as in Ref.~\cite{dutkiewiczError2025}, optimizing over $t$ yields an optimal depth $t = \Theta(\gamma^{-1})$ and total cost
$T = \Theta(\gamma \epsilon^{-2} a_0^{-2})$.
Thus, the optimal circuit depth is set by the noise rate rather than the target precision.
This generalizes the result $T = \Theta(\gamma \epsilon^{-2})$ of Ref.~\cite{dutkiewiczError2025} to the case of imperfect initial states, with an additional multiplicative overhead of $a_0^{-2}$.

\subsection{Gaussian kernel and arbitrary noise with noise unbiasing}
\label{sec:gaussian_unbiasing}

Unlike the case of global depolarizing noise, arbitrary noise cannot be modelled at the level of distribution outcome.
To deal with arbitrary noise, we instead use the NU-FMPE (Definition \ref{def:m-projection-phase-estimator-unbiased}), which fits a noiseless model to the filtered samples from the quasiprobability decomposition of the ideal noiseless outcome distribution $P(x)$.
In this section we report the results on the bias and variance of the NU-FMPE when we assume a Gaussian kernel, thus $P(x)$ and $Q(x|\phi)$ as in Eq.~\eqref{eq:gaussian-prop-dist-no-noise} and Eq.~\eqref{eq:gaussian-model-no-noise} respectively, and a quasiprobability decomposition $P(x) = \sum_{a=0}^r \alpha_a P_a(x)$.

\begin{lemma} \label{lem:gaussian-unbiased}
    Consider the noise-unbiased moment projection phase estimator (Def.~\ref{def:m-projection-phase-estimator-unbiased}) with  Gaussian kernel function $f_\sigma$ (Def.~\ref{def:gaussian_kernel_function} with $\epsilon_{\text{synth}} = 0$), coefficients $\alpha$ and ground state overlap $a_0$.
    Assume that $\mathcal{D}$ is a promise interval (Def.~\ref{def:promise_interval}) with inner buffer $d_{\text{in}} = |\mathcal{D}|/6$ and outer buffer $d_{\text{out}} = d$.
    Then, for any $\sigma \leq \min(|\mathcal{D}|/6, 1/3c)$,
    the asymptotic bias and variance of the estimator in Lemma~\ref{lem:nme-convergence} satisfy
    \begin{align}
    \label{eq:unbiased-bias-bound}
    |\phi_0-\phi_*| &\leq 4.4\,\sigma^2 g_\sigma(d) \frac{1-a_0}{a_0} \frac{1+c|\mathcal{D}|}{d} + O\big( g^2_\sigma(d)\big),\\
    %\label{eq:unbiased-var-bound}
    \lim_{M\to\infty} M \Var[\tilde\phi] &\leq 102\, \|\alpha\|_1^2\sigma^2 (1 + |\mathcal{D}| c)^2\log(\frac{1}{2c\sigma}) + O\big(g_\sigma(d)\big).    
    \end{align}
\end{lemma}
The bias bound differs from the noiseless case (Lemma~\ref{lem:gaussian-no-noise}) only by a multiplicative constant factor of $1.1\frac{1+c|\mathcal{D}|}{c|\mathcal{D}|}$, and is independent of the quasiprobability overhead.
The variance, however, is increased by factor of  $51\|\alpha\|_1^2 (1 + |\mathcal{D}| c)^2\log(\frac{1}{2c\sigma})$.
As expected, error mitigation removes noise-induced bias at the cost of sample overhead.
We now combine this result with Lemma~\ref{lem:generating-filtered-samples-pec} to derive the overall resource requirements of the NU-FMPE.

\begin{theorem} [Cost of NU-FMPE with arbitrary noise]
\label{thm:fnmpe-cost}
    Let $U$ be a unitary with spectral gap $\Delta$ around a target state $\phi_0$ ($\forall_{j>0} |\phi_j - \phi_0| > \Delta$).
    Assume noisy oracle access to a controlled version of $U$, and an initial state $\ket{\psi}$ such that $|\braket{\phi_0}{\psi}|^2 \geq \eta$.
    Assume that for any circuit using $t$ calls to controlled-$U$ and one preparation of $\ket{\psi}$, there exists a
    quasiprobability decomposition into implementable noisy circuits with
    coefficient vector $\beta$, whose 1-norm satisfies
    $\|\beta\|_1 \leq A$.
    Further assume an initial estimate $\phi_{guess}$ of $\phi_0$ such that $|\phi_0-\phi_{guess}|<\Delta/3$.
    Then, the NU-FMPE (Def.~\ref{def:m-projection-phase-estimator-unbiased}) using $t$ calls to $U$ per circuit produces an estimate $\tilde\phi$ with RMS error $\epsilon$ using $M' = O((A+\kappa A^2)\eta^{-1}t^{-2}\epsilon^{-2}\log(\Delta t))$ samples and $T = O((A+\kappa A^{2})\eta^{-1}\epsilon^{-2}t^{-1}\log(\Delta t))$ total calls to $U$, as long as $t= {\Omega}(\Delta^{-1}\log^{1/2}(\Delta\epsilon^{-1}\eta^{-1}))$  and $M'=\Omega(\eta^{-1}(A+\kappa A^2))$,
    and $\kappa \in [0,1)$  is the filtering parameter introduced in Lemma~\ref{lem:generating-filtered-samples-pec}.
\end{theorem}
\begin{proof}
We prove this theorem by combining Lemma~\ref{lem:gaussian-unbiased} with Lemma~\ref{lem:generating-filtered-samples-pec}.
By Lemma~\ref{lem:generating-filtered-samples-pec}, we can construct a promise interval with that satisfies the assumptions of Lemma~\ref{lem:gaussian-unbiased}, with $|\mathcal{D}|, d_\text{in}, d_\text{out} = \Theta(\Delta)$, and sample from a PEC decomposition with weights $\alpha$ of the probability distibution $P(x)$ on this interval, by using circuit with $t$ uses of $U$ that synthesises a Gaussian with $\sigma = \Theta(t^{-1})$ (Lemma~\ref{lem:gaussian_kernel_synthesis}).
We fix the regularisation constant in NU-FMPE to be $c = \Theta(\Delta^{-1})$, so that the product $c|\mathcal{D}|$ which appears in the bounds in Lemma~\ref{lem:gaussian-unbiased} is $\Theta(1)$.
Due to Lemma~\ref{lem:gaussian-unbiased}, $t= {\Omega}(\Delta^{-1}\log^{1/2}(\Delta\epsilon^{-1}\eta^{-1}))$ is enough to ensure the bias is $O(\epsilon)$, and $M = \Theta(\|\alpha\|_1^2 \epsilon^{-2}t^{-2}\log(\Delta t))$ filtered samples are enough to ensure the variance is $O(\epsilon^2)$.
By Lemma~\ref{lem:generating-filtered-samples-pec}, the probability of accepting a sample is at least $\Omega( \eta^{-1}\|\alpha\|_1/\|\beta\|_1)$.
Therefore, the number of shots needed is $M' = O(\|\alpha\|_1 \|\beta\|_1\eta^{-1}t^{-2}\epsilon^{-2}\log(\Delta t))$.
By Lemma~\ref{lem:generating-filtered-samples-pec}, $\|\alpha\|_1 = O(1+\kappa\|\beta\|_1)$, so $M' = O(A(1+\kappa A)\,\eta^{-1}t^{-2}\log(\Delta t)\epsilon^{-2})$.
\end{proof}

Finally, we apply the result above to practically relevant noise models and express the complexity in terms of the circuit fidelity $F$.
For local stochastic noise channels with uniform error rate that is locally invertible the quasiprobability overhead norm $\|\beta\|_1$ satisfies $\|\beta\|_1 \leq F^{-2}$ \cite{caiQuantum2023,temmeError2017,endoPractical2018,dutkiewiczError2025}.
This class includes local Pauli channels with constant and qubit-independent error rate, e.g.~local depolarizing noise.
Substituting this relation into Theorem~\ref{thm:fnmpe-cost} yields the following corollary:
\begin{corollary} [Cost of NU-FMPE with Pauli noise]
\label{thm:fnmpe-cost-pauli}
    Let $U$ be a unitary with spectral gap $\Delta$ around a target state $\phi_0$ ($\forall_{j>0} |\phi_j - \phi_0| > \Delta$).
    Assume noisy oracle access to a controlled version of $U$, and an initial state $\ket{\psi}$ such that $|\braket{\phi_0}{\psi}|^2 \geq \eta$.
    Assume local Pauli noise, such that any circuit using $t$ calls to controlled-$U$ and one preparation of $\ket{\psi}$ has circuit fidelity $F$.
    Further assume an initial estimate $\phi_{guess}$ of $\phi_0$ such that $|\phi_0-\phi_{guess}|<\Delta/3$. 
    Then, the NU-FMPE (Def.~\ref{def:m-projection-phase-estimator-unbiased}) using $t$ calls to $U$ per circuit produces an estimate $\tilde\phi$ with RMS error $\epsilon$ using $M' = O((F^{-2}+\kappa F^{-4})\eta^{-1}t^{-2}\epsilon^{-2})$ samples and $T = O((F^{-2}+\kappa F^{-4})\epsilon^{-2}t^{-1}\eta^{-1})$ total calls to $U$, as long as $t= {\Omega}(\Delta^{-1}\log^{1/2}(\Delta\epsilon^{-1}\eta^{-1}))$ and $M'=\Omega(\eta^{-1}(F^2 + \kappa F^4))$ and $\kappa \in [0,1)$  is the filtering parameter introduced in Lemma~\ref{lem:generating-filtered-samples-pec}.
\end{corollary}

The observed worst-case scaling of $F^{-4}$ is standard for PEC.
However, in practice we expect the error mitigation overhead of NU-FMPE to be significantly smaller.
If the noise is uniformly distributed, then the overhead is $\sim F^{-2}+\frac{|\mathcal{D}|}{2\pi} F^{-4}$.
Following Ref.~\cite{dutkiewiczError2025}, when estimating energies in an early-fault tolerant setting, we expect the freedom to choose the number $t$ of uses of $U$ in the QPE circuit such that $F^2\geq 1/e$, which is optimal whenever the fidelity decreases exponentially with $t$.
In this case, $F^{-2}\frac{|\mathcal{D}|}{2\pi} \leq \frac{e|\mathcal{D}|}{2\pi}$, so as long as our promise interval has width significantly smaller than $2\pi/e\sim 2$, the $\kappa F^{-4}$ term will be subdominant to the $F^{-2}$ contribution.
Therefore, the NU-FMPE remains nearly unbiased whilst scaling only inverse-quadratically in the fidelity in the practically relevant regime.

\section{Numerical demonstration}
\label{sec:numerics}

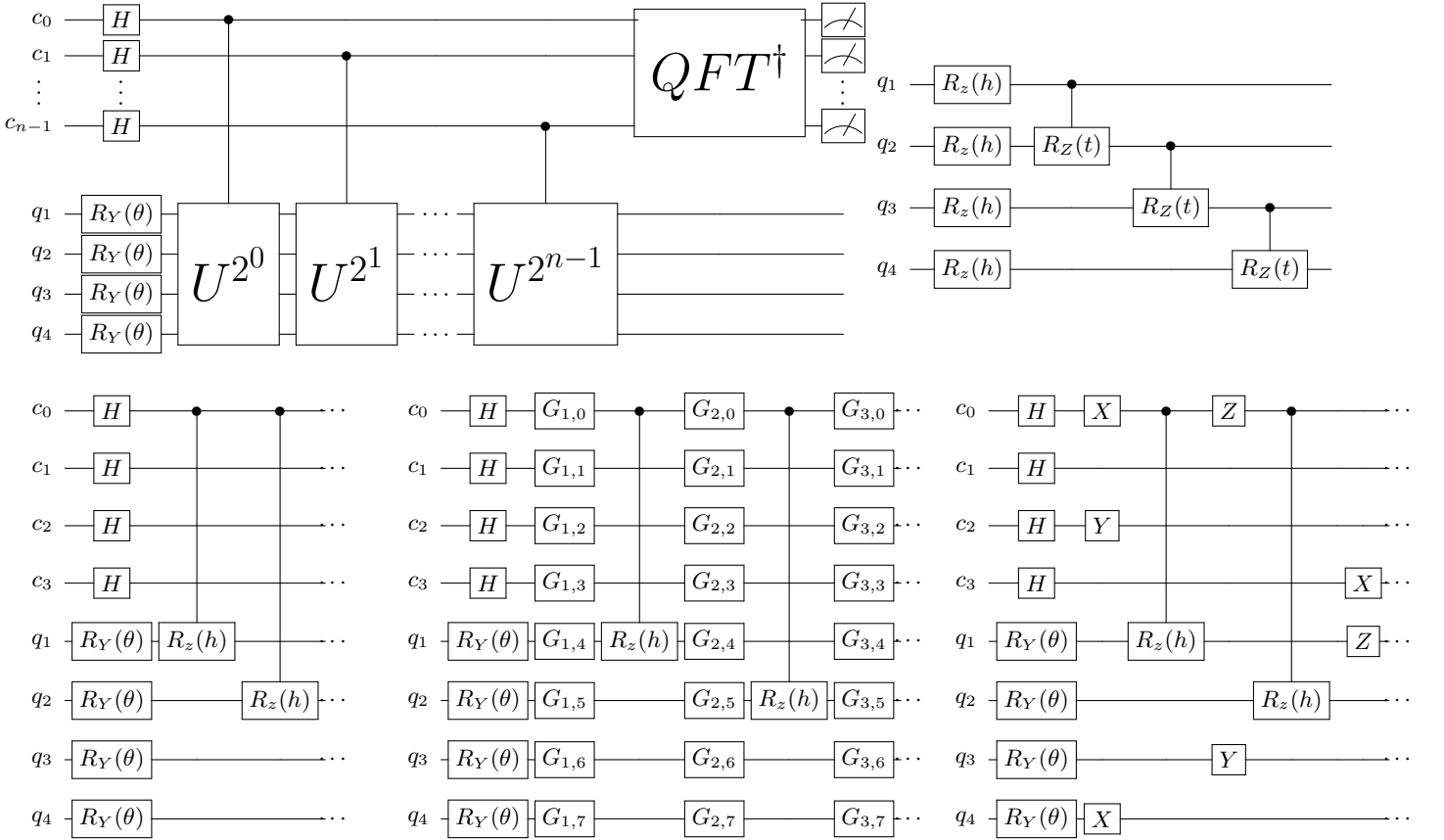
\begin{figure}
\centering

%==================================================
% ROW 1: NOISELESS CIRCUIT
%==================================================

\begin{minipage}{0.55\linewidth}
\centering
\begin{equation*}
    \Qcircuit @C=0.7em @R=.1em {
     %1st control bit
    \lstick{c_0} & \gate{H}&\ctrl{5}&\qw&\qw& \qw & \multigate{3}{\text{\huge $QFT^\dagger$}}& \meter\\ 
    %2nd control bit
    \lstick{c_1} &\gate{H}&\qw&\ctrl{4}& \qw &\qw& \ghost{\text{ \huge $QFT^\dagger$}}& \meter\\
    %vdots
    \lstick{\rotatebox{90}{$\cdots$}\ }&\push{\rotatebox{90}{$\cdots$}}&&&&&&\rotatebox{90}{$\cdots$}\\
    %last control bit
    \lstick{c_{n-1}} & \gate{H}&\qw &\qw&\qw&\ctrl{2}& \ghost{\text{ \huge $QFT^\dagger$}}& \meter\\ 
    %break line
    \push{\rule{0em}{2em}}&&&\\
    %system register
    \lstick{q_1}&\gate{R_Y(\theta)}& \multigate{3}{\text{\huge $U^{2^0}$}} & \multigate{3}{\text{\huge $U^{2^1}$}} & \push{\rule{0.15em}{0em}\dots\rule{0.15em}{0em}}\qw & \multigate{3}{\text{\huge $U^{2^{n-1}}$}} & \qw&\qw\\
    \lstick{q_2}&\gate{R_Y(\theta)}& \ghost{\text{\huge $U^{2^0}$}} & \ghost{\text{\huge $U^{2^1}$}} & \push{\rule{0.15em}{0em}\dots\rule{0.15em}{0em}}\qw & \ghost{\text{\huge $U^{2^{n-1}}$}} & \qw&\qw\\
    \lstick{q_3}&\gate{R_Y(\theta)}& \ghost{\text{\huge $U^{2^0}$}} & \ghost{\text{\huge $U^{2^1}$}} & \push{\rule{0.15em}{0em}\dots\rule{0.15em}{0em}}\qw & \ghost{\text{\huge $U^{2^{n-1}}$}} & \qw&\qw\\
    \lstick{q_4}&\gate{R_Y(\theta)}& \ghost{\text{\huge $U^{2^0}$}} & \ghost{\text{\huge $U^{2^1}$}} & \push{\rule{0.15em}{0em}\dots\rule{0.15em}{0em}}\qw & \ghost{\text{\huge $U^{2^{n-1}}$}} & \qw&\qw
    }
    \end{equation*}
\end{minipage}
\hfill
\begin{minipage}{0.4\linewidth}
\centering
\[
\Qcircuit @C=1.0em @R=1.0em {
\lstick{q_1} & \gate{R_z(h)} & \ctrl{1} & \qw & \qw & \qw \\
\lstick{q_2} & \gate{R_z(h)} & \gate{R_Z(t)} & \ctrl{1} & \qw & \qw\\
\lstick{q_3} & \gate{R_z(h)} & \qw & \gate{R_Z(t)} & \ctrl{1} & \qw\\
\lstick{q_4} & \gate{R_z(h)} & \qw & \qw  & \gate{R_Z(t)} & \qw
}
\]
\end{minipage}

\vspace{0.5em}

%==================================================
% BOTTOM ROW: NOISE SAMPLING
%==================================================

%-----------------------------
% (a) Noiseless layers
%-----------------------------
\begin{minipage}{0.2\linewidth}
\centering
\[
\Qcircuit @C=0.3em @R=0.9em {
\lstick{c_0} & \gate{H} & \ctrl{4} & \ctrl{5} & \push{\rule{0em}{1.5em}}\qw & \cdots \\
\lstick{c_1} & \gate{H} & \qw      & \qw      & \push{\rule{0em}{1.5em}}\qw & \cdots \\
\lstick{c_2} & \gate{H} & \qw      & \qw      & \push{\rule{0em}{1.5em}}\qw & \cdots \\
\lstick{c_3} & \gate{H} & \qw      & \qw      & \push{\rule{0em}{1.5em}}\qw & \cdots \\
\lstick{q_1} & \gate{R_Y(\theta)}      & \gate{R_z(h)} & \qw & \push{\rule{0em}{1.5em}}\qw & \cdots \\
\lstick{q_2} & \gate{R_Y(\theta)}      & \qw      & \gate{R_z(h)} & \push{\rule{0em}{1.5em}}\qw & \cdots \\
\lstick{q_3} & \gate{R_Y(\theta)} & \qw & \qw & \qw & \cdots \\
\lstick{q_4} & \gate{R_Y(\theta)} & \qw & \qw & \qw & \cdots
}
\]
\end{minipage}
\hfill
%-----------------------------
% (b) Abstract noisy circuit
%-----------------------------
\begin{minipage}{0.33\linewidth}
\centering
\[
\Qcircuit @C=0.3em @R=0.9em {
\lstick{c_0} & \gate{H} & \gate{G_{1,0}} & \ctrl{4} & \gate{G_{2,0}} & \ctrl{5} & \gate{G_{3,0}} & \qw &\cdots \\
\lstick{c_1} & \gate{H} & \gate{G_{1,1}} & \qw      & \gate{G_{2,1}} & \qw      & \gate{G_{3,1}} & \qw &\cdots \\
\lstick{c_2} & \gate{H} & \gate{G_{1,2}} & \qw      & \gate{G_{2,2}} & \qw      & \gate{G_{3,2}} & \qw &\cdots \\
\lstick{c_3} & \gate{H} & \gate{G_{1,3}} & \qw      & \gate{G_{2,3}} & \qw      & \gate{G_{3,3}} & \qw &\cdots \\
\lstick{q_1} & \gate{R_Y(\theta)}      & \gate{G_{1,4}} & \gate{R_z(h)} & \gate{G_{2,4}} & \qw & \gate{G_{3,4}} & \qw &\cdots \\
\lstick{q_2} & \gate{R_Y(\theta)} & \gate{G_{1,5}} & \qw      & \gate{G_{2,5}} & \gate{R_z(h)} & \gate{G_{3,5}} & \qw &\cdots \\
\lstick{q_3} & \gate{R_Y(\theta)} & \gate{G_{1,6}} & \qw      & \gate{G_{2,6}} & \qw & \gate{G_{3,6}} & \qw &\cdots \\
\lstick{q_4} & \gate{R_Y(\theta)} & \gate{G_{1,7}} & \qw      & \gate{G_{2,7}} & \qw & \gate{G_{3,7}} & \qw &\cdots
}
\]
\end{minipage}
\hfill
%-----------------------------
% (c) One sampled realization
%-----------------------------
\begin{minipage}{0.3\linewidth}
\centering
\[
\Qcircuit @C=0.33em @R=0.9em {
\lstick{c_0} & \gate{H} & \gate{X} & \ctrl{4} & \gate{Z} & \ctrl{5} & \qw & \qw & \push{\rule{0em}{1.5em}}\qw &\cdots \\
\lstick{c_1} & \gate{H} & \qw      & \qw      & \qw      & \qw      & \qw & \qw & \push{\rule{0em}{1.5em}}\qw &\cdots \\
\lstick{c_2} & \gate{H} & \gate{Y} & \qw      & \qw      & \qw      & \qw & \qw &\push{\rule{0em}{1.5em}}\qw &\cdots \\
\lstick{c_3} & \gate{H} & \qw      & \qw      & \qw      & \qw      & \qw & \gate{X} &\push{\rule{0em}{1.5em}}\qw & \cdots \\
\lstick{q_1} & \gate{R_Y(\theta)} & \qw      & \gate{R_z(h)} & \qw & \qw & \qw & \gate{Z} & \qw & \cdots \\
\lstick{q_2} & \gate{R_Y(\theta)} & \qw      & \qw      & \qw & \gate{R_z(h)} & \qw & \qw &\qw &\cdots \\
\lstick{q_3} & \gate{R_Y(\theta)} & \qw      & \qw      & \gate{Y} & \qw & \qw &\qw &\qw &\cdots \\
\lstick{q_4} & \gate{R_Y(\theta)} & \gate{X} & \qw      & \qw & \qw & \qw & \qw &\qw &\cdots
}
\]
\end{minipage}

\caption{
Circuits used in this work for example numerical QPE implementation.
\emph{Top row:} (Left) Noiseless circuit (Right) Decomposition of the system unitary
\(U = \prod_j e^{-i t Z_j Z_{j+1}}\prod_j e^{-i h Z_j}\).
\emph{Bottom row:}
(Left) First layers of the noiseless circuit for $n=4$ control qubits;
(Centre) insertion of stochastic gates \(G_{i,j}\) after each moment;
(Right) one explicit noisy circuit realization obtained by sampling Pauli errors.
}
\label{fig:ising_circuits}
\end{figure}

In this section, we implement and demonstrate the various estimators described throughout this text on a simplified problem for illustrative purposes.
This allows us to show the potential phase estimation practitioner how the above theory can be used in practice.
It furthermore allows us to test the effect of finite statistics on our estimators. The results derived above typically hold in the asymptotic limit in the number of samples, but the experimentalist is far more concerned over whether they need 5 or 5 million circuit samples to guarantee robustness.
And finally, the results allow us to check the critical constant factors in our phase estimation performance.

For simplicity, in this section we consider a 1D Ising model on 4 qubits:
\begin{equation}
    H = h \sum_{i=1}^4 Z_i + t \sum_{i=1}^3 Z_i Z_{i+1}
\end{equation}
with magnetic field $h = 0.27$, and coupling strength $t =-0.46$.
This is chosen to give a well-separated ground state, $\phi_1-\phi_0 = 1.46$, and we can choose a simple filtering interval $\mathcal{D}=[-\pi, -\pi/2]$ [Fig.~\ref{fig:ising_distribution}] that contains only $\phi_0$.
We obtain an initial state by applying $Y$ rotations with angle $\theta = 0.8\,\text{rad}$ to each qubit, yielding a ground state overlap $a_0 \approx 0.52$. 
In Fig.~\ref{fig:ising_distribution}, we plot the spectral function $a(x)$ [Eq.~\eqref{eq:spectral_distribution}], as well as its convolution with the Fejer kernel [Eq.~\eqref{eq:continuous-fejer-kernel}] from a textbook QPE implementation~\cite{nielsen2001quantum}.

We implement the full QPE circuit (see Fig.~\ref{fig:ising_circuits}, top row) using a quantum register of variable size ($n$ between 4 and 10 control qubits) using cirq.
For a simple noise model, we choose local depolarizing noise
\begin{equation}
    \rho \rightarrow (1-p_{\mathrm{err}})\rho + \tfrac{p_{\mathrm{err}}}{3}(X_j\rho X_j + Y_j\rho Y_j + Z_j\rho Z_j),
\end{equation}
applied on all qubits between each moment of gates (see Fig.~\ref{fig:ising_circuits}, bottom row).
We fix the error rate $p_{\mathrm{err}}$ such that the probability of no error is $(1-p_{\mathrm{err}})^{N_q\times N_{\mathrm{d}}}=\tfrac{1}{e}$, where $N_q$ is the total number of qubits (system plus control) and $N_{\mathrm{d}}$ is the circuit depth.
The resultant ochre probability distribution $p(x)$ in Fig.~\ref{fig:ising_distribution} can be split into a rescaled noiseless distribution and the distribution coming from circuits where at least one error has occurred [Fig.~\ref{fig:ising_distribution}, maroon]
\begin{equation}
\label{eq:numerics_pec_decomposition}
    p(x) = (1-F) p_1(x) + F(a*f)(x),
\end{equation}
which we will use as the quasiprobability distribution when testing the NU-FMPE\footnote{This decomposition ignores the overhead from constructing a real quasiprobability distribution, as in practice we do not have access to $p_1(x)$.
}.
We observe that the distribution $p_1(x)$ in Fig.~\ref{fig:ising_distribution} still retains a peak around $\phi_0$, implying that it contains significant phase estimation information as well.
This makes sense, as the effective volume of our QPE circuit is likely small, and PEC schemes that make use of this~\cite{aharonov2025reliable} can likely be incorporated into the NU-FMPE as well.
Pseudocode for the sampling routines used in the above is as follows:\\
\begin{minipage}[t]{0.49\textwidth}
    \begin{algorithm}[H]
    \label{alg:filtered_sampler}
    \SetAlgoInsideSkip{smallskip}
    
    \caption{Filtered sampler}
    
    \KwIn{%
    Filtering interval $\mathcal{D} \subset [-\pi, \pi]$,\newline
    Number of circuit shots $M' \in \NN$.
    }
    \KwOut{Samples $X$}
    
    \For{$j \gets 1, \dots, M'$}{
        Sample noisy circuit $C$ \\
        $P(x) \gets |\bra{x}C\ket{0}|^2$ \tcp{Simulate $C$} 
        Sample $x_j \sim P(x) $
        
        \If {$x_j \in \mathcal{D}$} {
            add $x_j$ to $X$
            }
    }
    \Return{X}
    \end{algorithm}
\end{minipage}
\hfill
\begin{minipage}[t]{0.49\textwidth}
    \begin{algorithm}[H]
    \caption{Filtered PEC sampler}
    \label{alg:filtered_PEC_sampler}
    \KwIn{%
    Filtering interval $\mathcal{D} \subset [-\pi, \pi]$, \newline
    Fidelity $F \in (0,1]$,\newline
    Number of circuit shots $M' \in \NN$.
    }
    \KwOut{Samples $X_0$, $X_1$}
    
    $\alpha \gets [F^{-1}, 1-F^{-1}]$ \;
    \For{$j \gets 1, \dots, M'$} {
        Sample $a \sim \frac{|\alpha_a|}{\lVert \alpha \rVert_1}$ 
        \\
        %$x \gets$ sample according to $p_a$ with $a_j$
        \If{$a = 0$} {
            %$x \gets$ sample from circuit with local depolarizing noise channel
            Sample noisy circuit $C$\\
            $P(x) \gets |\bra{x}C\ket{0}|^2$ \tcp{Simulate $C$} 
            Sample $x \sim P(x) $
        }
        \Else{
            %$a \gets$ sample from circuit with at least one noise event
            \Repeat{
                $C$ {\rm contains at least one error}
            }{
                Sample noisy circuit $C$
            }
            $P(x) \gets |\bra{x}C\ket{0}|^2$ \tcp{Simulate $C$} 
            Sample $x \sim P(x) $
        }
        \If {$x \in \mathcal{D}$} {
            add $x$ to $X_{a}$
            }
    }
    \Return{$X_0, X_1$}
    \end{algorithm}
\end{minipage}

\begin{figure}
    \centering
    \includegraphics[width=0.6\textwidth]{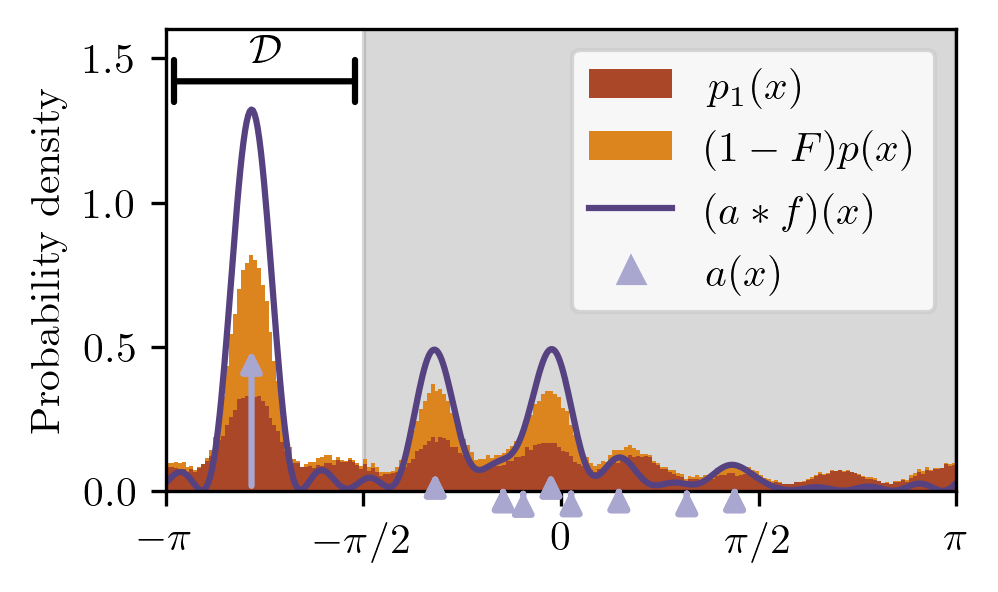}
    \caption{
        Spectral distribution and QPE results on a 4-qubit Ising model, with $n=4$ QPE control qubits.
        The arrows represent the eigenphases and corresponding amplitudes for the fast-forwarded unit-time evolution unitary of the considered model, i.e.~the spectral function $a(x)$.
        The blue line shows the smoothed spectral function $f*a(x)$, with the kernel function  corresponding to the continuous Fej\'er kernel expected from textbook QPE: $f(x) = f^{\text{(Fejer)}}K(x)$ from Eq.~\eqref{eq:continuous-fejer-kernel} with $K=2^{n}$.
        The ochre histogram shows the samples obtained from simulating $100k$ shots of textbook QPE with circuit local depolarizing noise, with error probabilities tuned such that the total circuit fidelity is $F=e^{-1}$ [i.e.~samples from the noisy QPE distribution $p(x)$, Eq.~\eqref{eq:numerics_pec_decomposition}]. 
        The maroon histogram represents the subset of samples coming from simulations where at least one noise event happened [i.e.~samples from the pure noise distribution, $p_1(x)$, Eq.~\eqref{eq:numerics_pec_decomposition}], which we use to simulate a simplified version of PEC/Noise Unbiasing.
        The un-shaded area represents the filtering interval $\mathcal{D}$.
    }
    \label{fig:ising_distribution}
\end{figure}

With our schemes for sampling from QPE probability distributions given, we now give pseudocode implementations for each of the estimators given in the text.
As we are using textbook quantum phase estimation, our model distribution $Q(x|\phi)$ for FMPE and NU-FMPE is the single phase distribution in Defs.~\ref{def:m-projection-phase-estimator-with-gdn},~\ref{def:m-projection-phase-estimator-unbiased} with the Fejer kernel (Eq.~\eqref{eq:continuous-fejer-kernel}).
\noindent
\begin{minipage}[t]{0.32\textwidth}
    \begin{algorithm}[H]
    \caption{Filtered PEC average}
    
    \KwIn{Samples $X_0$, $X_1$  generated by Alg.~\ref{alg:filtered_PEC_sampler}}
    \KwOut{Estimate $\tilde\phi$}

    $M_0 \gets |X_0|$ ; %number of samples with $a=0$
    $M_1 \gets |X_1|$ \;%number of samples with $a=1$

    \Return $\tilde\phi = \frac{\sum_{x \in X_0}x - \sum_{x \in X_1}x}{M_0 -M_1}$
    \end{algorithm}
\end{minipage}
\hfill
\begin{minipage}[t]{0.32\textwidth}
    \begin{algorithm}[H]    
    \caption{FMPE assuming GDN}
    
    \KwIn{%
    Samples $X$ \\ generated by Alg.~\ref{alg:filtered_sampler},\newline 
    promise interval $\mathcal{D}$,\newline
    number of control qubits $n \in \NN$,\newline
    fidelity $F \in (0,1]$,\newline
    overlap $a_0 \in (0,1]$
    }
    \KwOut{Estimate $\tilde\phi$}
    $M \gets |X|$ \;
    %$q(x|\phi) \gets Fa_0 f^{\text{(Fejer)}}_{2^n}(x-\phi) + (1-F)\frac{1}{2\pi}$
    $\begin{aligned}
    q(x|\phi) \gets {}&
    \textstyle Fa_0 f^{\text{(Fejer)}}_{2^n}(x-\phi)\\
    &+ (1-F)\textstyle\frac{1}{2\pi} ;
    \end{aligned}$\newline
    %\tcc{For $f^{\text{(Fejer)}}_{2^n}$ see Eq.\eqref{eq:continuous-fejer-kernel}}
    $Q(x|\phi) \gets \frac{q(x|\phi)}{\int_{\mathcal{D}}q(x|\phi) \dd x}$\;
    $\ell(\phi) \gets \frac{1}{M}\sum_{x\in X} \log Q(x|\phi)$\;
    \Return $\tilde\phi = \argmax \ell(\phi)$
    \end{algorithm}
\end{minipage}
\hfill
\begin{minipage}[t]{0.32\textwidth}
    \begin{algorithm}[H]
    
    \caption{NU-FMPE}
    
    \KwIn{Samples $X_0$, $X_1$  generated by Alg.~\ref{alg:filtered_PEC_sampler}, promise interval $\mathcal{D}$, number of control qubits $n \in \NN$, \newline regularization $c>0$
    }
    \KwOut{Estimate $\tilde\phi$}

    %$\ell(\phi) \gets \sum_{x\in X_0} \log Q(x|\phi) - \sum_{x\in X_1} \log Q(x|\phi)$
    $M \gets |X_0| + |X_1|$\;
    $q(x|\phi) \gets f^{\text{(Fejer)}}_{2^n}(x-\phi)$\;
    $Q_c(x|\phi) \gets \frac{q(x|\phi)}{\int_{\mathcal{D}}q(x|\phi) \dd x}+c$\;
    $\begin{aligned}
    \ell(\phi) \gets {}&
    \textstyle \frac{1}{M}\sum_{x\in X_0} \log Q_c(x|\phi) \\
    & - \tfrac{1}{M}\textstyle\sum_{x\in X_1} \log Q_c(x|\phi)\\
    &+c\int_{\mathcal{D}}\log Q_c(x|\phi) \dd x;
    \end{aligned}$\newline 
    \Return $\tilde\phi = \argmax \ell(\phi)$
    \end{algorithm}
\end{minipage}

\begin{figure}
    \centering
   \includegraphics[
   ]{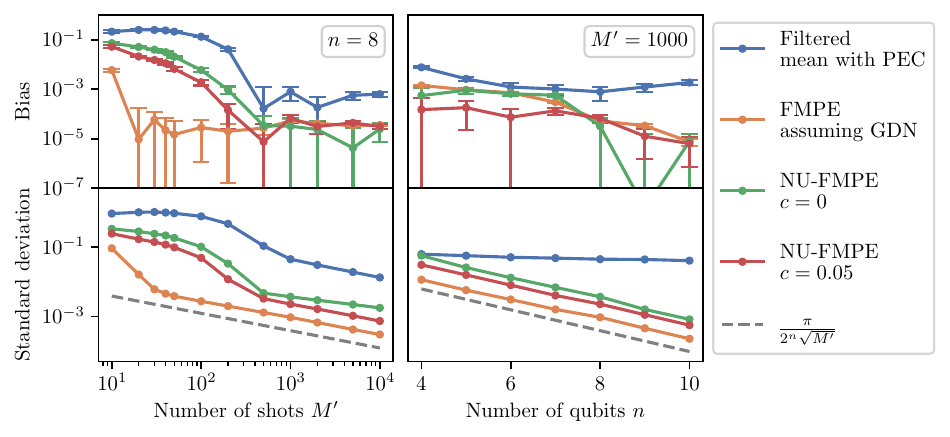}
    \caption{
    Performance of estimators considered in this work for 4-qubit Ising model [estimators distinguished in legend].
    The upper two plots show the estimator bias $|\tilde\phi - \phi_0|$, while the bottom two plots show the standard deviation.
    For the plots on the left, the number of control qubits $n$ is fixed to 8, and the total number of shots taken on a quantum computer (i.e. number of samples before filtering) $M'$ is varied.
    For the plots on the right, $M' = 1000$, and $n$ is varied.
    Error bars are obtained by bootstrapping \cite{efronBootstrap1979}.    
    }
    \label{fig:num_results}
\end{figure}

\topic{Fig 5 description}
In Fig.~\ref{fig:num_results}, we show the bias and the variance of all four estimators on our Ising model QPE simulation, for a range of control register sizes $n$ and numbers of circuit repetitions $M$ (shots).
(Note that we change the error rate with the number of qubits $n$ in order to maintain a constant circuit fidelity $F=\frac{1}{e}$.)
We observe that all estimators converge quickly in standard deviation to the standard quantum limit $M^{-1/2}$ with different constant factors, and the moment projection estimators converge as $2^{-n}$ whilst the mean estimator has a constant variance in $n$, as expected from the prior sections.
The moment projection estimator under a GDN assumption achieves this convergence to the standard quantum limit with fewer than $50$ circuit repetitions~\footnote{For the given starting state and circuit fidelity, filtering removes around $80\%$ of the data in these simulations, and the likelihood maximization is performed on only around $M=10$ data points at $M'=50$.}, making it an attractive choice for early fault-tolerant QPE implementations.
The bias of the mean estimator, the FMPE under a GDN assumption, and regularised NU-FMPE saturate as a function of $M$, which can be expected due to them not properly modeling experimental noise.
However the NU-FMPE without regularization has $0$ bias within the error bars of the simulation.
(In practice we expect a small residual bias due to the imperfect initial state, but this is likely too small to be observed.)
The bias trend with $n$ is less pronounced.
We see some evidence that that the FMPE and NU-FMPE have an exponentially decreasing bias while the mean estimator saturates.
However, more simulation at larger system sizes is required for a convincing demonstration of this result, and we present this solely as an example of the implementation of the estimators from this work.

\section{Conclusion and outlook}
\label{sec:conclusion}

In this work, we develop classical processing schemes for quantum phase estimation that both allow for imperfect state preparation and are robust to circuit-level noise.
Our approach assumes a phase interval containing only the target eigenphase, which for ground-state estimation can be constructed from a bound on the spectral gap $\Delta$.
We analyse the schemes using a Gaussian kernel in three settings: no noise, global depolarising noise, and arbitrary stochastic noise.
In all cases our estimators are asymptotically normal and their bias decreases exponentially with circuit depth.
This implies that the required minimal circuit depth $t$ to achieve a fixed bias grows logarithmically in target precision, the initial state overlap $a_0^{-1}$, and the circuit fidelity $F$, but $t = \widetilde{\Omega}(\Delta^{-1})$, consistent with known lower bounds.
In the noiseless case the estimator requires a total time $T=\Theta(\epsilon^{-2} t^{-1 }a_0^{-1})$, allowing for the Heisenberg scaling and achieving typical scaling with $a_0$.
For our global depolarising noise estimator, this worsens by a factor $\Theta(a_0^{-1}F^{-2})$, improving by a factor of $t^2$ over a naive mean estimator.
For arbitrary stochastic noise, the overhead is $\Theta((F^{-2}+\kappa F^{-4})\log(\Delta t))$, where $\kappa$ is the fraction of mitigated noise that falls outside the promise interval.
This interpolates between the $F^{-4}$ scaling of Probabilistic Error Cancellation, and a reduced $F^{-2}$ cost when one can postselect the noise away (but still pay the overhead for constructing a quasiprobability distribution).
Numerically we observe our noise-unbiased estimator requires around $500$ shots to overcome small-sample-statistics effects and converges to the standard quantum limit for the 4-qubit Ising model with local Pauli noise.
Surprisingly, this is outperformed by an estimator that (incorrectly) assumes global depolarising noise; this converges to the standard quantum limit with fewer than $<50$ shots and with lower variance than the noise-unbiased estimator.

Our new estimator solves a key roadblock to quantum phase estimation in early fault tolerance, where high-overlap state preparation is likely unfeasible and noise remains a concern.
However, it remains to test the performance of the moment projection estimator on realistic noise models and phase estimation problems, and to project performance at beyond-classical system sizes.
This is especially relevant given the observed noise-robustness, which suggests that the overhead for noise unbiasing may not be necessary across a range of target error rates.
One can further consider learning a few phases within the promise region $\mathcal{D}$ instead of just one, to allow for excited state energy estimation, or ground state estimation of gapless systems.
It may also be worth re-optimizing the choice of kernel function in the presence of noise (both within the QFT-QPE and QSP-QPE framework), as previous optimality results do not hold in the presence of noise.
And finally, it remains to compile the results of this work to make resource estimates for phase estimation of real problems of interest, in fault-tolerant architectures being pursued today.

\section*{Author contributions and acknowledgments}

All authors contributed to the conception of the ideas, development of the methods and to writing and reviewing the manuscript. 
A.D. performed the majority of the analytical derivations. 
A.D. and S.P. carried out the numerical calculations. 
Generative AI tools were used to assist with proofreading the text and checking the mathematical results.

A.D. acknowledges funding from an Ada Lovelace Fellowship from the Quantum Software Consortium.

\bibliography{bibliography.bib}

@article{rendon2023low,
  title={Low-depth Gaussian State Energy Estimation},
  author={Rendon, Gumaro and Johnson, Peter D},
  journal={arXiv preprint arXiv:2309.16790},
  year={2023}
}

@article{ding2023even,
  title={Even shorter quantum circuit for phase estimation on early fault-tolerant quantum computers with applications to ground-state energy estimation},
  author={Ding, Zhiyan and Lin, Lin},
  journal={PRX Quantum},
  volume={4},
  number={2},
  pages={020331},
  year={2023},
  publisher={APS}
}

@article{berry2017improved,
    title = {Improved Techniques for Preparing Eigenstates of Fermionic Hamiltonians},
    author = {Berry, Dominic W. and Kieferov\'{a}, M\'{a}ria and Scherer, Artur and Sanders, Yuval R. and Low, Guang Hao and Wiebe, Nathan and Gidney, Craig and Babbush, Ryan},
    journal = {npj Quant. Inf.},
    volume = {4},
    issue = {1},
    pages = {22},
    year = {2018},
    url = {https://arxiv.org/abs/1711.10460}
}

@article{ding2023simultaneous,
  title={Simultaneous estimation of multiple eigenvalues with short-depth quantum circuit on early fault-tolerant quantum computers},
  author={Ding, Zhiyan and Lin, Lin},
  journal={Quantum},
  volume={7},
  pages={1136},
  year={2023},
  publisher={Verein zur F{\"o}rderung des Open Access Publizierens in den Quantenwissenschaften}
}

@article{babbushEncoding2018,
  title = {Encoding {{Electronic Spectra}} in {{Quantum Circuits}} with {{Linear T Complexity}}},
  author = {Babbush, Ryan and Gidney, Craig and Berry, Dominic W. and Wiebe, Nathan and McClean, Jarrod and Paler, Alexandru and Fowler, Austin and Neven, Hartmut},
  year = {2018},
  month = oct,
  journal = {Physical Review X},
  volume = {8},
  number = {4},
  pages = {041015},
  publisher = {American Physical Society},
  doi = {10.1103/PhysRevX.8.041015},
  urldate = {2023-07-25}
}

@article{gorecki2020pi,
  title={$\pi$-corrected {{Heisenberg}} limit},
  author={G{\'o}recki, Wojciech and Demkowicz-Dobrza{\'n}ski, Rafa{\l} and Wiseman, Howard M and Berry, Dominic W},
  journal={Physical Review Letters},
  volume={124},
  number={3},
  pages={030501},
  year={2020},
  publisher={APS}
}

@article{berry2024analyzing,
  title={Analyzing prospects for quantum advantage in topological data analysis},
  author={Berry, Dominic W and Su, Yuan and Gyurik, Casper and King, Robbie and Basso, Joao and Barba, Alexander Del Toro and Rajput, Abhishek and Wiebe, Nathan and Dunjko, Vedran and Babbush, Ryan},
  journal={PRX Quantum},
  volume={5},
  number={1},
  pages={010319},
  year={2024},
  publisher={APS}
}

@article{luisOptimum1996,
  title = {Optimum Phase-Shift Estimation and the Quantum Description of the Phase Difference},
  author = {Luis, A. and Pe{\v r}ina, J.},
  year = {1996},
  month = nov,
  journal = {Physical Review A},
  volume = {54},
  number = {5},
  pages = {4564--4570},
  publisher = {American Physical Society},
  doi = {10.1103/PhysRevA.54.4564},
  urldate = {2023-07-26}
}

@misc{dutkiewiczError2025,
  title = {Error Mitigation and Circuit Division for Early Fault-Tolerant Quantum Phase Estimation},
  author = {Dutkiewicz, Alicja and Polla, Stefano and Scheurer, Maximilian and Gogolin, Christian and Huggins, William J. and O'Brien, Thomas E.},
  journal = {PRX Quantum},
  volume = {6},
  issue = {4},
  pages = {040318},
  numpages = {36},
  year = {2025},
  month = {Oct},
  publisher = {American Physical Society},
  doi = {10.1103/mlmy-yskj},
  url = {https://link.aps.org/doi/10.1103/mlmy-yskj}
}

@book{nielsen2001quantum,
  title = {Quantum Computation and Quantum Information},
  author = {Nielsen, Michael A and Chuang, Isaac L},
  year = {2001},
  edition = {2},
  volume = {54},
  pages = {60},
  publisher = {Cambridge University Press}
}

@article{rendon2024improved,
  title={Improved accuracy for Trotter simulations using Chebyshev interpolation},
  author={Rendon, Gumaro and Watkins, Jacob and Wiebe, Nathan},
  journal={Quantum},
  volume={8},
  pages={1266},
  year={2024},
  publisher={Verein zur F{\"o}rderung des Open Access Publizierens in den Quantenwissenschaften}
}

@article{dutkiewicz2022heisenberg,
    author = {Dutkiewicz, Alicja and Terhal, Barbara M and O'Brien, Thomas E},
    doi = {10.22331/q-2022-10-06-830},
    journal = {Quantum},
    pages = {830},
    publisher = {Verein zur F{\"o}rderung des Open Access Publizierens in den Quantenwissenschaften},
    title = {Heisenberg-limited quantum phase estimation of multiple eigenvalues with few control qubits},
    volume = {6},
    year = {2022}
}

@article{obrien2019quantum,
  title = {Quantum Phase Estimation of Multiple Eigenvalues for Small-Scale (Noisy) Experiments},
  author = {O'Brien, Thomas E. and Tarasinski, Brian and Terhal, Barbara M.},
  year = {2019},
  journal = {New Journal of Physics},
  volume = {21},
  pages = {023022}
}

@inproceedings{vanapeldoorn2023quantum,
  title={Quantum tomography using state-preparation unitaries},
  author={van Apeldoorn, Joran and Cornelissen, Arjan and Gily{\'e}n, Andr{\'a}s and Nannicini, Giacomo},
  booktitle={Proceedings of the 2023 annual ACM-SIAM symposium on discrete algorithms (SODA)},
  pages={1265--1318},
  year={2023},
  organization={SIAM}
}

@inproceedings{cornelissen2023sublinear,
  title={A sublinear-time quantum algorithm for approximating partition functions},
  author={Cornelissen, Arjan and Hamoudi, Yassine},
  booktitle={Proceedings of the 2023 annual ACM-Siam symposium on discrete algorithms (SODA)},
  pages={1245--1264},
  year={2023},
  organization={SIAM}
}

@article{ding2024quantum,
  title={Quantum Multiple Eigenvalue Gaussian filtered Search: an efficient and versatile quantum phase estimation method},
  author={Ding, Zhiyan and Li, Haoya and Lin, Lin and Ni, HongKang and Ying, Lexing and Zhang, Ruizhe},
  journal={Quantum},
  volume={8},
  pages={1487},
  year={2024},
  publisher={Verein zur F{\"o}rderung des Open Access Publizierens in den Quantenwissenschaften}
}

@article{geFaster2019,
  title = {Faster Ground State Preparation and High-Precision Ground Energy Estimation with Fewer Qubits},
  author = {Ge, Yimin and Tura, Jordi and Cirac, J. Ignacio},
  year = {2019},
  month = feb,
  journal = {Journal of Mathematical Physics},
  volume = {60},
  number = {2},
  pages = {022202},
  issn = {0022-2488},
  doi = {10.1063/1.5027484},
  urldate = {2025-02-08}
}

@article{Lin20Preparation,
  doi = {10.22331/q-2020-12-14-372},
  url = {10.22331/q-2020-12-14-372},
  title = {Near-optimal ground state preparation},
  author = {Lin, Lin and Tong, Yu},
  journal = {{Quantum}},
  issn = {2521-327X},
  publisher = {{Verein zur F{\"{o}}rderung des Open Access Publizierens in den Quantenwissenschaften}},
  volume = {4},
  pages = {372},
  month = dec,
  year = {2020}
}

@article{dongGround2022,
  title = {Ground State Preparation and Energy Estimation on Early Fault-Tolerant Quantum Computers via Quantum Eigenvalue Transformation of Unitary Matrices},
  author = {Dong, Yulong and Lin, Lin and Tong, Yu},
  year = {2022},
  month = oct,
  journal = {PRX Quantum},
  volume = {3},
  number = {4},
  eprint = {2204.05955},
  pages = {040305},
  issn = {2691-3399},
  doi = {10.1103/PRXQuantum.3.040305},
  urldate = {2024-05-22},
  archiveprefix = {arXiv}
}

@article{martyn2021grand,
  title={Grand unification of quantum algorithms},
  author={Martyn, John M and Rossi, Zane M and Tan, Andrew K and Chuang, Isaac L},
  journal={PRX quantum},
  volume={2},
  number={4},
  pages={040203},
  year={2021},
  publisher={APS}
}

@book{amariMethods2000,
  title = {Methods of {{Information Geometry}}},
  author = {Amari, Shun-ichi and Nagaoka, Hiroshi},
  year = {2000},
  publisher = {American Mathematical Soc.},
  googlebooks = {vc2FWSo7wLUC},
  isbn = {978-0-8218-4302-4},
  langid = {english}
}

@book{murphyMachine2012,
  title = {Machine Learning: A Probabilistic Perspective},
  shorttitle = {Machine Learning},
  author = {Murphy, Kevin P.},
  year = {2012},
  series = {Adaptive Computation and Machine Learning Series},
  publisher = {MIT Press},
  address = {Cambridge, MA},
  isbn = {978-0-262-01802-9},
  langid = {english},
  lccn = {Q325.5 .M87 2012}
}

@article{nielsenWHAT2018,
  title = {{{WHAT IS}}...an {{Information Projection}}?},
  author = {Nielsen, Frank},
  year = {2018},
  month = mar,
  journal = {Notices of the American Mathematical Society},
  volume = {65},
  number = {03},
  pages = {1},
  publisher = {American Mathematical Society (AMS)},
  issn = {0002-9920, 1088-9477},
  doi = {10.1090/noti1647},
  urldate = {2025-07-29},
  langid = {english}
}

@misc{tuananhleReverse2017,
  title = {Reverse vs {{Forward KL}}},
  author = {{Tuan Anh Le}},
  year = {2017},
  urldate = {2025-07-29},
  howpublished = {https://www.tuananhle.co.uk/notes/reverse-forward-kl.html}
}

@article{Harrow09Quantum,
    title = {Quantum algorithm for solving linear systems of equations},
    author = {Aram W. Harrow and Avinatan Hassidim and Seth Lloyd},
    journal = {Physical Review Letters},
    volume = {15},
    number = {103},
    pages = {150502},
    year = {2009},
    url = {https://arxiv.org/abs/0811.3171},
    doi = {10.1103/PhysRevLett.103.150502}
}

@article{Shor95Polynomial,
  title={Polynomial-time algorithms for prime factorization and discrete logarithms on a quantum computer},
  author={Shor, Peter W},
  journal={SIAM Review},
  volume={41},
  number={2},
  pages={303--332},
  year={1999},
  publisher={SIAM}
}

@article{acharyaQuantum2024,
  title = {Quantum Error Correction below the Surface Code Threshold},
  author = {Acharya, Rajeev and {Aghababaie-Beni}, Laleh and Aleiner, Igor and Andersen, Trond I. and Ansmann, Markus and Arute, Frank and Arya, Kunal and Asfaw, Abraham and Astrakhantsev, Nikita and Atalaya, Juan and Babbush, Ryan and Bacon, Dave and Ballard, Brian and Bardin, Joseph C. and Bausch, Johannes and Bengtsson, Andreas and Bilmes, Alexander and Blackwell, Sam and Boixo, Sergio and Bortoli, Gina and Bourassa, Alexandre and Bovaird, Jenna and Brill, Leon and Broughton, Michael and Browne, David A. and Buchea, Brett and Buckley, Bob B. and Buell, David A. and Burger, Tim and Burkett, Brian and Bushnell, Nicholas and Cabrera, Anthony and Campero, Juan and Chang, Hung-Shen and Chen, Yu and Chen, Zijun and Chiaro, Ben and Chik, Desmond and Chou, Charina and Claes, Jahan and Cleland, Agnetta Y. and Cogan, Josh and Collins, Roberto and Conner, Paul and Courtney, William and Crook, Alexander L. and Curtin, Ben and Das, Sayan and Davies, Alex and De Lorenzo, Laura and Debroy, Dripto M. and Demura, Sean and Devoret, Michel and Di Paolo, Agustin and Donohoe, Paul and Drozdov, Ilya and Dunsworth, Andrew and Earle, Clint and Edlich, Thomas and Eickbusch, Alec and Elbag, Aviv Moshe and Elzouka, Mahmoud and Erickson, Catherine and Faoro, Lara and Farhi, Edward and Ferreira, Vinicius S. and Burgos, Leslie Flores and Forati, Ebrahim and Fowler, Austin G. and Foxen, Brooks and Ganjam, Suhas and Garcia, Gonzalo and Gasca, Robert and Genois, {\'E}lie and Giang, William and Gidney, Craig and Gilboa, Dar and Gosula, Raja and Dau, Alejandro Grajales and Graumann, Dietrich and Greene, Alex and Gross, Jonathan A. and Habegger, Steve and Hall, John and Hamilton, Michael C. and Hansen, Monica and Harrigan, Matthew P. and Harrington, Sean D. and Heras, Francisco J. H. and Heslin, Stephen and Heu, Paula and Higgott, Oscar and Hill, Gordon and Hilton, Jeremy and Holland, George and Hong, Sabrina and Huang, Hsin-Yuan and Huff, Ashley and Huggins, William J. and Ioffe, Lev B. and Isakov, Sergei V. and Iveland, Justin and Jeffrey, Evan and Jiang, Zhang and Jones, Cody and Jordan, Stephen and Joshi, Chaitali and Juhas, Pavol and Kafri, Dvir and Kang, Hui and Karamlou, Amir H. and Kechedzhi, Kostyantyn and Kelly, Julian and Khaire, Trupti and Khattar, Tanuj and Khezri, Mostafa and Kim, Seon and Klimov, Paul V. and Klots, Andrey R. and Kobrin, Bryce and Kohli, Pushmeet and Korotkov, Alexander N. and Kostritsa, Fedor and Kothari, Robin and Kozlovskii, Borislav and Kreikebaum, John Mark and Kurilovich, Vladislav D. and Lacroix, Nathan and Landhuis, David and {Lange-Dei}, Tiano and Langley, Brandon W. and Laptev, Pavel and Lau, Kim-Ming and Guevel, Lo{\"i}ck Le and Ledford, Justin and Lee, Kenny and Lensky, Yuri D. and Leon, Shannon and Lester, Brian J. and Li, Wing Yan and Li, Yin and Lill, Alexander T. and Liu, Wayne and Livingston, William P. and Locharla, Aditya and Lucero, Erik and Lundahl, Daniel and Lunt, Aaron and Madhuk, Sid and Malone, Fionn D. and Maloney, Ashley and Mandr{\'a}, Salvatore and Martin, Leigh S. and Martin, Steven and Martin, Orion and Maxfield, Cameron and McClean, Jarrod R. and McEwen, Matt and Meeks, Seneca and Megrant, Anthony and Mi, Xiao and Miao, Kevin C. and Mieszala, Amanda and Molavi, Reza and Molina, Sebastian and Montazeri, Shirin and Morvan, Alexis and Movassagh, Ramis and Mruczkiewicz, Wojciech and Naaman, Ofer and Neeley, Matthew and Neill, Charles and Nersisyan, Ani and Neven, Hartmut and Newman, Michael and Ng, Jiun How and Nguyen, Anthony and Nguyen, Murray and Ni, Chia-Hung and O'Brien, Thomas E. and Oliver, William D. and Opremcak, Alex and Ottosson, Kristoffer and Petukhov, Andre and Pizzuto, Alex and Platt, John and Potter, Rebecca and Pritchard, Orion and Pryadko, Leonid P. and Quintana, Chris and Ramachandran, Ganesh and Reagor, Matthew J. and Rhodes, David M. and Roberts, Gabrielle and Rosenberg, Eliott and Rosenfeld, Emma and Roushan, Pedram and Rubin, Nicholas C. and Saei, Negar and Sank, Daniel and Sankaragomathi, Kannan and Satzinger, Kevin J. and Schurkus, Henry F. and Schuster, Christopher and Senior, Andrew W. and Shearn, Michael J. and Shorter, Aaron and Shutty, Noah and Shvarts, Vladimir and Singh, Shraddha and Sivak, Volodymyr and Skruzny, Jindra and Small, Spencer and Smelyanskiy, Vadim and Smith, W. Clarke and Somma, Rolando D. and Springer, Sofia and Sterling, George and Strain, Doug and Suchard, Jordan and Szasz, Aaron and Sztein, Alex and Thor, Douglas and Torres, Alfredo and Torunbalci, M. Mert and Vaishnav, Abeer and Vargas, Justin and Vdovichev, Sergey and Vidal, Guifre and Villalonga, Benjamin and Heidweiller, Catherine Vollgraff and Waltman, Steven and Wang, Shannon X. and Ware, Brayden and Weber, Kate and White, Theodore and Wong, Kristi and Woo, Bryan W. K. and Xing, Cheng and Yao, Z. Jamie and Yeh, Ping and Ying, Bicheng and Yoo, Juhwan and Yosri, Noureldin and Young, Grayson and Zalcman, Adam and Zhang, Yaxing and Zhu, Ningfeng and Zobrist, Nicholas},
  year = {2025},
  journal = {Nature},
  volume = {638},
  pages = {920–-926},
}

@article{acharyaSuppressing2023,
  title = {Suppressing Quantum Errors by Scaling a Surface Code Logical Qubit},
  author = {Acharya, Rajeev and Aleiner, Igor and Allen, Richard and Andersen, Trond I. and Ansmann, Markus and Arute, Frank and Arya, Kunal and Asfaw, Abraham and Atalaya, Juan and Babbush, Ryan and Bacon, Dave and Bardin, Joseph C. and Basso, Joao and Bengtsson, Andreas and Boixo, Sergio and Bortoli, Gina and Bourassa, Alexandre and Bovaird, Jenna and Brill, Leon and Broughton, Michael and Buckley, Bob B. and Buell, David A. and Burger, Tim and Burkett, Brian and Bushnell, Nicholas and Chen, Yu and Chen, Zijun and Chiaro, Ben and Cogan, Josh and Collins, Roberto and Conner, Paul and Courtney, William and Crook, Alexander L. and Curtin, Ben and Debroy, Dripto M. and Del Toro Barba, Alexander and Demura, Sean and Dunsworth, Andrew and Eppens, Daniel and Erickson, Catherine and Faoro, Lara and Farhi, Edward and Fatemi, Reza and Flores Burgos, Leslie and Forati, Ebrahim and Fowler, Austin G. and Foxen, Brooks and Giang, William and Gidney, Craig and Gilboa, Dar and Giustina, Marissa and Grajales Dau, Alejandro and Gross, Jonathan A. and Habegger, Steve and Hamilton, Michael C. and Harrigan, Matthew P. and Harrington, Sean D. and Higgott, Oscar and Hilton, Jeremy and Hoffmann, Markus and Hong, Sabrina and Huang, Trent and Huff, Ashley and Huggins, William J. and Ioffe, Lev B. and Isakov, Sergei V. and Iveland, Justin and Jeffrey, Evan and Jiang, Zhang and Jones, Cody and Juhas, Pavol and Kafri, Dvir and Kechedzhi, Kostyantyn and Kelly, Julian and Khattar, Tanuj and Khezri, Mostafa and Kieferov{\'a}, M{\'a}ria and Kim, Seon and Kitaev, Alexei and Klimov, Paul V. and Klots, Andrey R. and Korotkov, Alexander N. and Kostritsa, Fedor and Kreikebaum, John Mark and Landhuis, David and Laptev, Pavel and Lau, Kim-Ming and Laws, Lily and Lee, Joonho and Lee, Kenny and Lester, Brian J. and Lill, Alexander and Liu, Wayne and Locharla, Aditya and Lucero, Erik and Malone, Fionn D. and Marshall, Jeffrey and Martin, Orion and McClean, Jarrod R. and McCourt, Trevor and McEwen, Matt and Megrant, Anthony and Meurer Costa, Bernardo and Mi, Xiao and Miao, Kevin C. and Mohseni, Masoud and Montazeri, Shirin and Morvan, Alexis and Mount, Emily and Mruczkiewicz, Wojciech and Naaman, Ofer and Neeley, Matthew and Neill, Charles and Nersisyan, Ani and Neven, Hartmut and Newman, Michael and Ng, Jiun How and Nguyen, Anthony and Nguyen, Murray and Niu, Murphy Yuezhen and O'Brien, Thomas E. and Opremcak, Alex and Platt, John and Petukhov, Andre and Potter, Rebecca and Pryadko, Leonid P. and Quintana, Chris and Roushan, Pedram and Rubin, Nicholas C. and Saei, Negar and Sank, Daniel and Sankaragomathi, Kannan and Satzinger, Kevin J. and Schurkus, Henry F. and Schuster, Christopher and Shearn, Michael J. and Shorter, Aaron and Shvarts, Vladimir and Skruzny, Jindra and Smelyanskiy, Vadim and Smith, W. Clarke and Sterling, George and Strain, Doug and Szalay, Marco and Torres, Alfredo and Vidal, Guifre and Villalonga, Benjamin and Vollgraff Heidweiller, Catherine and White, Theodore and Xing, Cheng and Yao, Z. Jamie and Yeh, Ping and Yoo, Juhwan and Young, Grayson and Zalcman, Adam and Zhang, Yaxing and Zhu, Ningfeng and {Google Quantum AI}},
  year = {2023},
  month = feb,
  journal = {Nature},
  volume = {614},
  number = {7949},
  pages = {676--681},
  publisher = {Nature Publishing Group},
  issn = {1476-4687},
  doi = {10.1038/s41586-022-05434-1},
  urldate = {2024-06-28},
  copyright = {2023 The Author(s)},
  langid = {english}
}

@article{aspuru2005simulated,
  title={Simulated quantum computation of molecular energies},
  author={Aspuru-Guzik, Al{\'a}n and Dutoi, Anthony D and Love, Peter J and Head-Gordon, Martin},
  journal={Science},
  volume={309},
  number={5741},
  pages={1704--1707},
  year={2005},
  publisher={American Association for the Advancement of Science}
}

@article{bluvsteinLogical2024,
  title = {Logical Quantum Processor Based on Reconfigurable Atom Arrays},
  author = {Bluvstein, Dolev and Evered, Simon J. and Geim, Alexandra A. and Li, Sophie H. and Zhou, Hengyun and Manovitz, Tom and Ebadi, Sepehr and Cain, Madelyn and Kalinowski, Marcin and Hangleiter, Dominik and Bonilla Ataides, J. Pablo and Maskara, Nishad and Cong, Iris and Gao, Xun and Sales Rodriguez, Pedro and Karolyshyn, Thomas and Semeghini, Giulia and Gullans, Michael J. and Greiner, Markus and Vuleti{\'c}, Vladan and Lukin, Mikhail D.},
  year = {2024},
  month = feb,
  journal = {Nature},
  volume = {626},
  number = {7997},
  pages = {58--65},
  publisher = {Nature Publishing Group},
  issn = {1476-4687},
  doi = {10.1038/s41586-023-06927-3},
  urldate = {2024-06-28},
  copyright = {2023 The Author(s)},
  langid = {english}
}

@misc{clinton2024quantum,
  title={Quantum phase estimation without controlled unitaries},
  author={Clinton, Laura and Cubitt, Toby S and Garcia-Patron, Raul and Montanaro, Ashley and Stanisic, Stasja and Stroeks, Maarten},
  archiveprefix = {arXiv},
  eprint = {2410.21517},
  year={2024}
}

@article{goings2022reliably,
  title={Reliably assessing the electronic structure of cytochrome P450 on today’s classical computers and tomorrow’s quantum computers},
  author={Goings, Joshua J and White, Alec and Lee, Joonho and Tautermann, Christofer S and Degroote, Matthias and Gidney, Craig and Shiozaki, Toru and Babbush, Ryan and Rubin, Nicholas C},
  journal={Proceedings of the National Academy of Sciences},
  volume={119},
  number={38},
  pages={e2203533119},
  year={2022},
  publisher={National Academy of Sciences}
}

@article{leeEven2021,
  title = {Even More Efficient Quantum Computations of Chemistry through Tensor Hypercontraction},
  author = {Lee, Joonho and Berry, Dominic W and Gidney, Craig and Huggins, William J and Mcclean, Jarrod R and Wiebe, Nathan and Babbush, Ryan},
  year = {2021},
  journal = {Physical Review X},
  eprint = {2011.03494v3},
  urldate = {2022-04-25},
  archiveprefix = {arXiv}
}

@article{lin2022heisenberg,
    author = {Lin, Lin and Tong, Yu},
    doi = {10.1103/prxquantum.3.010318},
    journal = {PRX Quantum},
    number = {1},
    pages = {010318},
    publisher = {APS},
    title = {Heisenberg-limited ground-state energy estimation for early fault-tolerant quantum computers},
    volume = {3},
    year = {2022}
}

@misc{najafiOptimum2023,
  title = {Optimum Phase Estimation with Two Control Qubits},
  author = {Najafi, Peyman and Costa, Pedro C. S. and Berry, Dominic W.},
  year = {2023},
  month = mar,
  number = {arXiv:2303.12503},
  eprint = {2303.12503},
  publisher = {arXiv},
  doi = {10.48550/arXiv.2303.12503},
  urldate = {2024-10-02},
  archiveprefix = {arXiv}
}

@article{obrien2019calculating,
  title={Calculating energy derivatives for quantum chemistry on a quantum computer},
  author={O’Brien, Thomas E and Senjean, Bruno and Sagastizabal, Ramiro and Bonet-Monroig, Xavier and Dutkiewicz, Alicja and Buda, Francesco and DiCarlo, Leonardo and Visscher, Lucas},
  journal={npj Quantum Information},
  volume={5},
  number={1},
  pages={113},
  year={2019},
  publisher={Nature Publishing Group UK London}
}

@misc{paetznick2024demonstration,
  title={Demonstration of logical qubits and repeated error correction with better-than-physical error rates},
  author={Paetznick, A and da Silva, MP and Ryan-Anderson, C and Bello-Rivas, JM and Campora III, JP and Chernoguzov, A and Dreiling, JM and Foltz, C and Frachon, F and Gaebler, JP and others},
  archiveprefix = {arXiv},
  eprint = {2404.02280},
  number = {arXiv:2404.02280},
  year={2024}
}

@misc{patel2024optimal,
  title={Optimal coherent quantum phase estimation via tapering},
  author={Patel, Dhrumil and Tan, Shi Jie Samuel and Subasi, Yigit and Sornborger, Andrew T},
  archiveprefix = {arXiv},
  eprint = {2403.18927},
  year={2024}
}

@article{reiherElucidating2017,
  title = {Elucidating Reaction Mechanisms on Quantum Computers},
  author = {Reiher, Markus and Wiebe, Nathan and Svore, Krysta M. and Wecker, Dave and Troyer, Matthias},
  year = {2017},
  month = jul,
  journal = {Proceedings of the National Academy of Sciences},
  volume = {114},
  number = {29},
  pages = {7555--7560},
  publisher = {Proceedings of the National Academy of Sciences},
  doi = {10.1073/pnas.1619152114},
  urldate = {2024-06-30}
}

@article{rendon2022effects,
  title={Effects of cosine tapering window on quantum phase estimation},
  author={Rendon, Gumaro and Izubuchi, Taku and Kikuchi, Yuta},
  journal={Physical Review D},
  volume={106},
  number={3},
  pages={034503},
  year={2022},
  publisher={APS}
}

@article{russo2021evaluating,
    author = {Russo, Antonio E and Rudinger, Kenneth Michael and Morrison, Benjamin CA and Baczewski, Andrew David},
    doi = {10.2172/1855915},
    journal = {Physical Review Letters},
    number = {21},
    pages = {210501},
    publisher = {APS},
    title = {Evaluating energy differences on a quantum computer with robust phase estimation},
    volume = {126},
    year = {2021}
}

@article{wangQuantum2023,
  title = {Quantum Algorithm for Ground State Energy Estimation Using Circuit Depth with Exponentially Improved Dependence on Precision},
  author = {Wang, Guoming and Fran{\c c}a, Daniel Stilck and Zhang, Ruizhe and Zhu, Shuchen and Johnson, Peter D.},
  year = {2023},
  month = nov,
  journal = {Quantum},
  volume = {7},
  pages = {1167},
  publisher = {Verein zur F{\"o}rderung des Open Access Publizierens in den Quantenwissenschaften},
  doi = {10.22331/q-2023-11-06-1167},
  urldate = {2025-02-08},
  langid = {british}
}

@inproceedings{Gilyen19QSVT,
  title={Quantum singular value transformation and beyond: exponential improvements for quantum matrix arithmetics},
  author={Gily{\'e}n, Andr{\'a}s and Su, Yuan and Low, Guang Hao and Wiebe, Nathan},
  booktitle={Proceedings of the 51st annual ACM SIGACT symposium on theory of computing},
  pages={193--204},
  year={2019}
}

@article{kimmel2015robust,
    author = {Kimmel, Shelby and Low, Guang Hao and Yoder, Theodore J},
    doi = {10.1103/physreva.92.062315},
    journal = {Physical Review A},
    number = {6},
    pages = {062315},
    publisher = {APS},
    title = {Robust calibration of a universal single-qubit gate set via robust phase estimation},
    volume = {92},
    year = {2015}
}

@article{low2017optimal,
  title={Optimal Hamiltonian simulation by quantum signal processing},
  author={Low, Guang Hao and Chuang, Isaac L},
  journal={Physical review letters},
  volume={118},
  number={1},
  pages={010501},
  year={2017},
  publisher={APS}
}

@article{lowHamiltonian2019,
  title = {Hamiltonian Simulation by Qubitization},
  author = {Low, Guang Hao and Chuang, Isaac L.},
  year = {2019},
  month = jul,
  journal = {Quantum},
  volume = {3},
  pages = {163},
  publisher = {Verein zur Forderung des Open Access Publizierens in den Quantenwissenschaften},
  issn = {2521-327X},
  doi = {10.22331/q-2019-07-12-163}
}

@misc{wang2024faster,
  title={Faster ground state energy estimation on early fault-tolerant quantum computers via rejection sampling},
  author={Wang, Guoming and Fran{\c{c}}a, Daniel Stilck and Rendon, Gumaro and Johnson, Peter},
  year={2024},
archiveprefix = {arXiv},
  eprint = {2304.09827v1}
}

@misc{Wocjan06Several,
    title = {Several natural {B}{Q}{P}-Complete problems},
    author = {Pawel Wocjan and Shengyu Zhang},
    archiveprefix = {arXiv},
    eprint = {quant-ph/0606179},
    number = {arXiv:quant-ph/0606179},
    year = {2006},
    url = {https://arxiv.org/abs/quant-ph/0606179},
    doi = {10.48550/arXiv.quant-ph/0606179}
}

@article{aharonov2025reliable,
  title={Reliable high-accuracy error mitigation for utility-scale quantum circuits},
  author={Aharonov, Dorit and Alberton, Ori and Arad, Itai and Atia, Yosi and Bairey, Eyal and Dov, Matan Ben and Berkovitch, Asaf and Brakerski, Zvika and Cohen, Itsik and Fuchs, Eran and others},
  journal={arXiv preprint arXiv:2508.10997},
  year={2025}
}

@article{caiQuantum2023,
  title = {Quantum Error Mitigation},
  author = {Cai, Zhenyu and Babbush, Ryan and Benjamin, Simon C. and Endo, Suguru and Huggins, William J. and Li, Ying and McClean, Jarrod R. and O'Brien, Thomas E.},
  year = {2023},
  month = dec,
  journal = {Reviews of Modern Physics},
  volume = {95},
  number = {4},
  pages = {045005},
  publisher = {American Physical Society},
  doi = {10.1103/RevModPhys.95.045005},
  urldate = {2025-02-08}
}

@misc{wang2025efficient,
  title={Efficient ground-state-energy estimation and certification on early fault-tolerant quantum computers},
  author={Wang, Guoming and Fran{\c{c}}a, Daniel Stilck and Rendon, Gumaro and Johnson, Peter D},
  journal={Physical Review A},
  volume={111},
  number={1},
  pages={012426},
  year={2025},
  publisher={APS}
}

@article{endoPractical2018,
  title = {Practical {{Quantum Error Mitigation}} for {{Near-Future Applications}}},
  author = {Endo, Suguru and Benjamin, Simon C. and Li, Ying},
  year = {2018},
  month = jul,
  journal = {Physical Review X},
  volume = {8},
  number = {3},
  pages = {031027},
  publisher = {American Physical Society},
  doi = {10.1103/PhysRevX.8.031027},
  urldate = {2024-10-02}
}

@article{temmeError2017,
  title = {Error {{Mitigation}} for {{Short-Depth Quantum Circuits}}},
  author = {Temme, Kristan and Bravyi, Sergey and Gambetta, Jay M.},
  year = {2017},
  month = nov,
  journal = {Physical Review Letters},
  volume = {119},
  number = {18},
  pages = {180509},
  publisher = {American Physical Society},
  doi = {10.1103/PhysRevLett.119.180509},
  urldate = {2023-12-01}
}

@article{somma2019quantum,
    author = {Rolando D. Somma},
    doi = {10.1088/1367-2630/ab5c60},
    journal = {New Journal of Physics},
    pages = {123025},
    title = {Quantum eigenvalue estimation via time series analysis},
    url = {https://iopscience.iop.org/article/10.1088/1367-2630/ab5c60/pdf},
    volume = {21},
    year = {2019}
}

@article{efronBootstrap1979,
  title = {Bootstrap {{Methods}}: {{Another Look}} at the {{Jackknife}}},
  shorttitle = {Bootstrap {{Methods}}},
  author = {Efron, B.},
  year = 1979,
  month = jan,
  journal = {The Annals of Statistics},
  volume = {7},
  number = {1},
  pages = {1--26},
  publisher = {Institute of Mathematical Statistics},
  issn = {0090-5364, 2168-8966},
  doi = {10.1214/aos/1176344552},
  urldate = {2026-01-15}
}

@article{Berry07Efficient,
    title = {Efficient quantum algorithms for simulating sparse {H}amiltonians},
    author = {Dominic W. Berry and Graeme Ahokas and Richard Cleve and Barry C. Sanders},
    journal = {Communications in Mathemathical Physics},
    volume = {270},
    number = {359},
    year = {2007},
    url = {https://arxiv.org/abs/quant-ph/0508139},
    doi = {10.1007/s00220-006-0150-x}
}

@misc{akahoshiCompilation2024,
  title = {Compilation of {{Trotter-Based Time Evolution}} for {{Partially Fault-Tolerant Quantum Computing Architecture}}},
  author = {Akahoshi, Yutaro and Toshio, Riki and Fujisaki, Jun and Oshima, Hirotaka and Sato, Shintaro and Fujii, Keisuke},
  year = {2024},
  month = oct,
  number = {arXiv:2408.14929},
  eprint = {2408.14929},
  publisher = {arXiv},
  doi = {10.48550/arXiv.2408.14929},
  urldate = {2025-02-08},
  archiveprefix = {arXiv}
}

@article{akahoshiPartially2024,
  title = {Partially {{Fault-Tolerant Quantum Computing Architecture}} with {{Error-Corrected Clifford Gates}} and {{Space-Time Efficient Analog Rotations}}},
  author = {Akahoshi, Yutaro and Maruyama, Kazunori and Oshima, Hirotaka and Sato, Shintaro and Fujii, Keisuke},
  year = {2024},
  month = mar,
  journal = {PRX Quantum},
  volume = {5},
  number = {1},
  pages = {010337},
  publisher = {American Physical Society},
  doi = {10.1103/PRXQuantum.5.010337},
  urldate = {2025-02-08}
}

@article{belliardoAchieving2020,
  title = {Achieving {{Heisenberg}} Scaling with Maximally Entangled States: {{An}} Analytic Upper Bound for the Attainable Root-Mean-Square Error},
  shorttitle = {Achieving {{Heisenberg}} Scaling with Maximally Entangled States},
  author = {Belliardo, Federico and Giovannetti, Vittorio},
  year = {2020},
  month = oct,
  journal = {Physical Review A},
  volume = {102},
  number = {4},
  pages = {042613},
  publisher = {American Physical Society},
  doi = {10.1103/PhysRevA.102.042613},
  urldate = {2023-07-28}
}

@article{bultriniBattle2023,
  title = {The Battle of Clean and Dirty Qubits in the Era of Partial Error Correction},
  author = {Bultrini, Daniel and Wang, Samson and Czarnik, Piotr and Gordon, Max Hunter and Cerezo, M. and Coles, Patrick J. and Cincio, Lukasz},
  year = {2023},
  month = jul,
  journal = {Quantum},
  volume = {7},
  pages = {1060},
  publisher = {Verein zur F{\"o}rderung des Open Access Publizierens in den Quantenwissenschaften},
  doi = {10.22331/q-2023-07-13-1060},
  urldate = {2025-02-08},
  langid = {british}
}

@article{campbellEarly2022,
  title = {Early Fault-Tolerant Simulations of the {{Hubbard}} Model},
  author = {Campbell, Earl T.},
  year = {2022},
  month = jan,
  journal = {Quantum Science and Technology},
  volume = {7},
  number = {1},
  eprint = {2012.09238},
  pages = {015007},
  issn = {2058-9565},
  doi = {10.1088/2058-9565/ac3110},
  urldate = {2024-08-02},
  archiveprefix = {arXiv},
  langid = {english}
}

@misc{ding2023robust,
  title = {Robust Ground-State Energy Estimation under Depolarizing Noise},
  author = {Ding, Zhiyan and Dong, Yulong and Tong, Yu and Lin, Lin},
  year = {2023},
  number = {arXiv:2307.11257},
  eprint = {2307.11257},
  archiveprefix = {arXiv}
}

@misc{guNoiseresilient2022,
  title = {Noise-Resilient Phase Estimation with Randomized Compiling},
  author = {Gu, Yanwu and Ma, Yunheng and Forcellini, Nicolo and Liu, Dong E.},
  year = {2022},
  month = aug,
  number = {arXiv:2208.04100},
  eprint = {2208.04100},
  publisher = {arXiv},
  urldate = {2022-10-05},
  archiveprefix = {arXiv}
}

@article{higgins2009demonstrating,
    author = {Higgins, B. L. and Berry, D. W. and Bartlett, S. D. and Mitchell, M. W. and Wiseman, H. M. and Pryde, G. J.},
    doi = {10.1088/1367-2630/11/7/073023},
    journal = {New Journal of Physics},
    pages = {073023},
    title = {Demonstrating {{H}}eisenberg-limited unambiguous phase estimation without adaptive measurements},
    url = {https://iopscience.iop.org/article/10.1088/1367-2630/11/7/073023},
    volume = {11},
    year = {2009}
}

@misc{katabarwaEarly2023,
  title = {Early {{Fault-Tolerant Quantum Computing}}},
  author = {Katabarwa, Amara and Gratsea, Katerina and Caesura, Athena and Johnson, Peter D.},
  year = {2023},
  month = nov,
  number = {arXiv:2311.14814},
  eprint = {2311.14814},
  publisher = {arXiv},
  urldate = {2023-11-28},
  archiveprefix = {arXiv}
}

@article{kshirsagarProving2024,
  title = {On Proving the Robustness of Algorithms for Early Fault-Tolerant Quantum Computers},
  author = {Kshirsagar, Rutuja and Katabarwa, Amara and Johnson, Peter D.},
  year = {2024},
  month = nov,
  journal = {Quantum},
  volume = {8},
  pages = {1531},
  publisher = {Verein zur F{\"o}rderung des Open Access Publizierens in den Quantenwissenschaften},
  doi = {10.22331/q-2024-11-20-1531},
  urldate = {2025-02-08},
  langid = {british}
}

@article{liangModeling2024,
  title = {Modeling the Performance of Early Fault-Tolerant Quantum Algorithms},
  author = {Liang, Qiyao and Zhou, Yiqing and Dalal, Archismita and Johnson, Peter},
  year = {2024},
  month = may,
  journal = {Physical Review Research},
  volume = {6},
  number = {2},
  pages = {023118},
  publisher = {American Physical Society},
  doi = {10.1103/PhysRevResearch.6.023118},
  urldate = {2025-02-08}
}

@misc{nelsonAssessment2024,
  title = {An Assessment of Quantum Phase Estimation Protocols for Early Fault-Tolerant Quantum Computers},
  author = {Nelson, Jacob S. and Baczewski, Andrew D.},
  year = {2024},
  month = feb,
  number = {arXiv:2403.00077},
  eprint = {2403.00077},
  publisher = {arXiv},
  urldate = {2024-08-19},
  archiveprefix = {arXiv},
  langid = {english}
}

@misc{toshioPractical2024,
  title = {Practical Quantum Advantage on Partially Fault-Tolerant Quantum Computer},
  author = {Toshio, Riki and Akahoshi, Yutaro and Fujisaki, J and Oshima, H and Sato, S and Fujii, Keisuke},
  year = {2024},
  month = aug,
  eprint = {2408.14848}
}

@article{wanRandomized2022,
  title = {Randomized {{Quantum Algorithm}} for {{Statistical Phase Estimation}}},
  author = {Wan, Kianna and Berta, Mario and Campbell, Earl T.},
  year = {2022},
  month = jul,
  journal = {Physical Review Letters},
  volume = {129},
  number = {3},
  pages = {030503},
  issn = {0031-9007, 1079-7114},
  doi = {10.1103/PhysRevLett.129.030503},
  urldate = {2023-01-11},
  langid = {english}
}

@article{wangState2022,
  title = {State {{Preparation Boosters}} for {{Early Fault-Tolerant Quantum Computation}}},
  author = {Wang, Guoming and Sim, Sukin and Johnson, Peter D.},
  year = {2022},
  month = oct,
  journal = {Quantum},
  volume = {6},
  pages = {829},
  publisher = {Verein zur F{\"o}rderung des Open Access Publizierens in den Quantenwissenschaften},
  doi = {10.22331/q-2022-10-06-829},
  urldate = {2024-05-22},
  langid = {british}
}

@article{zhangComputing2022,
  title = {Computing {{Ground State Properties}} with {{Early Fault-Tolerant Quantum Computers}}},
  author = {Zhang, Ruizhe and Wang, Guoming and Johnson, Peter},
  year = {2022},
  month = jul,
  journal = {Quantum},
  volume = {6},
  pages = {761},
  publisher = {Verein zur F{\"o}rderung des Open Access Publizierens in den Quantenwissenschaften},
  doi = {10.22331/q-2022-07-11-761},
  urldate = {2024-05-22},
  langid = {british}
}

@article{cleve1998quantum,
  title = {Quantum Algorithms Revisited},
  author = {Cleve, Richard and Ekert, Artur and Macchiavello, Chiara and Mosca, Michele},
  year = {1998},
  journal = {Proceedings of the Royal Society of London. Series A: Mathematical, Physical and Engineering Sciences},
  volume = {454},
  number = {1969},
  pages = {339--354},
  publisher = {The Royal Society}
}

@article{cai2021practical,
  title={A practical framework for quantum error mitigation},
  author={Cai, Zhenyu},
  journal={arXiv preprint arXiv:2110.05389},
  year={2021}
}

\appendix

\section{Proofs of Lemmas~\ref{lem:m-projection-convergence},~\ref{lem:m-projection-expansion},~\ref{lem:nme-convergence}, and~\ref{lem:nme-expansion}}
\label{app:m-projection-proof}

% custom code to remove appendix subsections from TOC
\stoptocentries

\subsection{Proof of Lemmas~\ref{lem:m-projection-convergence} and \ref{lem:nme-convergence}}

In this section we prove Lemma~\ref{lem:nme-convergence}. Note that this also proves Lemma~\ref{lem:m-projection-convergence} by fixing $r = 1$, $\alpha = 1$, and $c=0$ (in which case $P_c(x)=R(x)=P(x), Q_c(x|\phi)=Q(x|\phi)$).

The estimate $\tilde\phi$ satisfies the stationary point condition
\begin{equation}
    \ell'(\tilde\phi|\{x_j,a_j\}) = 0,
\end{equation}
where $\tilde\phi$ and $\ell$ are defined in Def.~\ref{def:NME}. 
Expanding this equation around $\phi_*$ (as defined in Lemma~\ref{lem:nme-convergence}) gives
\begin{equation}
    \ell'(\tilde\phi|\{x_j,a_j\})
    =
    \ell'(\phi_*|\{x_j,a_j\})
    +
    (\tilde\phi-\phi_*)\ell''(\phi_*|\{x_j,a_j\})
    +
    O((\tilde\phi-\phi_*)^2),
\end{equation}
which implies
\begin{equation}
\label{eq:stationary-point}
    \tilde\phi-\phi_*
    =
    -\frac{\ell'(\phi_*|\{x_j,a_j\})}
    {\ell''(\phi_*|\{x_j,a_j\})}
    +
    O((\tilde\phi-\phi_*)^2).
\end{equation}
We now analyze the numerator and denominator of the expression above in the limit $M\to\infty$.
First, we consider the denominator. 
Using the definition of $\ell(\phi_*|\{x_j,a_j\})$,
\begin{equation}
\ell''(\phi_*|\{x_j,a_j\})
=
\frac{\|\alpha\|_1}{M}\sum_{j=1}^M
\mathrm{sgn}(\alpha_{a_j})
[\partial_\phi^2\log  Q_c(x_j|\phi)]_{\phi = \phi_*}
+c\int_{\mathcal{D}}\dd x\,[\partial_{\phi}^2\log Q_c(x|\phi)]_{\phi=\phi_*}.
\end{equation}
The pairs $(x_j,a_j)$ are independent samples generated according to the
procedure in Def.~\ref{def:NME}, so in the limit $M\to\infty$, by the law of large numbers the mean above converges to its expectation value
\begin{align}
\E[\ell''(\phi_*|\{x_j,a_j\})]&=\|\alpha\|_1\sum_{a=0}^{r-1}
\frac{|\alpha_a|}{\|\alpha\|_1}
\int_{\mathcal D}\dd x \,P_a(x)\mathrm{sgn}(\alpha_{a})
[\partial_\phi^2\log  Q_c(x|\phi)]_{\phi = \phi_*}
+c\int_{\mathcal{D}}\dd x\,[\partial_{\phi}^2\log Q_c(x|\phi)]_{\phi=\phi_*}
\\
&=
\int_{\mathcal D}\dd x \sum_{a=0}^{r-1} \alpha_a \,P_a(x)
[\partial_\phi^2\log  Q_c(x|\phi)]_{\phi = \phi_*}
+c\int_{\mathcal{D}}\dd x\,[\partial_{\phi}^2\log Q_c(x|\phi)]_{\phi=\phi_*}
\\
&=\int_{\mathcal D}\dd x \,P_c(x)
[\partial_\phi^2\log  Q_c(x|\phi)]_{\phi = \phi_*}\\
&\equiv \bar\ell''(\phi_*).
\end{align}
Here, $P_c(x)$ is the regularized distribution defined in Lemma~\ref{lem:m-projection-convergence},
and $\bar\ell(\phi):= \int_{\mathcal D}\dd x \,P_c(x)\log  Q_c(x|\phi)$ is the objective function in the limit $M\to\infty$.

Next consider the numerator
\begin{equation}\label{eq:nme_numerator}
\ell'(\phi_*|\{x_j,a_j\})
=
\frac{\|\alpha\|_1}{M}
\sum_{j=1}^M
\mathrm{sgn}(\alpha_{a_j})
[\partial_\phi\log  Q_c(x_j|\phi)]_{\phi = \phi_*} +c\int_{\mathcal{D}}\dd x\,[\partial_{\phi}\log Q_c(x|\phi)]_{\phi=\phi_*}.
\end{equation}
This is again a mean of independent random variables, and it converges to
\begin{align}
\E[ \ell'(\phi_*|\{x_j,a_j\})]&=\|\alpha\|_1\sum_{a=0}^{r-1}\frac{|\alpha_a|}{\|\alpha\|_1}
\int_{\mathcal D}\dd x \,\mathrm{sgn}(\alpha_a)P_a(x)
[\partial_\phi\log  Q_c(x|\phi)]_{\phi = \phi_*} +c\int_{\mathcal{D}}\dd x\,[\partial_{\phi}\log Q_c(x|\phi)]_{\phi=\phi_*}\label{eq:nme_a1_midpoint}\\
&=
\int_{\mathcal D}\dd x \,P_c(x)
[\partial_\phi\log  Q_c(x|\phi)]_{\phi = \phi_*}\\
&=\left[ \partial_\phi\int_{\mathcal D}\dd x \,P_c(x)
\log  Q_c(x|\phi) \right]_{\phi = \phi_*}\\
&=\bar\ell'(\phi_*).
\end{align}
By definition $\phi_*$ is the maximum of $\bar\ell(\phi)$, and satisfies the stationary point condition $\bar\ell'(\phi_*) = 0$. Therefore the expectation value of the numerator vanishes.
To calculate the variance, we note that the second term of Eq.~\eqref{eq:nme_numerator} is not a random variable, and so 
\begin{align}
    \mathrm{Var}[\ell'(\phi_*|\{x_j,a_j\})] &= \mathrm{Var}\Big[\frac{\|\alpha\|_1}{M}\sum_{j=1}^M\mathrm{sgn}(\alpha_{a_j})[\partial_{\phi}\log Q_c(x_j|\phi)]_{\phi=\phi_*}\Big]\\
    &=\frac{1}{M} \mathrm{Var}\left[\|\alpha\|_1\mathrm{sgn}(\alpha_{a})[\partial_{\phi}\log Q_c(x|\phi)]_{\phi=\phi_*} \right] \\
    &=\frac{1}{M}\left(\mathbb{E}\Big[\Big(\|\alpha\|_1\mathrm{sgn}(\alpha_{a})[\partial_{\phi}\log Q_c(x|\phi)]_{\phi=\phi_*}\Big)^2\Big] - \mathbb{E}\Big[\|\alpha\|_1\mathrm{sgn}(\alpha_{a})[\partial_{\phi}\log Q_c(x|\phi)]_{\phi=\phi_*}\Big]^2\right)\\
    &\leq\frac{1}{M}\mathbb{E}\Big[\Big(\|\alpha\|_1\mathrm{sgn}(\alpha_{a})[\partial_{\phi}\log Q_c(x|\phi)]_{\phi=\phi_*}\Big)^2\Big].
\end{align}
We can evaluate the this as
\begin{align}
\mathrm{Var}[\ell'(\phi_*|\{x_j,a_j\})] 
&\leq\frac{1}{M}\mathbb{E}\Big[\Big(\|\alpha\|_1\mathrm{sgn}(\alpha_{a})[\partial_{\phi}\log Q_c(x|\phi)]_{\phi=\phi_*}\Big)^2\Big]\\
&\leq\frac{\|\alpha\|^2_1}{M}\mathbb{E}\Big[[\partial_{\phi}\log Q_c(x|\phi)]_{\phi=\phi_*}^2\Big]\\
&= \frac{\|\alpha\|^2_1}{M}
\int_{\mathcal D}\dd x \, R(x)
\left[
\partial_\phi\log  Q_c(x|\phi)\right]_{\phi = \phi_*}^2
\end{align}
and so by the central limit theorem,
\begin{equation}
\sqrt{M}\,
\ell'(\phi_*|\{x_j,a_j\})
\end{equation}
converges in distribution to a Gaussian with mean $0$ and variance
\begin{equation}
\leq \|\alpha\|_1^2 \int_{\mathcal D}\dd x \,R(x)
[\partial_\phi\log  Q_c(x|\phi)]_{\phi = \phi_*}^2.
\end{equation}

Now we combine these results with Eq.~\eqref{eq:stationary-point}.
Since the maximizer of $\int_{\mathcal D}\dd x \,P_c(x)
\log  Q_c(x|\phi)$ is unique and lies in the interior of $D_\phi$, the estimator converges to $\tilde\phi \to \phi_*$ as $M\to\infty$ and the Taylor expansion becomes valid.
Then, applying Slutsky's theorem, $\sqrt{M}(\tilde\phi-\phi_*)$ converges in distribution to a Gaussian with mean $0$ and variance
\begin{equation}
\epsilon^2
\leq
\|\alpha\|_1^2\frac{
\int_{\mathcal D}\dd x \, R(x)
[\partial_\phi\log  Q_c(x|\phi)]_{\phi =\phi_*}^2
}{
\left(
\int_{\mathcal D}\dd x \,P_c(x)
[-\partial_\phi^2\log  Q_c(x|\phi)]_{\phi = \phi_*}
\right)^2
}.
\end{equation}
\qed

\subsection{Proof of Lemma~\ref{lem:m-projection-expansion}}

We first analyze the bias in the moment projection estimator (Def.~\ref{def:m-projection-estimator}).

\noindent By definition $\phi_*$ maximizes  $\bar\ell(\phi)=\int_{\mathcal D}\dd x \,P(x)\log Q(x|\phi)$, and therefore satisfies the stationary point condition 
\begin{equation}\label{eq:stationary-point-expansion}
    0=\bar\ell'(\phi_*)=  \int_{\mathcal{D}} \dd x (Q(x|\phi_0)+h(x))\left[\partial_\phi\log Q(x|\phi) \right]_{\phi = \phi_*}
\end{equation}
When $h(x)=0$,  $\phi_* = \argmin_\phi D_{KL}(Q(x|\phi_0)||Q(x|\phi)) = \phi_0$.
Therefore for small $\lVert h \rVert$ we can expand the above equation around $\phi_0$:
\begin{equation}
    0 =  \int_{\mathcal{D}} \dd x (Q(x|\phi_0)+h(x)) \big([\partial_\phi \log Q(x|\phi)]_{\phi = \phi_0} + (\phi_*-\phi_0) [\partial^2_\phi \log Q(x|\phi)]_{\phi = \phi_0}\big) + O((\phi_*-\phi_0)^2)
\end{equation}
Rearranging this formula obtains
\begin{equation}\label{eq:thm12_midpoint_eq}
    (\phi_*-\phi_0)  =-\frac{\int_{\mathcal{D}}\dd  x \,(Q(x|\phi_0)+h(x))[\partial_\phi \log Q(x|\phi)]_{\phi = \phi_0}}{\int_{\mathcal{D}}(Q(x|\phi_0)+h(x))[\partial^2_\phi \log Q(x|\phi)]_{\phi = \phi_0}} + O((\phi_*-\phi_0)^2).
\end{equation}
Then, since $Q(x|\phi)$ is normalised, we have that
\begin{equation}
\label{eq:derivative-normalisation}
\int_{\mathcal{D}}\dd  x\, Q(x|\phi_0)[\partial_\phi \log Q(x|\phi)]_{\phi = \phi_0} = \int_{\mathcal{D}}\dd  x \, [\partial_\phi Q(x|\phi)]_{\phi = \phi_0}(x) =0,
\end{equation}
so the numerator is
\begin{equation}
\int_{\mathcal D}\dd x \,\,h(x)
[\partial_\phi\log Q(x|\phi)]_{\phi=\phi_0} \leq \|h\|_1
\max_x
[\partial_\phi\log Q(x|\phi)]_{\phi=\phi_0}.
\end{equation}
We expand the denominator of Eq.~\eqref{eq:thm12_midpoint_eq} to lowest order in $\|h\|_1=\int_{\mathcal{D}}|h(x)|\dd x \,$ as
\begin{align}
\label{eq:fisher-info-equivalent-formula}
\int \dd x \,\,Q(x|\phi_0)[-\partial_\phi^2\log Q(x|\phi)]_{\phi = \phi_0}
&=
\int \dd x \,\,Q(x|\phi_0)\left(\left(\frac{[\partial_\phi Q(x|\phi)]_{\phi = \phi_0}}{Q_{\phi_0(x)}}\right)^2-\frac{[\partial_\phi^2Q(x|\phi)]_{\phi = \phi_0}}{Q_{\phi_0(x)}}\right)
\\
&=
\int \dd x \,\,Q(x|\phi_0)[\partial_\phi\log Q(x|\phi)]_{\phi = \phi_0}^2 -\partial_\phi^2\left[ \int \dd x \,\,Q(x|\phi_0)\right]_{\phi = \phi_0}
\\
&=\mathcal{I}_0.
\end{align}
Here, the second term of the second-to-last-line evaluates to zero again by the normalisation of $Q(x|\phi)$ (Eq.~\eqref{eq:derivative-normalisation}).
We can simplify
\begin{equation}\label{eq:mproj-bias-repeat}
    |\phi_*-\phi_0| \leq \|h\|_1\frac{\max_x [\partial_\phi\log Q(x|\phi)]_{\phi=\phi_0}}{\mathcal{I}_0} +O(\lVert h\rVert^2),
\end{equation}
where $\mathcal{I}_0$ is the Fisher information of $Q$ with respect to $\phi$, which is exactly Eq.~\eqref{eq:mproj-bias}

We now bound the variance $\epsilon^2$,
\begin{equation}
\epsilon^2
=
\frac{
\int \dd x \,P(x)
(\partial_\phi\log Q(x|\phi_*))^2
}{
\left(
\int \dd x \,P(x)[-\partial_\phi^2\log Q(x|\phi_*)]
\right)^2
}.
\end{equation}
Since $\phi_*=\phi_0+O(\|h\|_1)$ we may replace $\phi_*$ with $\phi_0$
up to $O(\|h\|_1)$ corrections.
\begin{equation}
\epsilon^2
=
\frac{
\int \dd x \,P(x)
(\partial_\phi\log Q(x|\phi_0))^2
}{
\left(
\int \dd x \,P(x)[-\partial_\phi^2\log Q(x|\phi_0)]
\right)^2
} + O(\|h\|_1).
\end{equation}
As in Eq.~\eqref{eq:fisher-info-equivalent-formula}, the leading order of the denominator in $\|h\|_1$ is $\mathcal{I}_0^2$.
We finally recover Eq.~\eqref{eq:mproj-var} by writing $P(x) = Q(x|\phi_0) + O(\|h\|_1)$ and collecting terms.\qed

\subsection{Proof of Lemma~\ref{lem:nme-expansion}}
Our proof for Lemma~\ref{lem:nme-expansion} follows similarly to the proof of Lemma~\ref{lem:m-projection-expansion}; the calculation for the bias follows Eq.~\eqref{eq:stationary-point}-\eqref{eq:mproj-bias-repeat} under the substitutions $Q(x|\phi)\rightarrow Q_c(x|\phi)$, $P(x)\rightarrow P(x)+c$.
However, our calculation for the variance does not simplify in the same way, as we integrate the numerator over the  marginal distribution $R(x)$, which we expect to be significantly different to $P_c(x)$.
Instead, because $R(x)$ is a valid, normalized probability distribution, we can bound the numerator as
\begin{equation}
\int \dd x \,R(x)
[\partial_\phi\log  Q_c(x|\phi)]_{\phi = \phi_0}]^2
\le
\max_x
[\partial_\phi\log  Q_c(x|\phi)]^2_{\phi=\phi_0},
\end{equation}
and identify within the denominator
\begin{equation}
    \int_{\mathcal{D}}\dd  x \, P_c(x)[-\partial^2_{\phi}\log Q_c(x|\phi)]_{\phi=\phi_*}=\mathcal{I}_c+O(\|h\|_1),
\end{equation}
from which the result follows.~\qed

\starttocentries

\section{Proofs of Lemmas~\ref{lem:gaussian-no-noise} and~\ref{lem:gaussian-gdn}}
\label{app:proof-of-gaussian-cors}

\stoptocentries

In this appendix we calculate the application of the result of Lemma~\ref{lem:m-projection-expansion} to the case of a Gaussian kernel, both without noise and with global depolarizing noise.
First, we note that for a Gaussian kernel function, the regularity assumptions in Lemma~\ref{lem:m-projection-expansion} are satisfied in both cases, so we can use the bounds in this lemma. 
Our target in both cases is to bound the bias $b=|\phi_*-\phi_0|$ [Eq.~\eqref{eq:mproj-bias}] and the variance $\lim_{M\rightarrow\infty} M\mathrm{Var}[\tilde{\phi}]$ [Eq.~\eqref{eq:mproj-var}].
We achieve this bound piecewise.
First, in step~\ref{app:h_bound} we bound the $1$-norm of the the model deviation $h$ from above by using the assumption on the distance of spurious phases from the filtering interval [Eq.~\eqref{eq:phi1_buffer-fixed}].
Then, in step~\ref{app:normalization_bound} we bound the normalization constant that arises from cutting the Gaussian model distribution to a finite filtering interval, and its derivative with respect to $\phi_0$, using the existence of gaps between the edges of the filtering interval and the target phase [Eq.~\eqref{eq:phi0_buffer-fixed}].
Next, in  we derive an upper bound on the score of the model distribution (step~\ref{app:score_bound}) and a lower bound on the Fisher information (steps~\ref{app:fi_bound}, \ref{App:f_bound}).
This is finally summarized in step~\ref{app:final_bounds}, where we derive the explicit forms of the bounds appearing in Lemma~\ref{lem:gaussian-no-noise} and Lemma~\ref{lem:gaussian-gdn}.

Before we begin, let us define some notation to be used in the rest of the section.
The model distribution used to define the moment projection estimator is
\begin{align}
    Q(x|\phi) &= \frac{q(x|\phi)}{\int_{\mathcal{D}}\dd x \, q(x|\phi)}\\
    \label{eq:q_phi_gaussian_with_gdn}
    q(x|\phi) &= Fa_0g_\sigma(x-\phi) + (1-F)\frac{\mathcal{M}_\sigma}{2\pi},\\
    %u(x) &= F\sum_{j\neq0} a_j g_\sigma(x-\phi_0),\\
    g_\sigma(x) &= \frac{1}{\sqrt{2\pi}\sigma}e^{-\frac{x^2}{2\sigma^2}}. \\
    \mathcal{M}_\sigma &= \int_{-\pi}^{\pi}g_\sigma(x)\dd x = \erf\left(\frac{\pi}{\sqrt2}\sigma^{-1}\right).
\end{align}
Note that the case of no noise is included by taking $F = 1$.
The model $q(x|\phi)$ is defined so that $q(x|\phi) / \mathcal{M}_\sigma$ matches $p(x)$ in the case of $a_0 = 1$, and $\mathcal{N}/\mathcal{M}_\sigma$ is equal to $P_A$ in Eq.~\eqref{eq:p_accept}.

For brevity of notation, let $\mathcal{N}$ be the normalization of $Q(x|\phi_0)$
\begin{equation}
    \mathcal{N} = \int_\mathcal{D} q(x|\phi_0)\dd x ,
\end{equation}
and let $(.)'$ denote derivative with respect to $\phi$,
\begin{align}
    q'_{\phi_0}(x) &=  \left[ \partial_\phi q(x|\phi) \right]_{\phi = \phi_0},\\
    \mathcal{N}' &= \left[ \partial_\phi \int_\mathcal{D} q(x|\phi)\dd x \right]_{\phi = \phi_0} = \int_{\mathcal{D}}q'_{\phi_0}(x) \dd x.\\
\end{align}

\subsection{Step 1: Bound on $\lVert h \rVert$}
\label{app:h_bound}

In order to bound the $1$-norm of the model deviation, $h(x)=P(x)-Q(x|\phi_0)$, we calculate a bound on a related distribution, $\lVert u \rVert$, defined by
\begin{equation}
\label{eq:u_distibution}
        u(x) = F\sum_{j\neq0} a_j g_\sigma(x-\phi_j)\\ 
\end{equation}
so the true probability distribution is
\begin{equation}
    P(x) = \frac{q(x|\phi_0)+u(x)}{\mathcal{N}+\lVert u \rVert}.
\end{equation}
Using the triangle inequality (noting that $u(x),q(x|\phi_0)>0$), we can bound $\lVert h \rVert $ as
\begin{align}
    \lVert h \rVert &= \int_{\mathcal{D}}\dd  x \ |P(x) - Q(x|\phi_0)|\\
    &= \int_{\mathcal{D}}\dd  x \ \left|\frac{q(x|\phi_0)+u(x)}{\mathcal{N}+\lVert u \rVert} -\frac{q(x|\phi_0)}{\mathcal{N}}\right|\\
    &= \frac{\int_{\mathcal{D}}\dd  x \ |\mathcal{N}u(x)-\lVert u \rVert q(x|\phi_0)|}{\mathcal{N}(\mathcal{N}+\lVert u \rVert)}\\
    &\leq \frac{\int_{\mathcal{D}}\dd  x \ (\mathcal{N}u(x)+\lVert u \rVert q(x|\phi_0))}{\mathcal{N}(\mathcal{N}+\lVert u \rVert)}\\
    &= \frac{2\mathcal{N}\lVert u \rVert }{\mathcal{N}(\mathcal{N}+\lVert u \rVert)}\\
    &\leq 2 \frac{\lVert u \rVert}{\mathcal{N}}.
    \label{eq:h-norm-u-norm-bound}
\end{align}

Now, we use the assumption on the distance of the spurious phases $\phi_{j\neq0}$ from the filtering region,
\begin{equation}
    \min_{j\neq0}\min_{x\in\mathcal{D}}|x-\phi_j|\geq d
\end{equation}
to bound
\begin{equation}
\label{eq:u-norm}
        \lVert u \rVert = F\sum_{j\neq0} a_j G(\phi_j). 
\end{equation}
We can bound
\begin{align}
    G_\sigma(\phi_j) &= \int_{\mathcal{D}}g_\sigma(x-\phi) \dd x \leq\int_d^{\infty}g_{\sigma}(y)dy\leq\frac{\sigma^2}{d}g_{\sigma}(d),
\end{align}
where the last inequality is obtained from standard bounds on Mill's ratio for a normal distribution.
Then, substituting this into Eq.~\eqref{eq:u-norm} yields
\begin{align}
    \lVert u \rVert 
    &\leq F(\sum_{j\neq0} a_j) \frac{\sigma^2}{d} g_\sigma(d)\\
    &= F(1-a_0)\frac{\sigma^2}{d} g_\sigma(d).
\end{align}
Combining the above with Eq.~\eqref{eq:h-norm-u-norm-bound} yields
\begin{equation}
\label{eq:h-norm-bound}
    \lVert h \rVert \leq 2\mathcal{N}^{-1}F(1-a_0)\frac{\sigma^2}{d} g_\sigma(d).
\end{equation}

\subsection{Step 2: Bounds on $\mathcal{N}$, $\mathcal{N}'$}
\label{app:normalization_bound}

In this step we bound the norm $\mathcal{N}$ and its derivative $\mathcal{N}'$ using the assumption that $\phi_0$ is well inside the filtering region $\mathcal{D} = [\phi_{guess}-|\mathcal{D}|/2, \phi_{guess}+|\mathcal{D}|/2]$:
\begin{equation}
\label{eq:phi0-phiguess-distance}
    \left|\phi_{guess}-\phi_0 \right| \leq \frac{|\mathcal{D}|}{3}.
\end{equation}

Using the definition of the Gaussian model $q(x|\phi)$ in Eq.~\eqref{eq:q_phi_gaussian_with_gdn}, we can evaluate
\begin{align}
    \mathcal{N} &= Fa_0G_\sigma(\phi) + (1-F)\frac{\mathcal{M}_\sigma |\mathcal{D}|}{2\pi},\\
    %&\geq Fa_0 G_\sigma(\phi)\\
    \mathcal{N}' &= Fa_0 (g_\sigma(\phi_{guess}-\phi_0 - |\mathcal{D}|/2)-g_\sigma(\phi_{guess}-\phi_0 +|\mathcal{D}|/2)).\\
\end{align}
The normalization $\mathcal{N}$ is minimal when $|\phi_{guess} - \phi_0|$ is maximal, for $\phi_0 = \phi_{guess} + |\mathcal{D}|/3$, and
\begin{align}
    \mathcal{N}&\geq Fa_0G_\sigma(\phi)\\
    &\geq Fa_0G_\sigma(\phi_{guess}+|\mathcal{D}|/3)\\
    &=Fa_0 \frac{\erf(\frac{5}{6\sqrt{2}}\frac{|\mathcal{D}|}{\sigma})+\erf(\frac{1}{6\sqrt{2}}\frac{|\mathcal{D}|}{\sigma})}{2}.
\label{eq:N-lower-bound}
\end{align}
Conversely, the derivative of the normalization reaches a maximum when $|\phi_{guess} - \phi_0|$ is maximal, for $\phi_0 = \phi_{guess} + |\mathcal{D}|/3$ (to convince oneself one can check that $\mathcal{N}''<0$ for $\phi_0\in\mathcal{D}$).
Thus, we can bound
\begin{align}
    |\mathcal{N}'|&=Fa_0 G'(\phi_0)\\
    &\leq Fa_0 (g_\sigma\left(\frac{|\mathcal{D}|}{6}\right)-g_\sigma\left(\frac{5|\mathcal{D}|}{6}\right))\\
    &\leq Fa_0 g_\sigma\left(\frac{|\mathcal{D}|}{6}\right)\\
    &= \frac{Fa_0}{\sigma} g_1\left(\frac{|\mathcal{D}|}{6\sigma}\right).
    \label{eq:N'-upper-bound}
\end{align}

We can also bound the norm from above as
\begin{equation}
\label{eq:N-upper-bound}
    \mathcal{N} \leq Fa_0 + (1-F)\frac{1}{2\pi}
\end{equation}
since $G_\sigma(\phi)$ and $\mathcal{M}_\sigma$ are both bounded by 1 for all values of $\phi$ and $\sigma$.

\subsection{Step 3: Bound on $\max_x \left[ \partial_\phi \log Q(x|\phi)\right]_{\phi = \phi_0}$ }
\label{app:score_bound}

In this step we bound $\left[ \partial_\phi \log Q(x|\phi)\right]_{\phi = \phi_0}$ that appears in the numerator in the first order of the bias bound in Eq.~\eqref{eq:mproj-bias}.

First, using notation introduced at the start of the section, we have
\begin{equation}
    \left[ \partial_\phi \log Q(x|\phi)\right]_{\phi = \phi_0} = \frac{q'_{\phi_0}(x)}{q(x|\phi_0)}-\frac{\mathcal{N}'}{\mathcal{N}}.
\end{equation}
Hence we can bound the numerator in Eq.~\eqref{eq:mproj-bias} as
\begin{equation}
\label{eq:logderivative_bound_with_normalisation}
    \max_x |\partial_\phi \log Q(x|\phi)| \leq  \max_x \left|\frac{q'_{\phi_0}(x)}{q(x|\phi_0)}\right| +  \frac{|\mathcal{N}'|}{\mathcal{N}}.
\end{equation}
We can bound the second term in the above expression using Eq.~\eqref{eq:N'-upper-bound} and Eq.~\eqref{eq:N-lower-bound} in Step \ref{app:normalization_bound} as
\begin{equation}
    \frac{|\mathcal{N}'|}{\mathcal{N}} \leq \frac{1}{\sigma} \frac{g_1\left(\frac{|\mathcal{D}|}{6\sigma}\right)}{\frac{1}{2}\left(\erf(\frac{5}{6\sqrt{2}}\frac{|\mathcal{D}|}{\sigma})+\erf(\frac{1}{6\sqrt{2}}\frac{|\mathcal{D}|}{\sigma})\right)}.
\end{equation}
As this is monotonically decreasing, for $\sigma \leq |\mathcal{D}|/6$ we can bound this as
\begin{equation}
    \frac{|\mathcal{N}'|}{\mathcal{N}} \leq \frac{|\mathcal{D}|}{\sigma^2} \frac{1}{6}\frac{g_1\left(1\right)}{\frac{1}{2}\left(\erf(\frac{5}{\sqrt{2}})+\erf(\frac{1}{\sqrt{2}})\right)} \lesssim 0.05\frac{|\mathcal{D}|}{\sigma^2}.
\end{equation}

To bound the first term,
inserting the definition of the Gaussian $g_{\sigma}(\phi)$ into $q$ yields
\begin{equation}
    \left|\frac{q'_{\phi_0}(x)}{q(x|\phi_0)}\right| = \frac{|x-\phi_0|}{\sigma^2}
    \frac{g_\sigma(x-\phi_0)}{{g_\sigma(x-\phi_0)+\frac{\mathcal{M}_\sigma}{2\pi}\frac{1-F}{Fa_0}}} \leq \frac{|x-\phi_0|}{\sigma^2}.
\end{equation}
Using assumption \eqref{eq:phi0-phiguess-distance}, we can bound this as
\begin{align}
    \max_x \left|\frac{q'_{\phi_0}(x)}{q(x|\phi_0)}\right| \leq\frac{\max_x|x-\phi_0|}{\sigma^2} \leq \frac{5|\mathcal{D}|}{6\sigma^2},
\end{align}
leading to a final bound
\begin{equation}
\label{eq:dlogQ_bound}
    \left|\max_x \left[ \partial_\phi \log Q(x|\phi)\right]_{\phi = \phi_0}\right| \leq |\mathcal{D}| \sigma^{-2}.
\end{equation}

\subsection{Step 4: $\Omega(\sigma^{-2})$ bound on the Fisher Information $\mathcal{I}_0$}
\label{app:fi_bound}

Next, we can calculate the Fisher information of the filtered distribution, $\mathcal{I}_0$ (see Eq.~\eqref{eq:fisher-info}):
\begin{align}
    \mathcal{I}_0 &= \int_{\mathcal{D}}\dd  x \,\ Q(x|\phi_0) \left[\partial_\phi \log Q(x|\phi)\right]^2_{\phi = \phi_0}\\
    &= \int_{\mathcal{D}}\dd  x\ \frac{q(x|\phi_0)}{\mathcal{N}} \left(\frac{q'_{\phi_0}(x)}{q(x|\phi_0)}-\frac{\mathcal{N}'}{\mathcal{N}}\right)^2\\
    &= \int_{\mathcal{D}}\dd  x\left(\frac{q'_{\phi_0}(x)^2}{\mathcal{N}q(x|\phi_0)}
    -2\frac{\mathcal{N}'}{\mathcal{N}^2}q'_{\phi_0}(x)+\frac{\mathcal{N}'^2 }{\mathcal{N}^3}q(x|\phi_0)\right)\\
    &= \frac{1}{\mathcal{N}}\left(\int_{\mathcal{D}}\dd  x\frac{q'_{\phi_0}(x)^2}{q(x|\phi_0)}
    -2\frac{\mathcal{N}'^2}{\mathcal{N}}+\frac{\mathcal{N}'^2 }{\mathcal{N}}\right)\\
    &= \frac{1}{\mathcal{N}}\left(\int_{\mathcal{D}}\dd  x\frac{q'_{\phi_0}(x)^2}{q(x|\phi_0)}
    -\frac{\mathcal{N}'^2 }{\mathcal{N}}\right).
\end{align}

For the Gaussian model, we can bound the boundary term using Eq.~\eqref{eq:N'-upper-bound} and Eq.~\eqref{eq:N-lower-bound} in Step \ref{app:normalization_bound} as
\begin{equation}
    \frac{\mathcal{N}'^2 }{\mathcal{N}} \leq \frac{F a_0}{\sigma^2} \frac{g_1^2\left(\frac{|\mathcal{D}|}{6\sigma}\right)}{\frac{1}{2}\left(\erf(\frac{5}{6\sqrt{2}}\frac{|\mathcal{D}|}{\sigma})+\erf(\frac{1}{6\sqrt{2}}\frac{|\mathcal{D}|}{\sigma})\right)}
\end{equation}
and rewrite the integral in the first term as
\begin{align}
    \label{eq:fi_integral}
    \int_{\mathcal{D}}\dd  x\frac{q'_{\phi_0}(x)^2}{q(x|\phi_0)}
    &= \frac{F^2a_0^2}{\sigma^4}\int_{\phi_{guess}- |\mathcal{D}|/2}^{\phi_{guess}+ |\mathcal{D}|/2} \frac{g_\sigma^2(x-\phi_0)(x-\phi_0)^2 \dd x}{Fa_0 g_\sigma(x-\phi_0) + (1-F)\frac{\mathcal{M}_\sigma}{2\pi}}\\
    &= \frac{Fa_0}{\sigma^4}\int_{\phi_{guess}- |\mathcal{D}|/2}^{\phi_{guess}+ |\mathcal{D}|/2} \frac{g_\sigma(x-\phi_0)(x-\phi_0)^2 \dd x}{ 1 + \frac{1-F}{Fa_0}\frac{\mathcal{M}_\sigma}{2\pi}\frac{1}{g_\sigma(x-\phi_0)}}\\
    &= \frac{Fa_0}{\sigma^2}\int_{\frac{\phi_{guess}-\phi_0 - |\mathcal{D}|/2}{\sigma}}^{\frac{\phi_{guess}-\phi_0 + |\mathcal{D}|/2}{\sigma}} \frac{g_1(\xi)\xi^2 \dd\xi}{1 + \sigma \frac{\mathcal{M}_\sigma}{2\pi}\frac{1-F}{Fa_0}\frac{1}{g_1(\xi)}},
\end{align}

To bound $\mathcal{I}_0$, notice that assumption in Eq.~\eqref{eq:phi0-phiguess-distance} implies that either $[\phi_0 -|\mathcal{D}|/6, \phi_0 + |\mathcal{D}|/2] \subset \mathcal{D}$ or $[\phi_0 -|\mathcal{D}|/2, \phi_0 + |\mathcal{D}|/6] \subset \mathcal{D}$.
Since the integrand in Eq.~\eqref{eq:fi_integral} is always positive, we can bound the integral by an integral over a subset of $\mathcal{D}$,
\begin{equation} \label{eq:fi-subset-integral-bound}
    \int_{\mathcal{D}}\dd  x\frac{q'_{\phi_0}(x)^2}{q(x|\phi_0)} \geq \frac{Fa_0}{\sigma^2}\int_{-\frac{|\mathcal{D}|}{6\sigma}}^{\frac{|\mathcal{D}|}{2\sigma}} \frac{g_1(\xi)\xi^2 d\xi}{1 + \sigma \frac{\mathcal{M}_\sigma}{2\pi}\frac{1-F}{Fa_0}\frac{1}{g_1(\xi)}}.
\end{equation}

This yields a lower bound on the Fisher information:
\begin{align}
\label{eq:FI-lower-bound}
    \mathcal{N}\mathcal{I}_0 &\geq
    \frac{F a_0}{\sigma^2}\left[\int_{-\frac{|\mathcal{D}|}{6\sigma}}^{\frac{|\mathcal{D}|}{2\sigma}} \frac{g_1(\xi)\xi^2 d\xi}{1 + \sigma \frac{\mathcal{M}_\sigma}{2\pi}\frac{1-F}{Fa_0}\frac{1}{g_1(\xi)}}-\frac{g_1^2\left(\frac{|\mathcal{D}|}{6\sigma}\right)}{\frac{1}{2}\left(\erf(\frac{5}{6\sqrt{2}}\frac{|\mathcal{D}|}{\sigma})+\erf(\frac{1}{6\sqrt{2}}\frac{|\mathcal{D}|}{\sigma})\right)}\right].
\end{align}

\subsection{Step 5: Upper bound on $(\mathcal{I}_0)^{-1}$}
\label{App:f_bound}

In order to use the lower bound in Eq.~\eqref{eq:FI-lower-bound} to upper bound $(\mathcal{I}_0)^{-1}$, we need to make sure that the RHS is positive.

For $F=1$, the RHS is positive for $\sigma < |D|/2$, and the term inside the brackets is monotonically increasing to $\int_{-\infty}^{\infty}g_1(\xi)\xi^2 d\xi = 1$. For any $\sigma < |D|/6$ it is $\gtrsim 0.5$, so we get
\begin{equation}
\label{eq:inverse-fi-bound-no-noise}
    \frac{1}{\mathcal{N}\mathcal{I}_0} \leq 2\frac{\sigma^2}{a_0} \quad \text{for } \sigma \leq |\mathcal{D}|/3, F=1.
\end{equation}

For $F<1$: assume $\sigma < |D|/6$, then we can again bound the integral by the integral of the subset $[-1, 1] \subset [-\frac{|\mathcal{D}|}{6\sigma},\frac{|\mathcal{D}|}{2\sigma}]$:
\begin{equation}
    \int_{-\frac{|\mathcal{D}|}{6\sigma}}^{\frac{|\mathcal{D}|}{2\sigma}} \frac{g_1(\xi)\xi^2 d\xi}{1 + \sigma \frac{\mathcal{M}_\sigma}{2\pi}\frac{1-F}{Fa_0}\frac{1}{g_1(\xi)}} 
    \geq \int_{-1}^{1} \frac{g_1(\xi)\xi^2 d\xi}{1 + \sigma \frac{\mathcal{M}_\sigma}{2\pi}\frac{1-F}{Fa_0}\frac{1}{g_1(\xi)}}.
\end{equation}
Within the subset, $|\xi| < 1$ and so $g_1(\xi) \geq g_1(1) = (2\pi e)^{-1/2}$, so we can further bound
\begin{align}
    \int_{-1}^{1} \frac{g_1(\xi)\xi^2 d\xi}{1 + \sigma \frac{\mathcal{M}_\sigma}{2\pi}\frac{1-F}{Fa_0}\frac{1}{g_1(\xi)}}
    &\geq \int_{-1}^{1} \frac{g_1(\xi)\xi^2 d\xi}{1 + \sigma \frac{\mathcal{M}_\sigma}{2\pi}\frac{1-F}{Fa_0}\frac{1}{g_1(1)}}\\
    &=  \frac{\int_{-1}^{1}g_1(\xi)\xi^2 d\xi}{1 + \sigma\frac{\mathcal{M}_\sigma}{2\pi}\frac{1-F}{Fa_0}\frac{1}{g_1(1)}}
\end{align}
We can also bound $\sigma \mathcal{M}_\sigma$ in the denominator as $\leq |\mathcal{D}|/6$, and finally we get
\begin{equation}
    \int_{\mathcal{D}}\dd  x\frac{q'_{\phi_0}(x)^2}{q(x|\phi_0)} \geq \frac{Fa_0}{\sigma^2}\frac{\int_{-1}^{1}g_1(\xi)\xi^2 d\xi}{1 + \frac{|\mathcal{D}|}{2\pi}\frac{1-F}{Fa_0}\frac{1}{6g_1(1)}}.
\end{equation}
Inserting $\sigma < \mathcal{D}/6$ into the denominator of the boundary term yields
\begin{equation}
    \mathcal{N}\mathcal{I}_0 \geq \frac{Fa_0}{\sigma^2}\left(\frac{\int_{-1}^{1}g_1(\xi)\xi^2 d\xi}{1 + \frac{|\mathcal{D}|}{2\pi}\frac{1-F}{Fa_0}\frac{1}{6g_1(1)}} -\frac{g_1^2\left(\frac{|\mathcal{D}|}{6\sigma}\right)}{\frac{1}{2}\left(\erf(\frac{5}{\sqrt{2}})+\erf(\frac{1}{\sqrt{2}})\right)}
    \right).
\end{equation}
To ensure the RHS of the inequality above is positive, we take $\sigma$ such that the negative term is at most half of the positive term, 
\begin{equation}
    \sigma \leq \frac{|\mathcal{D}|}{6} \log^{-1/2}\left(\frac{1 + \frac{|\mathcal{D}|}{2\pi}\frac{1-F}{Fa_0}\frac{1}{6g_1(1)}}{{\frac{\pi}{2}\left(\erf(\frac{5}{\sqrt{2}})+\erf(\frac{1}{\sqrt{2}})\right)}\int_{-1}^{1}g_1(\xi)\xi^2 d\xi} \right).
\end{equation}
We get the final bound
\begin{equation}
\label{eq:inverse-fi-bound-with-noise}
\frac{1}{\mathcal{N}\mathcal{I}_0} \leq C_1 \frac{\sigma^2}{Fa_0}\left(1 + C_2\frac{|\mathcal{D}|}{2\pi}\frac{1-F}{Fa_0}\right)
\quad \text{for } \sigma \leq \frac{|\mathcal{D}|}{6} \frac{1}{\sqrt{C_3+\log\left(1 + C_2\frac{|\mathcal{D}|}{2\pi}\frac{1-F}{Fa_0}\right)}}
\end{equation}
with constants
\begin{align}
    C_1 &= 10.1 \gtrsim \frac{1}{\frac{1}{2}\int_{-1}^{1}g_1(\xi)\xi^2 d\xi},\\
    C_2 &= 0.7 \gtrsim \frac{1}{6g_1(1)},\\
    C_3 &= 1.2 \gtrsim -\log\left({{\frac{\pi}{2}\left(\erf(\frac{5}{\sqrt{2}})+\erf(\frac{1}{\sqrt{2}})\right)}\int_{-1}^{1}g_1(\xi)\xi^2 d\xi} \right).
\end{align}

\subsection{Step 6: Final bound on the bias and variance}
\label{app:final_bounds}

Now we are ready to combine all the steps above to get the bound on the bias.
The first order in Eq.~\eqref{eq:mproj-bias} can first be bound using the bound on $\lVert h \rVert$ in Eq.~\eqref{eq:h-norm-bound} from step~\ref{app:h_bound} and the bound on $\left| \max_x [\partial_\phi \log Q(x|\phi)]_{\phi = \phi_0}\right|$ in Eq.~\eqref{eq:dlogQ_bound} in step~\ref{app:score_bound}
    \begin{align}
        \lVert h \rVert \frac{\left| \max_x [\partial_\phi \log Q(x|\phi)]_{\phi = \phi_0}\right|}{\mathcal{I}_0} &\leq 2F(1-a_0) \frac{|\mathcal{D}|}{d}\sigma^2  g_\sigma(d) \frac{1}{\mathcal{N}\mathcal{I}_0\sigma^2}.
    \end{align}
Then we can bound $(\mathcal{N}\mathcal{I}_0)^{-1}$ using Eq.~\eqref{eq:inverse-fi-bound-no-noise} and Eq.~\eqref{eq:inverse-fi-bound-with-noise}
leading to the bounds
\begin{align}
    \lVert h \rVert \frac{\left| \max_x [\partial_\phi \log Q(x|\phi)]_{\phi = \phi_0}\right|}{\mathcal{I}_0} &\leq \frac{1-a_0}{a_0} \frac{|\mathcal{D}|}{d}\sigma^2g_\sigma(d) \times 4 &\quad F=1, \\
    \lVert h \rVert \frac{\left| \max_x [\partial_\phi \log Q(x|\phi)]_{\phi = \phi_0}\right|}{\mathcal{I}_0} &\leq \frac{1-a_0}{a_0} \frac{|\mathcal{D}|}{d}\sigma^2g_\sigma(d) \times 2C_1\left(1 + C_2\frac{|\mathcal{D}|}{2\pi}\frac{1-F}{Fa_0}\right) &\quad F<1. 
\end{align}

To get the variance bound in Lemmas~\ref{lem:gaussian-no-noise} and \ref{lem:gaussian-gdn}, we combine Eq.~\eqref{eq:inverse-fi-bound-no-noise} and Eq.~\eqref{eq:inverse-fi-bound-with-noise} with Eq.~\eqref{eq:N-upper-bound} to get
\begin{align}
    \frac{1}{\mathcal{I}_0} &\leq 2\sigma^2 &\quad F=1,\\
    \frac{1}{\mathcal{I}_0} &\leq C_1 {\sigma^2}{}\left(1 + C_2\frac{|\mathcal{D}|}{2\pi}\frac{1-F}{Fa_0}\right)\left(1 + \frac{1-F}{Fa_0}\frac{1}{2\pi}\right)&\quad F<1.
\end{align}
To bound the number of shots needed $M' = M/P_A$ in Theorem~\ref{thm:fmpe-cost-gdn}, it is useful to also have
\begin{equation}
\label{eq:fisher-info-before-filtering}
    \frac{1}{P_A\mathcal{I}_0} =\frac{\mathcal{M}_\sigma}{\mathcal{N}\mathcal{I}_0}\leq \frac{1}{\mathcal{N}\mathcal{I}_0} \leq C_1 \frac{\sigma^2}{Fa_0}\left(1 + C_2\frac{|\mathcal{D}|}{2\pi}\frac{1-F}{Fa_0}\right),
\end{equation}
as the variance is
\begin{equation}
\label{eq:var-shots}
    \Var[\tilde\phi] \approx \frac{1}{M\mathcal{I}_0} = \frac{1}{M'P_A\mathcal{I}_0}.
\end{equation}

\section{Proof of Lemma~\ref{lem:gaussian-unbiased}}

First, note that since we assume $\sum_{a=0}^{r-1} \alpha_a P_a(x) = \sum_j a_j g_\sigma(x-\phi_j)$, the deviation $h(x)$ in Eq.~\eqref{eq:h-nme} is the same as in Lemma~\ref{lem:gaussian-no-noise}, and we can use the bound derived in step~\ref{app:h_bound} in App.~\ref{app:proof-of-gaussian-cors},
\begin{equation}
    \|h\| \leq 2G_\sigma(\phi_0)^{-1}\frac{(1-a_0)}{a_0}\sigma^2 d^{-1} g_\sigma(d).
\end{equation}.

Secondly, we need to bound $\mathcal{I}_c$ in Eq.~\eqref{eq:fisher-info-nme} from below. 
For our regularised model $Q_c(x|\phi) = \frac{g_\sigma(x-\phi)}{G_\sigma(\phi)} +c$ we have
\begin{align}
    \mathcal{I}_c &= \frac{1}{G_\sigma(\phi_0)} \int_{\mathcal{D}}\frac{\left([\partial_\phi g_\sigma(x-\phi)]_{\phi = \phi_0} - g_\sigma(x-\phi_0)\frac{G'_\sigma(\phi_0)}{G_\sigma(\phi_0)}\right)^2 \dd x}{g_\sigma(x-\phi_0) + c G_\sigma(\phi_0)}\\
    &= \frac{1}{G_\sigma(\phi_0)} \int_{\mathcal{D}}\frac{ g^2_\sigma(x-\phi_0)\left(\frac{x-\phi_0}{\sigma^2} -\frac{G'_\sigma(\phi_0)}{G_\sigma(\phi_0)}\right)^2 \dd x}{g_\sigma(x-\phi_0) + c G_\sigma(\phi_0)},
\end{align}
which we can bound by neglecting a positive term and using $G_\sigma(\phi_0)\leq 1$ as
\begin{align}
    \mathcal{I}_c &\geq \frac{1}{G_\sigma(\phi_0)} \int_{\mathcal{D}}\frac{ g^2_\sigma(x-\phi_0)\left(\frac{x-\phi_0}{\sigma^2} -\frac{G'_\sigma(\phi_0)}{G_\sigma(\phi_0)}\right)^2 \dd x}{g_\sigma(x-\phi_0) + c}\\
    &\geq \frac{1}{G_\sigma(\phi_0)}\frac{1}{\sigma^2} \left(
    \frac{1}{\sigma^2}\int_{\mathcal{D}}\frac{g^2_\sigma(x-\phi_0)(x-\phi_0)^2 \dd x}{g_\sigma(x-\phi_0) + c}
    - 
    2\left|\frac{G'_\sigma(\phi_0)}{G_\sigma(\phi_0)} \int_{\mathcal{D}}\frac{g^2_\sigma(x-\phi_0)(x-\phi_0) \dd x}{g_\sigma(x-\phi_0) + c}\right|
    \right).
    \label{eq:regularised_fi_step1}
\end{align}
We can bound the integral in the negative boundary term using $c\geq 0$ as
\begin{align}
     \left|\int_{\mathcal{D}}\frac{g^2_\sigma(x-\phi_0)(x-\phi_0) \dd x}{g_\sigma(x-\phi_0) + c} \right|&\leq \left|\int_{\mathcal{D}}\frac{g^2_\sigma(x-\phi_0)(x-\phi_0) \dd x}{g_\sigma(x-\phi_0)} \right|\\
     &= \left|\int_{\mathcal{D}}g_\sigma(x-\phi_0)(x-\phi_0) \dd x \right|\\
     &=\sigma^2 |G'_\sigma(\phi_0)|.
\end{align}
Then using the same arguments as in step~\ref{app:normalization_bound} in App \ref{app:proof-of-gaussian-cors} [Eq.~\eqref{eq:N'-upper-bound}], we can further bound the negative term as
\begin{equation}
    2\sigma^2\left|\frac{G'_\sigma(\phi_0)^2}{G_\sigma(\phi_0)}\right| \leq 2g^2_1\left(\frac{|\mathcal{D}|}{6\sigma}\right).
\end{equation}
The dominant first term in Eq.~\eqref{eq:regularised_fi_step1} has the same form as in the unregularised case [Eq.~\eqref{eq:fi_integral}],
\begin{align}
     \frac{1}{\sigma^2}\int_{\mathcal{D}}\frac{g^2_\sigma(x-\phi_0)(x-\phi_0)^2 \dd x}{g_\sigma(x-\phi_0) + c} &= \int_{\frac{\phi_{guess}-\phi_0 - |\mathcal{D}|/2}{\sigma}}^{\frac{\phi_{guess}-\phi_0 + |\mathcal{D}|/2}{\sigma}} \frac{g_1(\xi)\xi^2 d\xi}{1 + c\sigma\frac{1}{g_1(\xi)}}\\
     &\geq \int_{-\frac{|\mathcal{D}|}{6\sigma}}^{\frac{|\mathcal{D}|}{2\sigma}} \frac{g_1(\xi)\xi^2 d\xi}{1 + c\sigma\frac{1}{g_1(\xi)}}.
\end{align}
Using the same arguments as in step~\ref{App:f_bound} in App.~\ref{app:proof-of-gaussian-cors}, we can bound it for $\sigma \leq |\mathcal{D}|/6$ as
\begin{align}
     \int_{-\frac{|\mathcal{D}|}{6\sigma}}^{\frac{|\mathcal{D}|}{2\sigma}} \frac{g_1(\xi)\xi^2 d\xi}{1 + c\sigma\frac{1}{g_1(\xi)}} \geq \frac{\int_{-1}^{1} g_1(\xi)\xi^2 d\xi}{1 + c \frac{|\mathcal{D}|}{6}\frac{1}{g_1(1)}}.
\end{align}
Combining the above steps, we get
\begin{equation}
    \mathcal{I}_c \geq \frac{1}{G_\sigma(\phi_0)}\frac{1}{\sigma^2} \left(\frac{\int_{-1}^{1} g_1(\xi)\xi^2 d\xi}{1 + c \frac{|\mathcal{D}|}{6}\frac{1}{g_1(1)}} - 2g^2_1\left(\frac{|\mathcal{D}|}{6\sigma}\right)\right).
\end{equation}
Again we take $\sigma$ small enough that the negative term is at most half of the positive term,
\begin{equation}
    \sigma \leq \frac{|\mathcal{D}|}{6}\log^{-1/2}\left(\frac{1 + c \frac{|\mathcal{D}|}{6}\frac{1}{g_1(1)}}{4\pi\int_{-1}^{1} g_1(\xi)\xi^2 d\xi}\right)
\end{equation}
to get a bound
\begin{equation}
    \frac{1}{\mathcal{I}_c} \leq {G_\sigma(\phi_0)}{\sigma^2}\times C_1 ({1 + C_2 c{|\mathcal{D}|}})
\end{equation}
where the constants $C_1$, $C_2$ are the same as in step~\ref{App:f_bound} in App.~\ref{app:proof-of-gaussian-cors}.

Thirdly, we need to upper bound $\max_{x\in\mathcal{D}} \left[ \partial_\phi \log Q_c(x|\phi)\right]_{\phi = \phi_0}$.
Importantily, as it appears in the variance in Lemma~\ref{lem:nme-expansion} we need a stronger bound than in step~\ref{app:score_bound} in App.~\ref{app:proof-of-gaussian-cors} to get an $O(\sigma^2)$ bound on the variance, which we get by using regularisation.
We have
\begin{align}
    \max_{x\in\mathcal{D}} \left| \partial_\phi \log Q_c(x|\phi)\right|_{\phi = \phi_0} &=\max_{x\in\mathcal{D}} \left| \frac{g_\sigma(x-\phi_0)\frac{x-\phi_0}{\sigma^2} -g_\sigma(x-\phi)\frac{G'_\sigma(\phi_0)}{G_\sigma(\phi_0)}}{g_\sigma(x-\phi_0)+c G_\sigma(\phi_0)}\right|\\
    &\leq \max_{x\in\mathcal{D}} \frac{g_\sigma(x-\phi_0)\frac{|x-\phi_0|}{\sigma^2}}{g_\sigma(x-\phi_0)+c G_\sigma(\phi_0)}
    +
    \frac{|G'_\sigma(\phi_0)|}{G_\sigma(\phi_0)}\max_{x\in\mathcal{D}} \frac{g_\sigma(x-\phi)}{g_\sigma(x-\phi_0)+c G_\sigma(\phi_0)}.
\end{align}
The second term is negligible:
\begin{equation}
    \frac{|G'_\sigma(\phi_0)|}{G_\sigma(\phi_0)}\max_{x\in\mathcal{D}} \frac{g_\sigma(x-\phi)}{g_\sigma(x-\phi_0)+c G_\sigma(\phi_0)} \leq \frac{|G'_\sigma(\phi_0)|}{G_\sigma(\phi_0)} \leq \frac{\frac{1}{\sigma}g_1\left(\frac{|\mathcal{D}|}{6\sigma}\right)}{\frac{1}{2}\left(\erf(\frac{5}{6\sqrt{2}}\frac{|\mathcal{D}|}{\sigma})+\erf(\frac{1}{6\sqrt{2}}\frac{|\mathcal{D}|}{\sigma})\right)}.
\end{equation}
We write the first term as
\begin{align}
    \frac{1}{\sigma}\max_{x\in\mathcal{D}} \frac{g_\sigma(x-\phi_0)\frac{|x-\phi_0|}{\sigma}}{g_\sigma(x-\phi_0)+c G_\sigma(\phi_0)} &\leq
    % \frac{1}{\sigma}\max_{x\in\RR }\frac{g_\sigma(x-\phi_0)\frac{|x-\phi_0|}{\sigma}}{g_\sigma(x-\phi_0)+c G_\sigma(\phi_0)}\\
    % &= \frac{1}{\sigma}\max_{\xi\geq0 }\frac{g_\sigma(\xi \sigma)\xi}{g_\sigma(\xi\sigma)+c G_\sigma(\phi_0)}\\
    % &=
    % \frac{1}{\sigma}\max_{\xi\geq0 }\frac{\xi}{1+c G_\sigma(\phi_0)\sqrt{2\pi}\sigma e^{\xi^2/2}}\\
    % &=
    \frac{1}{\sigma}\max_{\xi\geq0 }f(\xi)\\
    f(\xi)&:=\frac{\xi}{1+A e^{\xi^2/2}}\\
    A&:= 2c\sigma \lesssim c\sigma \sqrt{2\pi} \frac{\erf(\frac{5}{\sqrt{2}})+\erf(\frac{1}{\sqrt{2}})}{2}\leq c\sigma \sqrt{2\pi} G_\sigma(\phi_0).
\end{align}
Let $\xi_* = \sqrt{2\log(1/A)}$ so that $Ae^{\xi_*^2/2}=1$.
For $\xi \leq \xi_*$, we have $f(\xi) \leq \xi \leq \xi_*$.
For $\xi > \xi_*$, we use $f(\xi) \leq A^{-1}\xi e^{-\xi^2/2}$.
For $\sigma$ small enough that $A \leq e^{-1/2}$ ($\sigma \leq 0.3c^{-1} \lesssim (2e^{1/2}c)^{-1}$), $\xi_* \geq 1$, and $A^{-1}\xi e^{-\xi^2/2} \leq A^{-1}\xi_* e^{-\xi_*^2/2} = \xi_*$.
Therefore $f(\xi) \leq \xi^*$, and we can bound the score as
\begin{align}
    \max_{x\in\mathcal{D}} \left| \partial_\phi \log Q_c(x|\phi)\right|_{\phi = \phi_0} &\leq \frac{1}{\sigma}\left(\sqrt{-2\log(2c\sigma)}+\frac{g_1\left(\frac{|\mathcal{D}|}{6\sigma}\right)}{\frac{1}{2}\left(\erf(\frac{5}{6\sqrt{2}}\frac{|\mathcal{D}|}{\sigma})+\erf(\frac{1}{6\sqrt{2}}\frac{|\mathcal{D}|}{\sigma})\right)}\right).
\end{align}
We take $\sigma \leq |\mathcal{D}|/6$ and $\sigma c \leq 0.3$, small enough that this is bounded by
\begin{align}
    \max_{x\in\mathcal{D}} \left| \partial_\phi \log Q_c(x|\phi)\right|_{\phi = \phi_0} &\leq \frac{1}{\sigma}\sqrt{-\log(2c\sigma)}.
\end{align}

Finally, combining the steps above we get
\begin{equation}
    \frac{\max_{x\in\mathcal{D}} \left| \partial_\phi \log Q_c(x|\phi)\right|_{\phi = \phi_0}}{\mathcal{I}_c} \leq \sigma G_\sigma(\phi_0) \times C_1 ({1 + C_2 c{|\mathcal{D}|}})\sqrt{\log(\frac{1}{2c\sigma})}G_\sigma(\phi_0)
\end{equation}
leading to
\begin{align}
    b &\leq2\frac{(1-a_0)}{a_0}\sigma^3 d^{-1} g_\sigma(d)\times C_1(1 + C_2 |\mathcal{D}| c)\sqrt{\log(\frac{1}{2c\sigma})} + O(\sigma^4 g_\sigma(d)^2),\\
    \epsilon^2 &\leq \|\alpha\|_1^2 \sigma^2 \times C_1^2(1 + C_2 |\mathcal{D}| c)^2\log(\frac{1}{2c\sigma})G^2_\sigma(\phi_0) + O(\sigma^2 g_\sigma(d)).
\end{align}
Finally, to simplify the bias bound we use the fact that $x\sqrt{\log(x^{-1})} \leq (2e)^{-1/2}$ and therefore 
\begin{equation}
    2C_1\sigma\sqrt{\log(\frac{1}{2c\sigma})}\leq 2C_1\frac{1}{2c}\frac{1}{\sqrt{2e}} \lesssim \frac{4.4}{c}.
\end{equation}

\end{document}